%% file: card_guessing.tex
\documentclass[11pt, a4paper]{article}
\usepackage[margin=1in]{geometry}
\usepackage{amsmath,amsthm,amssymb,amsfonts}
\usepackage{mathtools}
\usepackage{booktabs}
\usepackage{color}
\usepackage{environ}
\usepackage{epigraph}
\usepackage{etoolbox}

\usepackage[dvipsnames]{xcolor}
\usepackage{hyperref}

\usepackage{tikz}

\usetikzlibrary{decorations.pathreplacing}
\usetikzlibrary{calligraphy}
\usetikzlibrary{fit,shapes.misc}
\usetikzlibrary{math}
\usetikzlibrary{calc,patterns,angles,quotes}

\usepackage{pdfpages}
\usepackage[makeroom]{cancel}
\usepackage{mathtools}
\usepackage[all,cmtip]{xy}
\usepackage{fullpage}

\usepackage{bbm}
\usepackage{dsfont}
\usepackage{amsthm}
\usepackage{graphicx}
\usepackage{amssymb}
\usepackage{epstopdf}
\usepackage{amsthm}
\usepackage{xfrac}
\usepackage{tikz}
\usepackage{algorithm}
\usepackage[noend]{algpseudocode}
\usepackage{todonotes}
\usepackage{xspace}
\usepackage{caption}
\usepackage[nameinlink]{cleveref}
\usepackage{mdframed}
\usepackage{diagbox}

\usepackage[
backend=bibtex,
style=numeric,
citestyle=numeric,
maxbibnames=99,
]{biblatex}

\bibliography{card_guessing.bib}
\renewbibmacro{in:}{}

\DeclareFieldFormat[inbook]{citetitle}{#1}
\DeclareFieldFormat[inproceedings]{citetitle}{#1}

\DeclareFieldFormat{booktitle}{#1}
\DeclareFieldFormat{journaltitle}{#1}
\DeclareFieldFormat*{subtitle}{#1}
\DeclareFieldFormat*{title}{\mkbibitalic{#1}}

\usepackage{pifont}

\newcommand\marktopleft[1]{%
	\tikz[overlay,remember picture]
	\node (marker-#1-a) at (0,1.5ex) {};%
}
\newcommand\markbottomright[1]{%
	\tikz[overlay,remember picture]
	\node (marker-#1-b) at (0,0) {};%
	\tikz[overlay,remember picture,thick,dashed,inner sep=3pt]
	\node[draw,rounded rectangle,fit=(marker-#1-a.center) (marker-#1-b.center)] {};%
}

\newcommand{\cD}{ \mathcal {D}}

\newcommand{\cF}{ \mathcal {F}}
\newcommand{\cG}{ \mathcal {G}}
\newcommand{\cH}{ \mathcal {H}}

\newcommand{\cR}{ \mathcal {R}}
\newcommand{\cS}{ \mathcal {S}}

\newcommand{\cX}{ \mathcal {X}}

\newcommand{\fOne}{\mathbbm{1}}

\newcommand{\fE}{\mathbb{E}}

\newcommand{\fN}{\mathbb{N}}

\newcommand{\tn}{\textnormal}

\newcommand{\innprod}[2]{\langle #1 \,, #2\rangle}

\newcommand{\zo}{\{0,1\}}

\newtheorem{theorem}{Theorem}[subsection]
\newtheorem{proposition}[theorem]{Proposition}
\newtheorem{fact}[theorem]{Fact}
\newtheorem{lemma}[theorem]{Lemma}
\newtheorem{corollary}[theorem]{Corollary}
\newtheorem{definition}[theorem]{Definition}
\newtheorem{claim}[theorem]{Claim}

\numberwithin{equation}{section}

\label{my style}

\DeclarePairedDelimiterX\Set[2]{\lbrace}{\rbrace}%
 { #1 \,\delimsize| \,\mathopen{} #2 }

\newcommand{\expected}{\mathop{\fE}}

 \algnewcommand\algorithmicforeach{\textbf{for each}}
 \algdef{S}[FOR]{ForEach}[1]{\algorithmicforeach\ #1\ \algorithmicdo}

\algnewcommand\algorithmicparameter{\textbf{Parameter:}}
\algnewcommand\Parameter{\item[\algorithmicparameter]}

\algnewcommand\algorithmicinput{\textbf{Input:}}
\algnewcommand\Input{\item[\algorithmicinput]}

\algnewcommand\algorithmicoutput{\textbf{Output:}}
\algnewcommand\Output{\item[\algorithmicoutput]}

\newif\ifboldnumber

\algrenewcommand\alglinenumber[1]{%
	\footnotesize\ifboldnumber\bfseries\fi\global\boldnumberfalse#1:}

\newtheorem*{question}{Question}

\DeclareMathOperator*{\gameval}{val}

\crefname{ineq}{inequality}{inequalities}
\creflabelformat{ineq}{#2{\upshape(#1)}#3}
\crefname{equa}{equality}{equalities}
\creflabelformat{equa}{#2{\upshape(#1)}#3}
\crefname{appr}{approximation}{approximations}
\creflabelformat{appr}{#2{\upshape(#1)}#3}
\crefname{claim}{claim}{claims}
\creflabelformat{claim}{#2{\upshape(#1)}#3}

\begin{document}


\newcommand{\dealersymbol}{\cD}
\newcommand{\cardsymbol}{c}
\newcommand{\encset}{D}
\newcommand{\reasonableRV}{R}
\newcommand{\reasonableexpected}{r}
\newcommand{\limit}{u}
\newcommand{\memorystatesymbol}{M}
\newcommand{\turnssetsymbol}{T}
\newcommand{\correctvector}{V}

\newcommand{\epochix}{i}
\newcommand{\turnix}{t}

\newcommand{\guessersymbol}{\cG}
\newcommand{\guessermemorysymbol}{m}
\newcommand{\guesssymbol}{g}
\newcommand{\permutation}{\pi}
\newcommand{\guesserrandomness}{\gamma}
\newcommand{\dealerrandomness}{\Delta}

\newcommand{\ampliparam}{\delta}

\newcommand{\epochsnum}{d} 

\newcommand{\mttbdealerenc}{$(\encset, \permutation, i)$-\mttbdealer}
\newcommand{\mttbdealerdec}{$(\encodedexplicit, \permutation, i)$-\mttbdealer}


\newcommand{\firstepochbegin}{k_1}
\newcommand{\firstepochbeginterm}{{n \over 8e \log n}}
\newcommand{\currentepochbegin}{k_i}
\newcommand{\lengthterm}{\guessermemorysymbol \cdot \log n}
\newcommand{\currentlength}{\ell}
\newcommand{\currentlimit}{\limit}
\newcommand{\currentulength}{\ell_\epochix}
\newcommand{\currentulimit}{\limit_\epochix}
\newcommand{\currentreasonableRV}{\reasonableRV_\epochix}
\newcommand{\currentreasonableexpected}{\reasonableexpected_\epochix}
\newcommand{\currentcorrectexpected}{\correctsexpected_\epochix}
\newcommand{\uglyterm}{n -\firstepochbegin + \currentepochbegin}

\newcommand{\limitbound}{\min\left\{k - \ell, \ell\right\}}
\newcommand{\currentlimitbound}{\min\left\{\currentepochbegin - \currentlength, {1 \over 2}\currentlength\right\}}

\newcommand{\permutationsfull}{\{\pi_t\}_{t=1}^n}
\newcommand{\permutations}{\{\pi_t\}}

\newcommand{\orderedset}{B}
\newcommand{\orderedsetdec}{B'}
\newcommand{\orderedsetexplicit}{B_1}
\newcommand{\orderedsetrecovered}{B_2}
\newcommand{\orderedsetsize}{k}
\newcommand{\orderedsetsizeparam}{\beta}
\newcommand{\correctsRV}{C}
\newcommand{\correctsexpected}{c}
\newcommand{\correctsnum}{\alpha}
\newcommand{\correctsdiff}{\delta}

\newcommand{\dealerdet}{$\dealersymbol_{\guessermemorysymbol,(\encset,\permutations)}$\xspace}
\newcommand{\dealerdetdec}{$\dealersymbol^\ast_{\guessermemorysymbol,(\encodedexplicit,\permutations, \turnssetsymbol, \correctvector)}$\xspace}
\newcommand{\dealerstaticvanilla}{$\dealersymbol_{\permutation}$\xspace}
\newcommand{\dealerstatic}{$\dealersymbol_{\permutation_{\orderedset,\guesserrandomness}}$\xspace}
\newcommand{\dealerstaticdec}{$\dealersymbol^\ast_{\orderedsetexplicit}$\xspace}

\newcommand{\encsetdomain}{{[n] \choose \firstepochbegin}}
\newcommand{\encsetdesc}{\encset \in \encsetdomain}
\newcommand{\decset}{\encset'}
\newcommand{\encodedexplicit}{\encset_1}
\newcommand{\encodedrecovered}{\encset_2}

\newcommand{\setdomain}{{[n] \choose k}}


\newcommand{\encstate}{2^\guessermemorysymbol}
\newcommand{\encperm}{n!}
\newcommand{\encexplicit}{{n \choose k-r_i}}
\newcommand{\encreasonable}{{k'_i \choose r_i}}

\newcommand{\encmttb}{\encstate \cdot \encexplicit \cdot \encreasonable}
\newcommand{\encmttbcost}{\log \left( \encmttb \right)}

\newcommand{\encdomainsymbol}{\cX}

\newcommand{\guesser}{$\guessersymbol$\xspace}
\newcommand{\guesserprime}{$\guessersymbol'$\xspace}
\newcommand{\guesserfixed}{$\guessersymbol_{(\guesserrandomness)}$\xspace}
\newcommand{\guessCard}{g_{\text{card}}}
\newcommand{\guessTransition}{g_{\text{state}}}

\newcommand{\dealer}{$\dealersymbol$\xspace}
\newcommand{\dealerm}{$\dealersymbol_\guessermemorysymbol$\xspace}
\newcommand{\dealeruni}{$\dealersymbol_\tn{universal}$\xspace}

\newcommand{\cardprime}{\cardsymbol'}
\newcommand{\guessprime}{\guesssymbol'}

\newcommand{\gametranscriptsequence}{\left\{ (\cardsymbol_i, \guesssymbol_i) \right\}}
\newcommand{\gametranscriptsequenceprime}{\left\{ (\cardprime_i, \guessprime_i) \right\}}
\newcommand{\gametranscript}{\gametranscriptsequence_{i=1}^n}

\newcommand{\gametranscriptprefix}{\gametranscriptsequence_{i=1}^t}

\newcommand{\epochsterm}{\{k_i\}}
\newcommand{\epochsgeometric}{\epochsterm_{i=1}^{\epochsnum}}
\newcommand{\epochllimits}{\{L_i\}_{i=1}^{\epochsnum-1}}

\newcommand{\reasonablesnum}{\alpha}
\newcommand{\reasonablesdiff}{\beta}

\newcommand{\reasonablesboundtermi}{{8 \cdot e \cdot \firstepochbegin \cdot \currentlength \over n}}
\newcommand{\reasonablesbound}{\max\left\{{4 \cdot e \cdot \firstepochbegin \cdot \ell \over n}, m\right\}}
\newcommand{\reasonablesboundi}{\max\left\{\reasonablesboundtermi, m\right\}}
\newcommand{\reasonablesbounduniterm}{\max\left\{{8 \cdot e \cdot \firstepochbegin \cdot \currentulength \over n}, m\right\}}
\newcommand{\reasonablesbounduni}{{\currentulength \over \log^2 n}}

\newcommand{\permutationset}{\cS_n}
\newcommand{\permutationdesc}{\permutation \in \permutationset}
\newcommand{\minorderphrase}{min-order\xspace}
\newcommand{\minorderpi}{$\pi$-\minorderphrase}
\newcommand{\minorderpit}{$\pi_\turnix$-\minorderphrase}

\newcommand{\mttbphrase}{Move-to-the-Back\xspace}
\newcommand{\mttbdealer}{\mttbphrase Dealer\xspace}

\newcommand{\universaldealer}{universal Dealer\xspace}

\newcommand{\memorybound}{{n \over \log^2 n}}
\newcommand{\result}{\ln m + 2 \ln \log n + O(1)}
\newcommand{\resultRandomSubsetSpace}{\log^2 n - \log n + 2}
\newcommand{\resultUniversal}{(1 + o(1)) \cdot \ln \guessermemorysymbol + 8 \ln \log n + O(1)}

\newcommand{\resultUniversaln}{(1 + o(1)) \cdot \ln n + 8 \ln \log n + O(1)}
\newcommand{\resultUniversalAsymptotic}{(1 + o(1)) \cdot \ln \guessermemorysymbol + O(\log \log n)}

\newcommand{\subsetrange}[1]{[1 \text{ --- } #1]}
 	
\newcommand{\encodefunc}{\tn{Encode}}
\newcommand{\decodefunc}{f_\text{decode}}

\newcommand{\encodefuncupper}{\tn{EncodeO}_{\guesserrandomness, \correctsnum}}
\newcommand{\decodefuncupper}{\tn{DecodeO}_{\guesserrandomness, \correctsnum}}

\newcommand{\encodefuncadaptive}{\tn{EncodeU}_{\guesserrandomness, \permutations, \reasonablesnum, i}}
\newcommand{\decodefuncadaptive}{\tn{DecodeU}_{\guesserrandomness, \permutations, \reasonablesnum, i}}
	
\setlength\epigraphwidth{.7\textwidth}
\setlength\epigraphrule{0pt}
\renewcommand{\footnotesize}{\fontsize{9pt}{11pt}\selectfont}


\label{title}
\title{Keep That Card in Mind: \\ [0.2 em]Card Guessing with Limited Memory\thanks{Research supported in part by grants from the Israel Science Foundation (no.\ 950/15 and 2686/20),  by the Simons Foundation Collaboration on the Theory of Algorithmic Fairness and by the Israeli Council for Higher Education (CHE) via the Weizmann Data Science Research Center}
}
\date{}

\author{Boaz Menuhin\thanks{
		Department of Computer Science and Applied Mathematics, Weizmann Institute of Science, Rehovot, Israel.
		Email: \href{mailto:boaz.menuhin@weizmann.ac.il}{boaz.menuhin@weizmann.ac.il}.
	}
	\and
	Moni Naor\thanks{
		Department of Computer Science and Applied Mathematics, Weizmann Institute of Science, Rehovot, Israel.
		Incumbent of the Judith Kleeman Professorial Chair.
		Email: \href{mailto:moni.naor@weizmann.ac.il}{moni.naor@weizmann.ac.il}.}
}
\maketitle

\begin{abstract}

A card guessing game is played between two players, Guesser and Dealer.
At the beginning of the game, the Dealer holds a deck of $n$ cards (labeled $1, ..., n$).
For~$n$ turns, the Dealer draws a card from the deck, the Guesser guesses which card was drawn, and then the card is discarded from the deck.
The Guesser receives a point for each correctly guessed card.

With perfect memory, a Guesser can keep track of all cards that were played so far and pick at random a card that has not appeared so far, yielding in expectation $\ln n$ correct guesses, regardless of how the Dealer arranges the deck.
With no memory, the best a Guesser can do will result in a single guess in expectation.

We consider the case of a memory bounded Guesser that has $m < n$ memory bits.
We show that the performance of such a memory bounded Guesser depends much on the behavior of the Dealer.
In more detail, we show that there is a gap between the static case, where the Dealer draws cards from a properly shuffled deck or a prearranged one, and the adaptive case, where the Dealer draws cards thoughtfully, in an adversarial manner.
Specifically:

\begin{enumerate}
	\item
		We show a Guesser with $O(\log^2 n)$ memory bits that scores a near optimal result against any static Dealer.
	\item
		We show that no Guesser with $m$ bits of memory can score better than $O(\sqrt{m})$ correct guesses against a random Dealer, thus, no Guesser can score better than $\min \{\sqrt{m}, \ln n\}$, i.e., the above Guesser is optimal.
	\item
		We show an efficient adaptive Dealer against which no Guesser with~$\guessermemorysymbol$ memory bits can make more than $\result$ correct guesses in expectation.
\end{enumerate}
These results are (almost) tight, and we prove them using compression arguments that harness the guessing strategy for encoding.

\end{abstract}


\newpage
\tableofcontents
\newpage

\section{Introduction}

\begin{epigraphs}
	\qitem{\itshape ``Those who cannot remember the past are condemned to repeat it''}{---George Santayana, \textit{The Life of Reason}, 1905~\cite{S1905}}
	\qitem{\itshape Even if you use randomness and cryptography!}{}
\end{epigraphs}

There are $n$ cards in a deck.
In each turn, one player, called Dealer, selects a card and a second player, called Guesser, guesses which card was drawn.
The selected cards cannot be drawn again.
The quantity of interest is how many cards were guessed correctly.

A Guesser with perfect memory can guess a card that has not appeared yet.
When there are $i$ cards left, the probability of correctly guessing is $1 \over i$.
By linearity of expectation, such a Guesser is expected to guess correctly about $\ln n$ times.
On the other hand, at any point in time, a memoryless Guesser cannot guess with probability better than $1 /n$, resulting in $1$ correct guess on expectation\footnote{These examples are taken from a textbook on algorithms by Kleinberg and Tardos~\cite{KT2006}, Chapter 13, Pages 721-722.}.
We are interested in the case where the Guesser has $m < n$ bits of memory.

One can think about card guessing as a streaming problem where the algorithm predicts the next element in a stream, under the promise that all elements in the stream are unique.
In the streaming model, a large sequence of elements is presented to an algorithm, usually one element at a time.
The algorithm, which cannot store the entire input, keeps the needed information and outputs some function on the stream seen so far.
As it may be impossible to output the exact value of the function without storing the entire input, it is a typical relaxation to consider an approximate value of the function as output.
In this sense, prediction is a form of approximation\footnote{We elaborate on prediction as a form of  approximation in \Cref{sec_discussion:sub_streaming_algorithms_and_data_structures}.}.

Some streaming problems are solvable by a deterministic algorithm, while others require a randomized one (for a survey on streaming algorithms, see~\cite{Muthukrishnan2005}).
It is common to analyze the performance of an algorithm with respect to a worst-case stream that is chosen ahead of time and is fixed throughout the execution of the algorithm.
We call such a stream \emph{oblivious}, or \emph{static}.
A recent line of works (see \Cref{subsec_related_work}) focuses on analyzing the performance of an algorithm as if an adversary looks at the algorithm's output in every turn and thoughtfully chooses the next element in order to make the algorithm to fail.
An algorithm that performs well against such an adversary is called \emph{adversarially robust}.

In this paper, we pinpoint the complexity of card guessing in different environments, especially with respect to the memory requirements of the Guesser.

\subsection{Our Results}
We study the case where the Guesser has bounded memory, and we ask how well a Guesser can perform?
It turns out that the performance of a memory bounded Guesser is highly sensitive to the behavior of the Dealer.

\begin{itemize}
	\item
		A Dealer is called \textbf{``random shuffle''} if every ordering of the deck has equal probability.
		This is equivalent to drawing a card at random in every turn.
\item
		A \textbf{static} Dealer draws cards one by one from a prearranged deck in some specific order and may choose the worst order.
	\item
		An \textbf{adaptive} Dealer is allowed to change the order of the deck throughout the game.

\end{itemize}
Clearly these dealers  are presented by increasing power.  Our results are as follows:
\begin{enumerate}
	\item
		Against the random-shuffle Dealer, there exists a Guesser with $\log^2 n + \log n$ memory bits that  makes at least~$1/2 \log n$ correct guesses in expectation, i.e. asymptotically similar to a Guesser with perfect memory.
		This Guesser can be amplified at the cost of $\log n$ factor to get closer to $\ln n$.
	\item
		There exists a Guesser with $\log^2 n - \log n + 2$ bits of memory and $2 \log n$ random bits that against any static Dealer   scores $1/4 \ln n - O(1)$ correct guesses in expectation.
		This Guesser can be amplified at the cost of $\log n$ factor to get closer to $\ln n$.
	\item
		The above Guessers are optimal:
		Every Guesser with $\guessermemorysymbol$ bits of memory can score at most $O(\sqrt{\guessermemorysymbol})$ correct guesses in expectation against the random-shuffle Dealer, regardless of the amount of randomness that the Guesser uses.
	\item
		For every $\guessermemorysymbol$ there exists a computationally efficient adaptive Dealer against which no Guesser with $\guessermemorysymbol$ memory bits can make more than $\result$ correct guesses in expectation, regardless of how much randomness and what cryptographic tools and assumptions the Guesser uses.
	\item
		Furthermore, there exists a computationally efficient adaptive \emph{universal} Dealer, i.e., that makes no assumption on the amount of memory of the Guesser, against which every Guesser with $\guessermemorysymbol$ bits of memory is expected to make at most $\resultUniversal$ correct guesses.
\end{enumerate}
See \Cref{table:results} for a comparison of our results.

		\begin{table}[h]
			\caption{Partial list of results: constructive results above the line, impossibility results below.}
			\label{table:results}
			\begin{tabular}{@{}llllll@{}}
				\midrule
				Guessing Technique & Memory & \begin{tabular}{@{}c@{}}Random\\ bits\end{tabular} & Dealer       & Score     &  \\
				\toprule
				Subset Guesser      & $\guessermemorysymbol$ & - & Any        & $\ln m$      &   \\
				Remember last cards & $\guessermemorysymbol$ & - & Any     & $\ln {m \over \log n}$ &   \\
				Subset+Remember last & $\guessermemorysymbol \le \sqrt n$ & - & Random & $2\ln m - \ln \log n - \ln 2$ & \\
				Following Subsets   & $O(\log^2 n)$ & -  & Random & $1/2 \log n$ & \\
				Randomized Subsets   & $O(\log^2 n)$ & $2 \log n$     & Static & $1/4 \ln n$ & \\
				\midrule
				Any Guesser      & $\guessermemorysymbol$ & $\infty$ & Random & $O(\min\{\ln n, \sqrt{m}\})$  &  \\
				Any Guesser      & $\guessermemorysymbol$ & $\infty$ & Adaptive & $\result$      &   \\
				Any Guesser      & $\guessermemorysymbol$ & $\infty$ & \begin{tabular}{@{}c@{}}Adaptive-\\ universal\end{tabular} & $\resultUniversalAsymptotic$      &   \\
				\bottomrule
			\end{tabular}
	\end{table}
To summarize, the main lesson from these results is the significant impact of adaptivity of the dealer, more than any other factor.

\subsection{Our Approaches and Techniques}

Our results separate the required amount of memory and randomness that the Guesser requires for playing against the three types of Dealers.

Quite surprisingly, we show that against the random-shuffle Dealer, a Guesser with very limited memory, and no randomness at all, can perform similarly to a Guesser with perfect memory.
More formally, the main result of \Cref{sec_static_dealer:sec_following_subset_guesser} is:
\begin{theorem}
	There exists a Guesser with $\log^2 n + \log n$ memory bits and no randomness that scores $1/2 \log n$ correct guesses in expectation when playing against the random-shuffle Dealer.
\end{theorem}
Our Guesser tracks which cards were drawn from multiple subsets of cards.
Using at most $2 \log n$ bits per subset, the Guesser can recover the last card that has not appeared from each subset and guess it.
Repeating that guess over and over, the Guesser can guarantee a single correct guess from each subset.

However, the last cards from the different subsets may be indistinct.
The subsets are built incrementally, i.e., the $i$th subset is contained in the $(i+1)$-th subset, as visualized in \Cref{figure:following_subset_illustration}.
It follows that the probability for the last card from each subset to appear before the last card from the following subset is fixed (and is at least $1/2$).
So we get that the expected number of correct guesses is proportional to the number of subsets tracked by the Guesser and the ratio between two following subsets.

With more space we can have denser subsets, getting that for $0<\ampliparam \le 1$ a Guesser with $\log_{1+\ampliparam} 2 \cdot \log^2 n + \log n$ can score ${\ampliparam \over (1 + \ampliparam) \ln(1+\ampliparam)} \ln n$ correct guesses in expectation, when playing against the random-shuffle Dealer.
For $\ampliparam = 1$ this is the above theorem.
The number of correct guesses goes to $\ln n$ as $\ampliparam$ goes to zero.

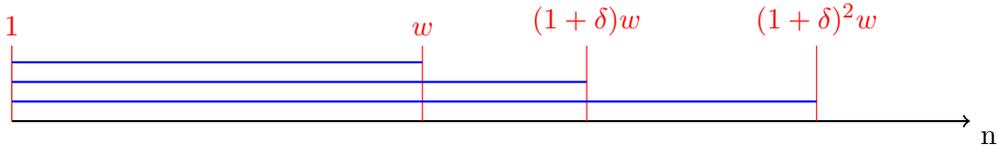
\begin{figure}[h]
	\begin{mdframed}	
			\input{parts/figures/following_subsets.tex}
		\caption{Following-Subsets. The blue lines represents subsets of $[n]$, the top blue line the subset $\{1, \dots, w\}$, the one in the middle $\{1, \dots, (1+\ampliparam) w\}$ and so on. In the basic construction, $(1+\ampliparam) = 2$.}
		\label{figure:following_subset_illustration}
	\end{mdframed}
\end{figure}

Consider a static Dealer, one that fixes the sequence of cards ahead of time but chooses the worst arrangement.
A simple adversarial argument shows that for every Guesser that uses no randomness, there exists an arrangement of the deck against which that deterministic Guesser scores at most $1$ correct guess (guessing a single card is inevitable, even by a memoryless Guesser).
In \Cref{sec_static_dealer:sec_random_subsets_guesser} we show that $2 \log n$ random bits suffice for a Guesser with $O(\log^2 n)$ bits of memory to score near perfect results when playing against any static Dealer.
\begin{theorem}
	\label{thm_random_subset_guesser}
	There exists a Guesser with $\resultRandomSubsetSpace$ bits of memory and $2\log n$ random bits that scores ${1 \over 4} \ln n$ correct guesses in expectation when playing against any static Dealer.
\end{theorem}
We use a pairwise independent permutation to split the cards to $\log n$ disjoint subsets of various sizes and track each subset similarly to the previous construction.
We show that in each turn, the Guesser recovers a correct guess from a certain subset in probability that is proportional to the number of cards left in the deck.
Namely, when $\turnix$ cards are left in the deck, the probability of a correct guess is at least $1/4t$, resulting in $1/4 \cdot \ln n$ correct guesses in expectation.

In \Cref{sec_static_dealer:sec_random_subsets_guesser:sec_amplification} we show that it is possible to amplify the above structure, and using $(\log n)$-wise independent functions for assigning cards to subsets, a Guesser can get closer to $\ln n$ correct guesses.
\begin{theorem}
	There exists a Guesser with $O(\log(1/\ampliparam) \log^2 n)$ bits of memory and with  $O(\log(1/\ampliparam) \log^2 n)$ bits of randomness that scores $(1-\ampliparam) \ln n$ correct guesses on expectation against any static Dealer.
\end{theorem}

Our deterministic and randomized Guessers are inspired by Garg and Schneider~\cite{GargS18} and Feige's~\cite{Feige19} algorithms for the first player in the Mirror Game in that they follow subsets of cards and track which member appeared.
However, it turns out that there are fundamental differences between strategies for Mirror Game and card guessing against an adaptive dealer.
We elaborate on the relation to the Mirror Game and to Feige's construction in \Cref{sec_discussion:subs_discussion_card_guessing_mirror_game}.

In \Cref{sec_static_dealer:upper_bound} we show that these Guessers are the best possible against the random-shuffle Dealer for $m \le \log^2 n$.
I.e., that there exists no Guessing technique that uses less memory and performs similarly.
\begin{theorem}
	\label{sec_introduction:thm_optimal_guessers}
	Any Guesser with $m$ memory bits can get \emph{at most} $O(\min\{\ln n, \sqrt{m}\})$ correct guesses in expectation when playing against the random-shuffle Dealer.
\end{theorem}
We show this by presenting an encoding scheme that utilizes correct guesses to encode an ordered set in an efficient manner.
Using a compression argument we show that $\log^2 n$ bits of memory are actually essential for getting $O( \ln n)$ correct guesses.

\paragraph{Proof by Compression: }
This is a quite general method (see below) to prove the success of an algorithm by showing that some events allow us to compress the random bits used.
Say a randomized algorithm can tolerate some number of bad events.
For some specific domain (e.g.\ ordered sets of some size), we introduce an encoding scheme that utilizes the occurrence of certain bad events in order to achieve a shorter description of elements in the domain.
We then consider the amount of bad events required to achieve a description that is shorter than the entropy of a random element in the domain.
We know that for any compression method, the probability of chopping off (saving) $w$ bits from a random string is at most $2^{-w}$.
We get that the probability for too many bad events is negligible.

The method has been applied in a variety of fields, for instance, to prove the success of the algorithm in the ``Algorithmic Lov\'{a}sz Local Lemma"~\cite{MT2009},
the success probability of Cuckoo Hashing~\cite{P12},
lower bounds on construction of cryptographic primitives~\cite{GT2000} and space-time trade-off for quantum algorithms~\cite{CLQ20}.
\newline
\newline
To prove \Cref{sec_introduction:thm_optimal_guessers}, we introduce an encoding scheme for \emph{ordered sets} that utilizes correct guesses to achieve a shorter description of the ordered set.
The idea is to simulate a game between a Guesser and a static Dealer, where the bottom of the deck is arranged according to the ordered set we wish to encode.
If sufficiently many correct guesses occurred, then the encoding function stores the necessary information required to reproduce the course of the game.

Namely: the memory state of the Guesser, the set of turns at which the Guesser predicted correctly and the cards that the Dealer draws in the other turns in their respective order.
By fixing the Guesser's randomness, every ordered set yields a single description by the encode function, and every description results in a single course of the game during decode.
This allows the decode function to reconstruct the ordered set.
A visual representation of the stored information is provided in \Cref{figure:low_memory_upper_bound_encoding}.

We get that we pay once for the memory state of the Guesser, and from that point on any correct guess shrinks the description of the ordered set.
We then show that making too many correct guesses implies compression, i.e., a description of expected length shorter than the entropy of a random element.
By doing so, we bound the expected number of correct guesses that any Guesser can make.
The result holds regardless of how much randomness and what cryptographic tools and assumptions are used by the Guesser.

\begin{figure}[h]
	\begin{mdframed}	
		\resizebox{\textwidth}{!}{%
			\input{parts/figures/correct_guess_encoding.tex}
		}
		\caption{Correct guesses encoding. Colored - information stored by the encoding scheme.}
		\label{figure:low_memory_upper_bound_encoding}
	\end{mdframed}
\end{figure}
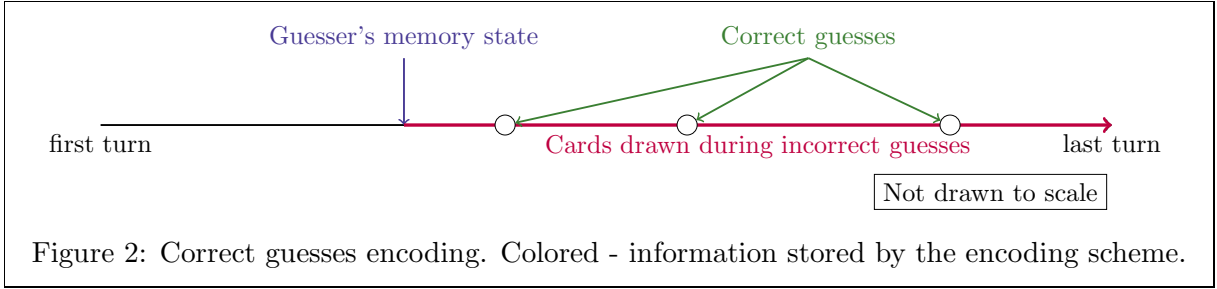

In \Cref{section_upper_bound}, we turn our attention to the adversarial adaptive Dealer.
We show that if the Dealer is allowed to be adaptive, then almost any advantage gained from sophisticated memory usage vanishes.
Furthermore, the adversarial dealer needs to know very little about the guesser.
We show two results:
\begin{theorem}
	For every $m$ there exists an efficient adaptive Dealer against which any Guesser with~$m$ bits of memory can score at most $\result$ correct guesses in expectation.
\end{theorem}
\begin{theorem}
	There exists a universal efficient adaptive Dealer against which any Guesser with~$m$ bits of memory can score at most $\resultUniversal$ correct guesses in expectation.
\end{theorem}

The two Dealers share the same general strategy that is parameterized differently.
The Dealer's strategy is simply to refrain from drawing recently guessed cards for some turns.
This result stands even if the Guesser is allowed to use unlimited randomness and cryptographic tools, while the Dealer is as simple as possible.
The computational efficiency and simplicity of the Dealer, as well as the fact that the result stands even against an all-powerful Guesser with randomness and cryptography, emphasize that it is the \emph{adaptivity that plays the key role}, rather than anything else.
Recall that by~\Cref{thm_random_subset_guesser}, a mild amount of randomness suffices to achieve near optimal results against any static Dealer.
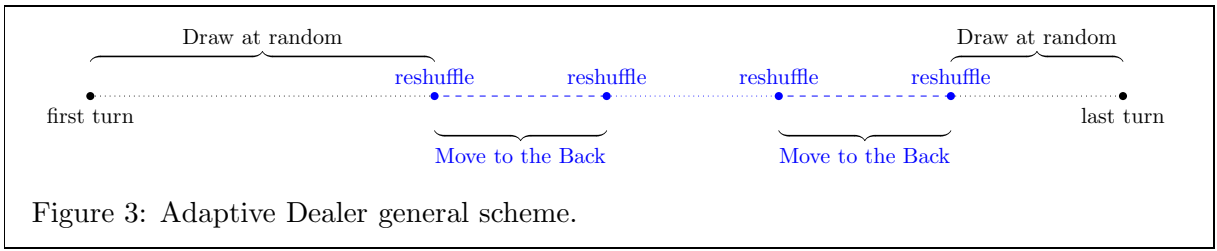
\begin{figure}[h]
	\begin{mdframed}	
		\resizebox{\textwidth}{!}{%
			\input{parts/figures/adaptive_dealer_general_scheme.tex}
		}
		\caption{Adaptive Dealer general scheme.}
		\label{figure:mttb_general_scheme}
	\end{mdframed}
\end{figure}

In more detail, our Dealer shuffles the deck at the beginning of the game and draws cards one by one.
At some point, the Dealer begins to refrain from selecting cards guessed by the Guesser.
When the Guesser guesses a card that resides in the deck, the Dealer takes that card and moves it to the back of the deck.
Moving cards to the back reduces the number of cards that the Dealer may select, and as a result, makes the Dealer predictable.
Therefore, at some point, the Dealer reshuffles the deck, making all cards available for drawing again, and repeats this behavior.
Towards the end of the game, the Dealer shuffles the deck one last time and draws cards at random.
The period of turns between reshuffles is called an epoch, and the strategy is called MtBE-strategy (which stands for Move to the Back Epoch strategy).

Call a guess \emph{reasonable} if it is a card that can be drawn, that is, a possibly correct guess.
To show that moving cards to the back is an effective strategy, we show that no memory bounded Guesser is expected to make too many \emph{reasonable} guesses during each epoch.
The idea is that if the Guesser can produce many reasonable guesses, then she knows something about the set of cards in the deck and can be used to describe it efficiently.
We show an encoding scheme for \emph{sets}, and using compression argument, we bound the expected amount of reasonable guesses in every epoch.

Since our Dealer is adaptive, it is difficult to predict the order in which cards will be drawn, thus we do not use the same encoding scheme for \emph{ordered} sets.
Instead, we encode \emph{unordered} sets using a similar encoding scheme.

As in \Cref{sec_introduction:thm_optimal_guessers}, the encoding scheme works by simulating a game between the Guesser and the Dealer, where the Dealer keeps the set of cards we wish to encode for the end.
If sufficiently many reasonable guesses occurred during an epoch, we use that epoch to achieve a shorter description of the set.
In detail, we store the Guesser's memory state, the cards drawn by the Dealer while the Guesser guessed incorrectly, the set of turns where the Guesser guessed reasonably, and which of these guesses were actually correct.
These objects allow us to simulate the same game during decoding and thus to recover the set.
Visualization of the encoded information provided in \Cref{figure:adaptive_dealer_upper_bound_encoding}.

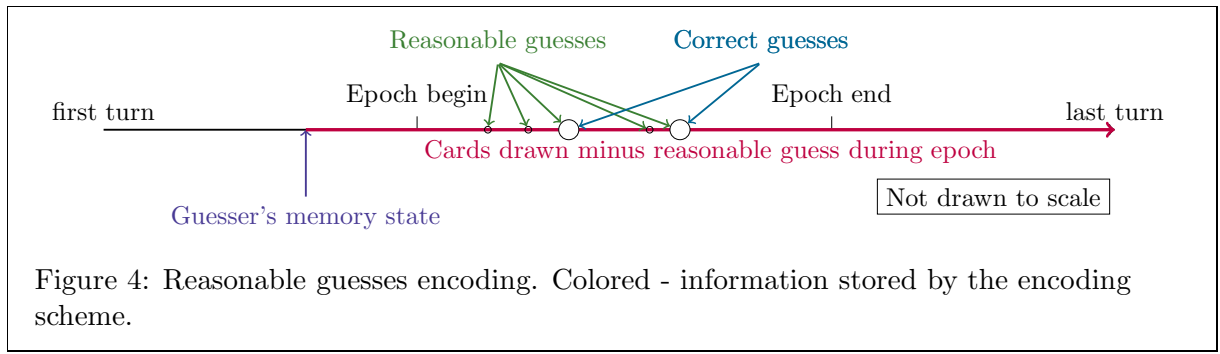
\begin{figure}[h]
	\begin{mdframed}	
		\resizebox{\textwidth}{!}{%
			\input{parts/figures/reasonable_guess_encoding.tex}
		}
		\caption{Reasonable guesses encoding. Colored - information stored by the encoding scheme.}
		\label{figure:adaptive_dealer_upper_bound_encoding}
	\end{mdframed}
\end{figure}

Since reasonable guesses allow us to achieve shorter descriptions, we get a bound on the expected number of reasonable guesses during each epoch.
By doing so, we bound the expected number of \emph{correct} guesses in each epoch, since the cards are drawn from a large enough set.
At this point, the analysis's calculations of the two Dealers vary and we analyze them differently.

Lastly, since a Guesser with $\guessermemorysymbol$ bits of memory can achieve $\ln \guessermemorysymbol$ correct guesses in expectation, these results are almost tight.

\subsection{Related Work}
\label{subsec_related_work}
\paragraph{Card Guessing.}
An early work concerning card guessing dates back to 1924~\Cite{Fisher1924}, when Ronald Fisher studied the game in the context of analyzing claims of \emph{psychic ability}.
Fisher suggested and analyzed a method to measure and determine a Guesser's claim to have supernatural abilities (namely clairvoyance and telepathy) by assigning scores to the Guesser's guesses and see how much they deviates from the expectation.
In 1981, Diaconis and Graham~\cite{DG1981} studied the case where there are $c_i$ copies of the $i$th card, and determined the optimal and worst strategies for some cases.
They considered the cases of no feedback at all, partial feedback (was the guess correct or not) for $c_i = 1$, and full feedback (which card was drawn).
Recently, Diaconis, Graham, He and Spiro~\cite{DGHS2020} asymptotically determined the expected score that an optimal strategy achieves for the case of partial feedback where $c_i \ge 2$.

A more useful yet equally dubious purpose is ``Card Counting'' in gambling (see Wikipedia entry~\cite{wiki:Card_counting}).
In 1962 Edward Thorp, a Professor of Mathematics, published a bestseller book~\cite{Thorp1962} about winning strategies in the game of Black-Jack.
The book covered analyses of the game from the viewpoint of the player and the Casino, as well ``low-memory'' strategies that increase the player's expected benefit.
The idea behind card counting in this context is that by knowing the distribution of the next card(s), one can evaluate their own hand better and thus bet accordingly.
Card counting has been applied to other card games, such as Bridge and Texas hold ’em Poker.
In these games, evaluating the probabilities for the upcoming cards is considered essential.
In the context of this paper, card counting is an applied ``low-memory'' card guessing technique that utilizes the structure of specific games.

\paragraph{Mirror Games.}
Garg and Schneider~\cite{GargS18} introduced the Mirror Game, a game closely related to card guessing.
In Mirror Game, there are two players, Alice and Bob, taking turns in saying numbers from $[n]$.
In every turn, a player says a number that was not said before by either player.
If a player says a number that was already declared, that player loses, and the other player is the winner.
If there are no more numbers to say, then it is a draw.
Alice is first.
Bob always has a simple deterministic winning strategy that requires only $\log n$ memory bits.
When Alice says $x$, bob says $n+1 -x$, and hence the name's origin.

Garg and Schneider showed that every deterministic winning (drawing) strategy for Alice requires~$\Omega(n)$ bits of memory.
They have also presented a randomized strategy for Alice that with high probability end with at least a draw for Alice that requires $O(\sqrt{n})$ bits of storage.
Their strategy relied on access to a secret (from Bob) random matching.
Using the same settings, Feige~\cite{Feige19} showed that $O(\log^3 n)$ bits of memory suffice.
Our Guessers for the static case (\Cref{sec_static_dealer:sec_following_subset_guesser,sec_static_dealer:sec_random_subsets_guesser}) are inspired by Feige's construction.

Additional discussion on the relation and differences between card guessing and Mirror Game is provided in \Cref{sec_discussion:subs_discussion_card_guessing_mirror_game}.

\paragraph{Adversarial streams and sampling.}
A streaming algorithm is called \emph{adversarially robust} if it's performance are guaranteed to hold even when the elements in the stream are chosen adaptively in an adversarial manner.
The question concerning the gap in memory consumption between the static case and the adversarially adaptive case has been the subject of recent line of works.

On the positive side, Ben-Eliezer, Jayaram, Woodruff and Yogev~\cite{BJWY2020} showed general transformations for a family of tasks, for turning a streaming algorithm to be adversarially robust, with some overhead.
Woodruff and Zhou~\cite{WZ2020} suggested another set of transformation for the same family of problems.
A different approach was taken by Hassidim, Kaplan, Mansour, Matias and Stemmer~\cite{HKMMS2020} who showed that it may be possible to get an even smaller overhead in some cases by using differential privacy as a protection against an adaptive adversary.

On the negative side, Hard and Woodruff~\cite{HW2013} showed that linear sketches are inherently not adversarially robust.
They showed it for the task of approximating $L^P$ norms but their technique stands for other tasks as well.
In a recent result, Kaplan, Mansour, Nissim and Stemmer~\cite{KMNS2021} showed a problem that requires polylogarithmic amount of memory in the static case but \emph{any} adversarially robust algorithm for it requires exponentially larger memory.
Our work join that of~\cite{KMNS2021} by showing a simpler, even more natural, streaming problem that separates adversarial streams from oblivious streams.

In a similar vein, given a large enough sample from some population, then we know that the measure of any fixed sub-population is well-estimated by its frequency in the sample.
The size of the sample needed is the VC dimension of the set system of the sub-populations of interest.
Ben-Eliezer and Yogev~\cite{BY2020} showed that when sampling from a stream, if the sample is public and an adversary may choose the stream based on the samples so far, then the VC Dimension may not be enough.
Alon, Ben-Eliezer, Dagan, Moran, Naor and Yogev~\cite{ABDMNY2021} showed that the Littlestone dimension, which might be much larger than the VC dimension, captures the size of the sample needed in this case.

\paragraph{Online Computation and Competitive Analysis.}
Another area where the exact power of the adversary comes up is in competitive analysis of online algorithms (see Borodin and El-Yaniv~\cite{BorodinE98}).
Here there are various types of adversaries, distinguished by whether they are adaptive or static and whether they decide on the movements of the competing algorithm in an online manner or an offline one. The result is a hierarchy of oblivious, adaptive online and adaptive offline adversaries.
It turns out that in request-answer games (a very general form capturing issues like paging), an algorithm competitive against the adaptive offline adversary may be transformed into a deterministic one with similar competitive ratio~\cite{Ben-DavidBorodinKarpTardosWigderson1994}.
We do not see a similar phenomenon in our setting, where one should recall
that a deterministic algorithm is hopeless against a static adversary.

\paragraph{Distinguishing Permutations and Functions.}

A stream of $q$ random elements from the domain $[n]$ is given to a memory bounded algorithm that attempts to determine whether the stream was sampled with or without repetitions.
When the stream ends, the algorithm outputs its determination, and is measured by it's ability to judge better than guessing at random.
If $q = \Theta(\sqrt{n})$,  then by the birthday paradox, a repetition occurs with high probability.
The algorithm uses $O(q\log n)$ memory bits to recognize this repetition.

Motivated by the fact that the task of distinguishing between random permutations and random functions has significant cryptographic implications, Jaeger and Tessaro~\cite{JaegerTessaro2019} introduced the above problem and showed a conditional bound on the advantage of the algorithm.
In particular, they showed that under an unproved combinatorial conjecture the advantage of an algorithm with $m$ bits of memory is bounded by $\sqrt{q \cdot m / n}$.
Dinur~\cite{Dinur2020} showed an unconditional upper bound on the advantage of ${\log q \cdot q \cdot m / n}$.
This was followed by the work of Shahaf, Ordentlich and Segev~\cite{ShahafOrdentlichSegev2020} who achieved the unconditional upper bound on the advantage of $\sqrt{q \cdot m / n}$.

\section{Preliminaries}

Throughout this paper we use $[n]$ and $[1 - n]$ to denote the set of integers~$\{1, \dots ,n\}$.
We denote the collection of subsets of~$[n]$ of size exactly $k$ by ${[n] \choose k} = \{B \subseteq [n]\colon |B| = k\}$.
We denote the set of permutations on $n$ elements by $\permutationset$.
All logs are base 2 unless explicitly stated otherwise, $\ln$ is the natural logarithm (base $e$).
We denote the set of binary strings of length $\ell$ by $\zo^\ell$.
We denote the set of binary strings of any finite length by $\zo^\ast = \cup_{i \in \fN} \zo^i$.
\subsection{Information Theory}

\begin{definition}[Entropy]
	Given a discrete random variable $X$ that takes values from domain~$\cX$ with probability mass function $p(x) = \Pr[X = x]$.
	The Binary Entropy (abbreviated Entropy) of $X$, denoted $H(X)$ is
	\[
		H(X) = - \sum_{x \in \cX} p(x) \cdot \log p(x).
	\]
\end{definition}

\begin{fact}
	The entropy of a random variable $X$ drawn uniformly at random from domain~$\cX$ is~$\log |\cX|$, i.e. \[
	H(X) = \log |\cX|.\]
\end{fact}

\begin{definition}[Prefix-free code]
	\label{def_prefix_free_code}
	A set of code-words $C \subseteq \zo^\ast$ is prefix-free if no code-word $c \in C$ is a prefix of another code-word $c' \in C$.
\end{definition}

\begin{proposition}[Theorem 5.3.1 in~\cite{CT06}]
	\label{thm_compression_larger_than_entropy}
	The expected length $L$ of any prefix-free binary code for a random variable $X$ is greater than or equal to the binary entropy $H(X)$.
\end{proposition}

\begin{lemma}
	\label{lemma_compresion_by_d_bits}
	Given a random variable $X$ uniformly drawn from domain $\cX$.
	For every encoding function $\encodefunc$ for $X$, The probability that the encoding of $X$ is $d$ bits less than it's entropy is at most $2^{-d}$, i.e.\
	\[
		\Pr_{X \in \cX}\left[|\encodefunc(X)| = H(X) - d\right] \le 2^{-d}.
	\]
\end{lemma}
\begin{proof}
	There are $2^{H(X)-d}$ possible descriptions of length $H(X)- d$.
	Thus, there are at most that many values in $\cX$ for which a description of such length is produced.
	Therefor, the probability to draw one of them is at most $2^{H(X) -d} \over |\cX|$.
	The entropy of random variable $X$ uniformly drawn from $\cX$ is $\log(|\cX|)$. Therefore
	\begin{flalign*}
		\Pr\left[|\encodefunc(X)| = H(X) - d\right]
		\le
		{2^{H(X) -d} \over |\cX|}
		=
		{2^{\log|\cX| -d} \over |\cX|}
		=
		{|\cX|\cdot 2^{-d} \over |\cX|}
		=
		2^{-d}
		.
	\end{flalign*}
\end{proof}

\subsection{Families of $k$-wise Independent Functions}
\newcommand{\pairwisefuncfamily}{\cH_{pair}}
\newcommand{\kwisefuncfamily}{\cH_k}
\newcommand{\kwisefunc}{h}
\newcommand{\randfuncfamily}{\cF_\tn{rand}}
\newcommand{\randfunc}{f}

\begin{definition}
	\label{def_k_wise_function}
	A family of functions $\kwisefuncfamily = \{h: [n] \to [n]\}$ is called $k$-wise independent, if for any sequence of $k$ different elements $x_1, \dots, x_k \in [n] $ and every $y_1, \dots, y_k \in [n]$ it holds that
	\[
		\Pr_{\kwisefunc \sim \kwisefuncfamily}\left[(h(x_1), h(x_2), \dots, h(x_k)) = (y_1, y_2, \dots, y_k) \right] = n^{-k}.
	\]
\end{definition}

There exists $k$-wise independent families of size $O(n^k)$.
In fact, the family of polynomials of degree $k-1$ over a finite field is a $k$-wise independent family.
This family is easy to work with as we can take $k \log n$ bits of randomness and refer to them as the coefficients of the polynomial to get a random member of the family.
Such $k$-wise independent functions are useful for when we want a function that acts as a random function on small sets.
But in certain aspects they behave as random functions even for bigger sets.

Let  $\randfuncfamily$ be the set of all functions $\randfunc:[n] \to [n]$.
Given a subset $B \subseteq [n]$ and a specific subset $S \subset [n]$, consider a random function $\randfunc \sim \randfuncfamily$.
The expected size of the intersection between the image of $B$ under $\randfunc$ and $S$ is ${|S| \over n} |B|$.
Consider the event that the intersection is empty, i.e.\ $\randfunc(B) \cap S = \emptyset$.
Compare the probability for the event when the function is chosen from $\randfuncfamily$ vs.\ when the function is chosen from $\kwisefuncfamily$.
Indyk~\cite{Indyk2001} showed that if ${|S| \over n}|B|$ is sufficiently smaller than $k$, then the two probabilities are not far apart.

\begin{proposition}[Lemma 2.1 in~\cite{Indyk2001}]
	\label{prop_indyk_lemma}
	For every subsets $B \in {[n] \choose \turnix}$, $S \in {[n] \choose j}$, and every family of $k$-independent functions $\kwisefuncfamily$,
	if ${(j-1) \turnix \over n} \le {k-1 \over 2e}$ then
	\[
	\left| \Pr_{\kwisefunc \sim \kwisefuncfamily}[h(B) \cap S = \emptyset] - \Pr_{\randfunc \sim \randfuncfamily}[f(B) \cap S = \emptyset] \right| \le 2^{-k+1}.
	\]
\end{proposition}

\section{Introduction to the Card Guessing Game}

A card guessing game is played between a Dealer and a Guesser.
At the beginning of the game, the Dealer holds a deck of $n$ distinct cards (labeled $1, \dots, n$).
In every turn, the Dealer chooses a card from the deck, draws it, and places it face-down.
The Guesser guesses which card was drawn, the card is then revealed and discarded from the Dealer's deck.
The Guesser gets a point for every correct guess.
The game continues, for $n$ turns, until the Dealer has no cards to draw.

Assume that the Dealer draws cards uniformly at random from the deck.
A Guesser with \emph{perfect memory} can keep track of all cards played so far and guess cards that are still in the deck.
In turn $n-\turnix$ there are $\turnix$ cards in the Dealer's deck and the Guesser's probability to guess the next card correctly is $1/\turnix$.
Hence, the expected number of correct guesses in the game is
\[
{1 \over n} + {1 \over n-1 } + \dots + {1 \over 2} + 1 = H_n \approx \ln n.
\]

On the other extreme, a Guesser with \textit{no memory at all} can guess the same card over and over again without knowing whether this card was picked already or not.
Such a behavior would result in $1$ correct guess with probability $1$.

\begin{question}
	\begin{center}
		How well can a Guesser with $m$ memory bits play?
	\end{center}
\end{question}

A \textbf{transcript} of a card guessing game is a sequence of pairs $\{(g_\turnix, d_\turnix)\}_{\turnix=1}^n$ that describes that at turn $\turnix$ the Guesser guessed the card $g_\turnix$ and that the Dealer drew the card $d_\turnix$.
The number of correct guesses during a game is the number of turns during which $g_\turnix = d_\turnix$.
In our game, the Guesser aims to maximize the number of correct guesses.

\paragraph{Guesser:}
A Guesser consists of two probabilistic functions:
\begin{enumerate}
	\item \textbf{State transition function:} taking a memory state and a card drawn, and assigning a new memory state.
	\item \textbf{Guessing function:} receiving a memory state and outputing a card to guess.
\end{enumerate}
A memory bounded Guesser can use only $\guessermemorysymbol$ bits to store the memory state, so we refer to each state as a member of $\zo^\guessermemorysymbol$.

\paragraph{Randomness:}
We assume that the Guesser has random bits that both of the above functions may use, e.g.\ to select a random element of a set as a guess.
We differentiate between random bits that are used \emph{on the fly}, i.e.\ ``read once'', and random bits that are accessed several times, i.e.\ \emph{long lasting}.
We charge the Guesser for the latter but not for the former.
The Guesser may use her long lasting random bits for a secret permutation, seed for a pseudo-random generator, and any random object that may assist her.

We are not concerned with the number of on the fly random bits.
For our constructive results, we measure the amount of long lasting random bits that the Guesser uses, as well as suggest computationally efficient solutions which are good against computationally powerful dealers.
As for the impossibility results, we will show that they hold even if the Guesser is computationally unbounded, uses cryptography, and regardless of how much randomness, of both kinds, the Guesser uses.

\paragraph{Static vs.\ Adaptive Dealers:}

To show that the Guesser's performance vary with the Dealer's abilities, we present the different flavors of Dealers we consider:
\begin{itemize}
	\item
		The most benign Dealer we consider is a Dealer that shuffles the deck at the beginning of the game and draws cards one by one.
		This is equivalent to drawing cards at random from the deck in every turn.
		We call this Dealer \textbf{random shuffle} as it remains with the same shuffle throughout the game, and the deck is shuffled uniformly at random.
	\item
		The second Dealer we consider may be familiar with the Guesser's behavior and fixes the deck in advance in some particular order.
		As the deck of the Dealer remains the same throughout the game, we call this Dealer \textbf{static}.
	\item
		The third possibility we consider is a Dealer that is adaptive and selects the cards according to past guesses made by the Guesser, thoughtfully, in an adversarial manner.
		For our impossibility result, we do not assume that the Dealer is familiar with the Guesser's algorithm.
		We call such a Dealer \textbf{adversarial~adaptive}, or just \textbf{adaptive} for short.
		If the Dealer is not even aware of the memory size of the Guesser, we call it \textbf{adaptive universal}.
\end{itemize}

The static Dealer and the adaptive Dealer aims to minimize the number of correct guesses.
For this purpose, the static Dealer chooses a worst case ordering of the deck in advance, and the adaptive Dealer uses a \textbf{choosing strategy}, a function from a transcript prefix $\{(g_\turnix, d_\turnix)\}_{\turnix=1}^{k-1}$ to a distribution over a set of cards from which the Dealer samples a card $d_k$ to draw.
For example, a (silly) adaptive Dealer may look at the last guess and if the guessed card is still in the deck then draw it.

While the Guesser is limited to using~$\guessermemorysymbol$ bits of memory, the Dealer remembers everything that happened since the beginning of the game and may act accordingly.
This puts the Guesser at a disadvantage, as the Guesser needs to remember and maintain both a sketch of the history and in particular the parts of history that are relevant to the Dealer's strategy.
In light of this, we observe that only some guesses may be fruitful.
Call a guess \textbf{reasonable} if it is one of the cards that the Dealer may draw.
Clearly, a correct guess is necessarily reasonable.
Against a Dealer that draws cards from the deck at random, a reasonable guess is a synonym for a card that was not played yet.
As remarked above, when showing the impossibility results we construct a computationally efficient Dealer.
This emphasizes the role of adaptivity, especially when compared to the Guesser.

The state of the Guesser consists of $m$ bits and we assume that they are secret, i.e.\ that the adversary cannot access them when choosing the next card. The only inforamtion the Delaer has is the history (transcript).
\subsection{Basic guessing techniques}
\label{sec_basic_guessing_strategies}
In this section, we describe basic guessing techniques that a Guesser with  $\guessermemorysymbol$ bits of memory may use.
For a comparison of these techniques, see the first three rows of~\Cref{table:results}.

\paragraph{Subset Guessing:}
The Guesser chooses a random (or predetermined) subset of cards $A \in {[n] \choose \guessermemorysymbol}$ and pretends as if there are only $\guessermemorysymbol$ cards in the deck.
In every turn, the Guesser guesses one of the cards from~$A$ that were not played so far.
Each card requires one bit, so this strategy requires~$\guessermemorysymbol$ bits in total.
Counting only turns in which the Dealer draw cards from~$A$, we get that the Guesser makes $\ln \guessermemorysymbol$ correct guesses in expectation over Guesser's and Dealer's randomness.

While this technique ensures that all guesses during the game are reasonable, only on $\guessermemorysymbol$ turns a card from $A$ will be drawn.
These $\guessermemorysymbol$ cards have to be drawn at some point in the game, and the Guesser is agnostic about when exactly these cards are selected.
It follows that this technique performs equally well against the different kinds of Dealers.

\paragraph{``Remember'' the last $k$ cards:}
With only $\log n$ memory bits, the Guesser can correctly guess the last card in the game:
Initialize memory with~$\sum_{i=1}^n x \pmod n$ and remove every drawn card from the sum.
Just before the last turn, the memory will contain the one card that was not drawn yet.
This technique generalizes well to $k$ cards by storing the sums~$S_p = {\sum_{x=1}^n x^p \pmod n}$ for $p = 1, \dots, k$ and removing~$d^p$ from the respective sums when the card $d$ is drawn.
When $k$ cards are left, solving the equation system reveals the missing cards (See Chapter One in \cite{Muthukrishnan2005}).
Since each sum requires $\log n$ bits, a total of $k \log n$ bits are required to accurately identify the last $k$ cards.
This allows the Guesser to reasonably guess in the last $k$ turns, and by guessing at random, the Guesser makes~$\ln k$ correct guesses in expectation using $k \cdot \log n$ bits of memory.
Thus, a Guesser with $\guessermemorysymbol$ bits of memory can score $\ln \lfloor{ m \over \log n}\rfloor$ correct guesses in expectation, when playing against any Dealer.
\newline
\newline
As we saw, these techniques works well against \emph{any} Dealer.
The two methods (Remembering last cards and subset guessing) are compatible and we can combine them against the random-shuffle Dealer (but not against the others):
of the $\guessermemorysymbol$ bits, use $\guessermemorysymbol/ 2$ for the first method and $\guessermemorysymbol /2$ for second one.
The last card from the subset is expected when there are ${n \over 1 + \guessermemorysymbol/2} < {2n \over m}$ cards left until the end of the game.
The Guesser ``remembers'' the last~$m \over 2 \log n$ cards, so for $m \leq \sqrt{n}$ the two useful periods do not overlap.
We get that a Guesser with $m \le \sqrt{n}$ memory bits can expect to score $2 \ln m - \ln \log n - \ln 2$.
For $m = \sqrt{n}$, this is near optimal.

As we will see in \Cref{sec_static_dealer}, it is possible to do much better.

\section{Static Dealer}
\label{sec_static_dealer}
We first present a guessing strategy that requires low memory and no randomness, and is highly effective against the random-shuffle Dealer (\Cref{sec_static_dealer:sec_following_subset_guesser}).
We then show a randomized version of it that requires low memory and little randomness, and is highly effective against any static Dealer (\Cref{sec_static_dealer:sec_random_subsets_guesser}).
In \Cref{sec_static_dealer:upper_bound} we show that these guessing techniques are optimal against the random-shuffle Dealer, and that no memory bounded Guesser with less memory can perform asymptotically better.

\subsection{Following-Subsets Guesser vs.\  Random shuffle Dealer}
\label{sec_static_dealer:sec_following_subset_guesser}
We present a computationally efficient guessing technique that requires low memory, no randomness, and is highly effective against the random-shuffle Dealer.
We first show that $\log^2 n + \log n$ memory bits suffice to score~$1/2 \log n$ correct guesses in expectation when playing against the random-shuffle Dealer, and then we generalize this technique for Guessers with more memory.

In terms of memory usage, we use the simple idea of summing cards as we did in the ``Remembering last cards'' guessing technique.
The general idea is to follow the cards that appeared in various subsets of $[n]$.
For each such subset we store two accumulators:
\begin{enumerate}
	\item
		Sum of the values of the cards from the set seen so far (``remember last card'').
	\item
		Number of cards from the set seen so far.
\end{enumerate}
The memory needed for the two accumulators is $O(\log n)$ bits.
In fact, for a set of size~$w$ only~$2 \log w$ bits are needed, $\log w$ to count how many cards from the set appeared, and another $\log w$ to recover the last card from the set, by storing the sum of all cards $\bmod w$.
At the time that all but one card appeared (as can be indicated by the number of cards accumulator), the Guesser can recover this single card, and be certain that this card wasn't played yet by the Dealer, and as a result, the Guesser can reasonably guess this card.

By tracking multiple sets, the Guesser may have more than one card to guess from.
Against the random-shuffle Dealer that plays with a randomly shuffled deck, this doesn't really matter which one is guessed (at least not for the expectation).

\paragraph{Subset construction:}
We consider all the subsets of the form $[1-w]$ for $w=2^i$. I.e.\ the subsets are:
\[
[1-2],[1-4],[1-8],[1-16],[1-32], \ldots , [1-n]
\]
If there is a subset (range) where a single card is missing, then this card is the current guess.

Observe that in this construction there cannot be competing good cards to guess.
For all~$k < k'$, if a card~$j$ is missing from the set $[1-k]$, then there cannot be a different one missing from the set $[1-k']$\footnote{The Guesser may conclude more than one missing card in some cases.
For example, if one card is missing from $[1-k]$ and exactly two cards are missing from $[1-k']$.
We ignore this ability because it doesn't seem to improve the Guesser's performance.}.

\begin{claim}
	There exists a Guesser with $\log^2 n + \log n$ memory bits that can score $1/2 \log n$ correct guesses in expectation when playing against the random-shuffle Dealer.
\end{claim}
\begin{proof}
	
Call a subset $[1,w]$ \textbf{useful} if the last card from it that appears is \emph{not} the last card in the next subset $[1,2w]$.
By guessing the last missing card from each subset over and over again, by the end of the game, each subset contributed a correct guess, but it could be that several subsets contributed the same guess.
However, if a subset is useful, then it is the only one to whom we attribute the correct guess.
So the number of correct guesses is simply the number of useful subsets.

The probability that a subset $[1,w]$ is useful is precisely the probability that in the `next' subset $[1,2w]$, the last card does not come from~$[1,w]$.
This is $(2w-w)/2w = 1/2$.	
By linearity of expectation, the expected number of useful subsets is therefore $1/2 \log n$ and this is also the expected number of correct guesses.

As the subset $[1-w]$ requires $2 \log w = 2 \log 2^i = 2i $ bits, we get that with
\[
2\sum_{i=1}^{\log n} \log 2^i = 2\sum_{i=1}^{\log n} i = \log n (1 + \log n ) = \log^2 n + \log n
\]
memory bits, a Guesser can make $1/2 \log n$ correct guesses in expectation.
\end{proof}

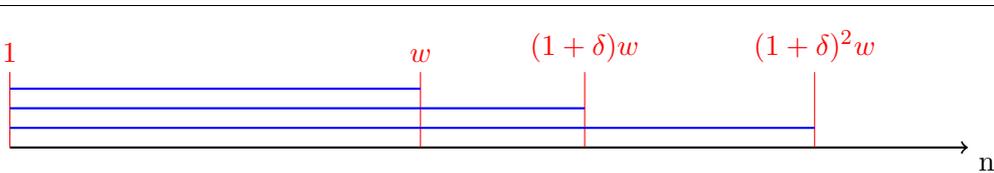
\begin{figure}[h]
	\begin{mdframed}	
		\input{parts/figures/following_subsets.tex}
		\caption{Following-Subsets.}
	\end{mdframed}
\end{figure}

Having more memory, we can have the subsets denser and have more subsets.
Suppose that the ratio between two successive ranges is $1+\ampliparam$ for $0< \ampliparam < 1$.
Then there are $\log_{1+\ampliparam} n$ such subsets.
The probability of a set being useful now (i.e.\ that its last member arriving does not belong to a subset that contains it) is $\ampliparam/(1+\ampliparam)$.
The expected number of useful sets is
\[
\frac{\ampliparam}{1+\ampliparam} \log_{1+\ampliparam} n = \frac{\ampliparam}{(1+\ampliparam) ln (1 +\ampliparam)} \ln n.
\]
This goes to $\ln n$ as $\ampliparam$ goes to zero.

In terms of space, the number of bits required for tracking $\log_{1+\ampliparam} n$ buckets is
\begin{flalign*}
	\sum_{i =1}^{\log_{1+\ampliparam} n} 2 i \log_2 (1+\ampliparam)
	&
	=
	2 \log_2 (1+\ampliparam) \sum_{i =1}^{\log_{1+\ampliparam} n} i
	\\
	&
	=
	2 \log_2 (1+\ampliparam){(\log_{1+\ampliparam} n)(1 + \log_{1+\ampliparam} n) \over 2}
	\\
	&
	=	
	\log_2 (1+\ampliparam){\log_2 n \over \log_2 (1+\ampliparam)}(1 + \log_{1+\ampliparam} n)
	\\
	&
	=
	\log_2^2 n \cdot \log_{1+\ampliparam} 2 + \log_2 n.
\end{flalign*}
Observe that the run time of the Guesser in every turn is at most $O(\log_{1+\ampliparam} n)$, thus the Guesser is computationally efficient.
\begin{corollary}
	For $0 < \ampliparam \le 1$, there exists a Guesser with $\log^2 n \cdot \log_{1+\ampliparam} 2 + \log n$ memory bits that makes
	\[
	\frac{\ampliparam}{(1+\ampliparam) ln (1 +\ampliparam)} \ln n
	\]
	correct guesses in expectation when playing against the random-shuffle Dealer.
\end{corollary}

\subsection{Random-Subsets Guesser vs.\ Static Dealer}
\label{sec_static_dealer:sec_random_subsets_guesser}

Consider a static Dealer such that instead of shuffling the deck uniformly at random, selects a worst case arrangement for the deck, knowing the Guesser's algorithm (but not her random bits).
For example, assume that the Dealer puts the Card `$1$' at the top of the deck and the Card `$2$' at the bottom of the deck.
In this case, the Following-Subsets technique yields a single correct guess.
The fact that the Dealer doesn't shuffle the deck uniformly but commits to a deck arrangement as the game begins can be interpreted as a mild adversarial intent and ability.

The Guesser can defend herself against such behavior by using a secret permutation~$\pi$, using her long lasting random bits.
She uses $\pi$ to randomize the subsets, where the subset $[1-w]$ tracks the cards $\pi(1), \dots, \pi(w)$.
The analysis and performance of the Following-Subsets technique holds as before, but $O(n \log n)$ bits of long lasting randomness are needed, which we wish to avoid.

We will show a related construction.
The Guesser uses her randomness to sample a secret permutation from a family of pairwise independent permutations, for example, from the family~$\pairwisefuncfamily=\{h(x) = ax + b\colon a \neq 0, b \ge 0\}$ over a finite field, and assigns the card~$x$ to the subset~$S_j$ if~$2^{j-1}< h(x) \le 2^j$.
That is, given a function $h$, the subset $S_j$ is the set of all $x \in [N]$ such that~$h(x) \in \{2^{j-1} + 1, \dots, 2^j\}$.

The Guesser tracks the cards that appeared from each subset, as we did previously.
In each turn, the Guesser attempts to recover a guess from a specific subset and guesses it.
In detail, when $\turnix \le n/2$ cards are left until the end of the game, the Guesser tries to recover a guess from the subset $S_j$ for $j = \log(n/2\turnix)$.
If all cards but one have appeared from~$S_j$, then the Guesser knows which card it is and guesses it.
For the first half of the game, the Guesser samples a random set of cards and guess cards that have not appeared from it.
A procedural description of the Guesser is provided in \Cref{alg_randomized_subsets_guesser}.

\begin{algorithm}[h]
	\caption{Randomized-Subsets}
	\label{alg_randomized_subsets_guesser}
	\begin{algorithmic}
		\State
			Sample a pairwise independent function $h \sim \pairwisefuncfamily$
		\State
			Split the cards to subsets, such that $x \in S_j$ if $h(x) \in \{2^{j-1} + 1, \dots, 2^j\}$
		\State
			Sample a set of cards $A$
		\While{$\turnix$ cards left for $\turnix \in \{n, \dots, n/2 \}$}
			\State
				Guess a random card from $A$ that has not appeared
			\State
				$d_\turnix \gets $ card drawn by Dealer
			\State
				Discard $d_\turnix$ from the subset $S_j$ that contains it
		\EndWhile
		\While{$\turnix$ cards left for $\turnix \in \{n/2 + 1, \dots, 1\}$}
			\State
				$j \gets \lfloor \log (n / 2\turnix) \rfloor$
				\Comment{Subset to consider}
			\If{$|S_j| = 1$}
				\Comment{Can recover the last card}
				\State
					$g_\turnix \gets $ last card in $S_j$
			\Else
				\State
					$g_\turnix \gets $ don't care
			\EndIf
			\State
				Guess $g_\turnix$
			\State
				$d_\turnix \gets $ card drawn by Dealer
			\State
				Discard $d_\turnix$ from the subset $S_j$ that contains it
		\EndWhile
	\end{algorithmic}
\end{algorithm}

We will consider what are the chances that, in some turn, a specific subset yields the correct guess.
That is, that the next card that the Dealer draws resides in a specific subset with a single missing card.

\begin{theorem}
	\label{thm_random_subset_guesser_performance}
	There exists a Guesser that uses $\resultRandomSubsetSpace$ memory bits and $2 \log n$ random bits and is expected to score at least ${1 \over 4} \ln n$ correct guesses in a game against any static Dealer.
\end{theorem}
\begin{proof}
	\label{proof_random_subset_guesser_performance}

	Consider a game between a Guesser that follows the Random-Subsets technique that plays against a static Dealer, i.e.\ one that selects an arbitrary sequence of cards to play.
	
	For the first half of the game, the Guesser samples a random subset $A$ and tracks it.
	Half of the cards in $A$ are expected to appear in the first $n/2$ turns.
	By guessing cards from $A$ at random, the probability for guessing the first card from $A$ correctly is $1 \over |A|$, the probability to guess the second is $1 \over |A| - 1$ and so forth.
	Resulting in roughly $\ln(|A|) - \ln (|A| / 2) = \ln 2$ correct guesses in the first half of the game, with approximation error depending on the size of $A$.
	As $A$ can be small and tracking it requires a small amount of memory that can be reused (consider storing it in the MSB of the subset counter), we neglect it from our calculations.
	
	When $\turnix < n/2$ cards are left, pick~$j$ such that~$\lfloor \log n/2\turnix \rfloor \leq j \leq \lceil \log n/2\turnix \rceil$.
	Let $T_\turnix = \{x_1, \dots, x_\turnix\}$ be the ordered set of the remaining cards in the deck.
	What are the chances that the Guesser picks the card $x_i$?
	That is, what are the chances that $x_i \in S_j$ and all other cards from $T_\turnix$, do not.

	\begin{flalign*}
		\Pr[S_j = \{x_i\}]
		&
		=
		\Pr[x_i \in S_j \wedge \forall x' \in T_\turnix \setminus \{x_i\}\colon x' \notin S_j]
		\\
		&
		=
		\Pr[x_i \in S_j] \cdot \Pr[\forall x' \in T_\turnix \setminus \{x_i\}\colon x' \notin S_j | x \in S_j]
		\\
		&
		=
		\Pr[x_i \in S_j] \cdot \left( 1 - \Pr[\exists x' \in T_\turnix \setminus \{x_i\}\colon x' \in S_j | x \in S_j] \right)
		\\
		&
		\ge
		\Pr[x_i \in S_j] \cdot \left( 1 - \sum_{x' \in T_\turnix \setminus \{x_i\}} \Pr[ x' \in S_j | x \in S_j] \right).
	\end{flalign*}
	Since the cards are assigned to $S_j$ by using a pairwise independent permutation, the probability for a card~$x_i$ to be in $S_j$ is $2^j / n$, and for the same reason $\Pr[ x' \in S_j | x_i \in S_j] = {2^j -1 \over n}$.
	It follows that the probability for $x_i$  to be recovered by the $j$th subset is
	\begin{flalign}
		\Pr[S_j = \{x_i\}]
		&
		\ge
		{2^j \over n} \cdot \left(1 - (\turnix-1) {2^j -1 \over n}\right)
		\ge
		{1 \over 4t}
		\label[ineq]{proof_random_subset_guesser_performance:subset_recover_ith_card}
	\end{flalign}
	Where \Cref{proof_random_subset_guesser_performance:subset_recover_ith_card} is true since $j \approx \log {n \over 2\turnix}$.
	
	Since this is true for every $x_i$, in particular it is true for $x_1$.
	So we get that the probability to guess correctly when $\turnix$ cards are left is at least $1 / 4\turnix$.
	By linearity of expectation, we get that the expected number of correct guesses throughout the $n /2$ last turns is at least ${1 /4} \cdot \ln(n/2)$.
	By assigning each card in $S_j$ a value in $[|S_j|]$, we get that tracking $\log n$ subsets requires $\resultRandomSubsetSpace$ memory bits.
	For each card the Guesser performs two operations, assign to bucket, and remove from bucket, thus the Guesser is computationally efficient.
	As $2 \log n$ random bits are required for the pairwise independent permutation, the \namecref{thm_random_subset_guesser_performance} follows.
\end{proof}

\subsubsection{Amplification Towards $\ln n$}
\label{sec_static_dealer:sec_random_subsets_guesser:sec_amplification}
The above construction is simple and works well to get $1/4 \ln n$.
How can we improve it and get to $\ln n$?
We modify the above Guesser slightly to get an amplifiable construction.

\newcommand{\hnum}{{2.5 \log (1 / \ampliparam)}}

The idea is to sample more functions $h_1, h_2, \ldots, h_\hnum$ and for each function $h_i$ generate its collection of sets $\{S^{h_i}_j\}_{j=1}^{\log n}$.
The algorithm is now:
when $\turnix\le n/2$ cards are left, for $j = \log (n / 2t)$, the Guesser attempts to recover a reasonable guess (a card that the Dealer may draw) from the $j$th subset in each of the collections.
The Guesser takes the first collection that yields a subset with a single missing card and makes this card her guess.
For the first half of the game, as previously, the Guesser sample some set of cards and guess cards that have not appeared (Subset guessing).
A procedural description of this Guesser is provided in \Cref{alg_random_subsets_amplified}.

We show that in the construction above, every collection yields a reasonable guess with constant probability.
Since the collections are chosen independently, the probability of failure to recover a reasonable guess goes down with the number of functions we sampled.
If we sample independently~$O(1/\ampliparam)$ functions, we can get to probability $1-\ampliparam$ of at least one subset succeeding in suggesting a reasonable guess.

The issue with the analysis of this process is showing that a reasonable guess is indeed correct with probability proportional to the number of cards left in the deck (and not some constant fraction of that).
Since the adversary chooses the order of the cards, we need a construction of functions where we can say that if a subset yields a reasonable guess then it is also correct with the right probability.
We do not know whether this is true for pairwise independent functions.
As we will see, it is true for higher independence.
So our Guesser samples functions from a family $\kwisefuncfamily$ of $k$-wise independent functions (\Cref{def_k_wise_function}), and assigns cards to subsets in the same way.

\begin{algorithm}
	\caption{Amplifed Random-Subsets}
	\label{alg_random_subsets_amplified}
	\begin{algorithmic}
		\Parameter
		$\ampliparam \le 1$,
		$k \ge 2\log n$
		\State
			Sample $h_1, h_2, \dots h_\hnum$ functions from a family of $k$-wise independent functions $\kwisefuncfamily$
		\For {$h \in H$}
			\State
				Construct a collection of subsets $S^h_1, \dots, S^h_{\log n}$ such that $x \in S^h_j$ if $h(x) \in \{2^{j-1} + 1, \dots, 2^j\}$
		\EndFor
		\State
			Sample a set of cards $A$.
		\While{$\turnix$ cards left for $\turnix \in \{n, \dots, n/2 \}$}
			\State
				Guess a random card from $A$ that has not appeared
			\State
				$d_\turnix \gets $ card drawn by Dealer
			\State
				Discard $d_\turnix$ from the subsets $S^h_j$ that contains it
		\EndWhile
		\While{$\turnix$ cards left for $\turnix \in \{n/2 + 1, \dots, 1\}$}
			\State
				$j \gets \lfloor \log (n / 2\turnix) \rfloor$
				\Comment{Subsets to consider}
			\If{there exists $h \in \{h_1, \dots, h_\hnum\}$ such that $|S^h_j| = 1$}
				\Comment{Can recover last card}
				\State
					Let $h$ be the first such subset for which $|S^h_j| = 1$
				\State
					$g_\turnix \gets $ last card from $S^h_j$
			\Else
				\State
					$g_\turnix \gets $ don't care
			\EndIf
			\State
				Guess $g_\turnix$
			\State
				$d_\turnix \gets $ card drawn by Dealer
			\State
				Discard $d_\turnix$ from the subsets $S^h_j$ that contains it
		\EndWhile
	\end{algorithmic}
\end{algorithm}

\newcommand{\jtfloor}{{\lfloor\log (n / 2\turnix)\rfloor}}
\newcommand{\collectionh}{{\{S^\kwisefunc_j\}_{j=1}^{\log n}}}
\newcommand{\jthsubseth}{{S^\kwisefunc_j}}

During the first half of the game, the analysis is the same as in \Cref{thm_random_subset_guesser_performance}.

Recall that $T_\turnix = \{x_1, \dots, x_\turnix\}$ is the ordered set of the remaining cards in the deck when~$\turnix$ cards are left.
We show that for $j = \jtfloor$ the probability that the $j$th subset of any collection yields a reasonable guess, is at least a constant.
Fix some $h$ and recall \Cref{proof_random_subset_guesser_performance:subset_recover_ith_card} that states that for any $i$ the probability that $\jthsubseth$ is the singleton set that consists of $x_i$ is at least $1/4 \turnix$.
Therefore, the event that $\jthsubseth$ yields a reasonable guess is the disjoint union of events where $S^h_j$ is the singleton consisting of $x_i$ for $i \in [\turnix]$.
That is, the probability that $\jthsubseth$ yields a reasonable guess is at least
\begin{flalign}
	\Pr[\jthsubseth \text{ yield a reasonable guess}] \ge \sum_{i=1}^{\turnix} {1 \over 4 \turnix} = {1 \over 4}.
	\label[ineq]{claim_amplified_random_subset_yields_a_reasonable_guess}
\end{flalign}

Let the indicator random variable $\reasonableRV_\turnix$ be the event that some subset yields a reasonable guess when $\turnix$ cards are left.
The probability that no $j$th subset, of any collection, yields a reasonable guess when $\turnix$ cards are left is at most
\begin{flalign*}
	\Pr[\reasonableRV_\turnix = 0]
	&
	\le
	\left( 1 - {1 / 4}\right)^{\hnum}
	\\
	&
	<
	\left( {3 / 4}\right)^{\log 1/\ampliparam \over \log {4/3}}
	\\
	&
	=
	\left( {3 / 4}\right)^{\log_{4 /3} 1/\ampliparam}
	\\
	&
	=
	\left( {3 / 4}\right)^{\log_{3/4} \ampliparam}
	\\
	&
	=
	\ampliparam.
\end{flalign*}
\begin{corollary}
	\label{claim_amplified_random_subset_reasonable_probability}
	The probability that the subsets $S_j^{h_1}, \dots, s_j^{h_\hnum}$ yields a reasonable guess at turn $\turnix > n/2$ is at least $1 - \ampliparam$.
\end{corollary}
Given that $\jthsubseth$ yields a reasonable guess, what are the chances for it to be correct?
That is, what are the chances that $\jthsubseth$ is the singleton that contains the next card?

\begin{claim}
	\label{claim_amplified_random_subset_correct_probability}
	When $\turnix < n/2$ cards are left, for $j = \jtfloor$, for $k \ge 2 \log n$,
	the probability over the choice of $\kwisefunc \sim \kwisefuncfamily$, that $\jthsubseth$ yields a correct guess given that $\jthsubseth$ yields a reasonable guess is at least
	\[
	{1 \over t + o(1)}.
	\]
\end{claim}
\begin{proof}
	\label{prf_amplified_random_subset_correct_probability}
	\newcommand{\jimage}{{V_j}}
	Recall that $x \in \jthsubseth$ if $\kwisefunc(x) \in \{2^{j-1}+1, \dots, 2^j\}$.
	Denote by $\jimage = \{2^{j-1}+1, \dots, 2^j\}$ the possible images of elements from $\jthsubseth$.
	Let $B_{-i}$ be the set $\{x_1, \dots, x_\turnix\} \setminus \{x_i\}$, and denote the image of $B_{-i}$ under~$\kwisefunc$ by~$\kwisefunc(B_{-i}) = \{\kwisefunc(x) \colon x \in B_{-i}\}$.
	Let $A_{-i}$ be the event that $\kwisefunc(B_{-i}) \cap \jimage = \emptyset$.

	The subset $\jthsubseth$ yields a reasonable guess if and only if it intersects with the remaining cards at precisely one point.
	The event of recovering a reasonable guess is a disjoint union of the events that $x_i$ is in the subset~$\jthsubseth$ and all other cards are not in $\jthsubseth$, that is, $\kwisefunc(B_{-i}) \cap \jimage = \emptyset$, or in other words $A_{-i}$.
	Given that $\jthsubseth$ yields a reasonable guess, then exactly one of the events $A_{-i}$ occurred.
	Let  $P_i$ be the probability of event $A_{-i}$ given $x_i \in \jthsubseth$. I.e.\
	\begin{flalign*}
		\Pr[\jthsubseth \text{ yields a reasonable guess}]
		&
		=
		\sum_{i \in [\turnix]} \Pr[A_{-i} \wedge x_i \in \jthsubseth ]
		\\
		&
		=
		\sum_{i \in [\turnix]} \Pr[A_{-i} | x_i \in \jthsubseth] \cdot \Pr[x_i \in \jthsubseth]
		\\
		&
		=
		\sum_{i \in [\turnix]} P_i \cdot \Pr[x_i \in \jthsubseth]
	\end{flalign*}
	Note that given $x_i \in \jthsubseth$, we have $(k - 1)$-wise independence of the values \[
	\kwisefunc(x_1), \kwisefunc(x_2),  \ldots ,\kwisefunc(x_{i-1}), \kwisefunc(x_{i+1}), \ldots ,\kwisefunc(x_t).
	\]
	recall that  $\randfuncfamily$ is the set of all functions $\{\randfunc\colon [n] \to [n]\}$.
	\Cref{prop_indyk_lemma} states that for every set~$B_{-i}$ of size $\turnix-1$, the probability for the event $A_{-i}$, over the choice of a random function from $\randfuncfamily$, is roughly the same as over the choice of a random function from $\kwisefuncfamily$.
	Applying \Cref{prop_indyk_lemma} on our parameters we get that if ${2^j + 1 \over n} (t-1) \le {k-2 \over 2e}$ then
	\[
		|\Pr_{\kwisefunc \sim \kwisefuncfamily}[A_{-i}| x_i \in S^h_j] - \Pr_{\randfunc \sim \randfuncfamily}[A_{-i}| x_i \in S^h_j]| \le 2^{-k+1}.
	\]
	In our case $j \approx \log {n \over 2t}$, thus ${2^j + 1 \over n} (t-1) < 1 \le {k-2 \over 2e}$ and since $k \ge 2 \log n$, the \namecref{prop_indyk_lemma} applies.

	The conditional probability for the event $A_{-i}$ over the choice of $\randfunc\sim \randfuncfamily$ depends solely on the size of the set~$B_{-i}$.
	Since $|B_{-i}|$ is the same for every $i$, we get that probability over the choice of $\kwisefunc \sim \kwisefuncfamily$ for the event $A_{-i}$ is about the same for every $i$.
	We conclude the distance between the probability for $A_{-i}$ and $A_{-i'}$, that is
	\begin{flalign}
		|\Pr_{\kwisefunc \sim \kwisefuncfamily}[A_{-i}| x_i \in S^h_j] - 	\Pr_{\kwisefunc \sim \kwisefuncfamily}[A_{-i'}| x_{i'} \in S^h_j]| \le 2^{-k+2}.
		\label[ineq]{prf_amplified_random_subset_correct_probability:distance_between_a_i_and_a_i_tag}
	\end{flalign}
	Observe that for $\randfunc \sim \randfuncfamily$ the event $A_{-i}$ is independent from the event $x_i \in \jthsubseth$, and as a result
	\begin{flalign*}
		\Pr_{\randfunc \sim \randfuncfamily} [A_{-i} | x_i \in \jthsubseth]
		&
		=
		\Pr_{\randfunc \sim \randfuncfamily} [A_{-i}]
		\\
		&
		=
		\Pr_{\randfunc \sim \randfuncfamily} [\randfunc(B_{-i}) \cap \jimage = \emptyset]
		\\
		&
		=
		\Pr_{\randfunc \sim \randfuncfamily} [\forall x \in B_{-i} \colon \randfunc(x) \notin \{2^{j-1}+1, \dots, 2^j\}]
		\\
		&
		=
		\left(1 - {2^{j-1} \over n}\right)^{\turnix-1}.
	\end{flalign*}
	\newcommand{\constant}{{\alpha}}
	Since $j \approx \log(n / 2\turnix)$, this is approximately some constant $\constant$ which is roughly $e^{-0.5}$.

The next card in the deck is $x_1$.
Given that $\jthsubseth$ yields a reasonable guess, what is the probability for its guess to be correct?
The conditional probability of guessing $x_1$ as the next card is
	\begin{flalign}
		{P_1 \cdot \Pr[x_1 \in \jthsubseth] \over \sum_{i=1}^\turnix P_i \cdot \Pr[x_i \in \jthsubseth]}
		&
		=
		{P_1 \over \sum_{i=1}^\turnix P_i }
		\label[equa]{prf_amplified_random_subset_correct_probability:k_wise_ind_function}
		\\
		&
		\ge
		{P_1 \over \sum_{i=1}^\turnix (P_1 + 2^{-k+2}) }
		\label[ineq]{prf_amplified_random_subset_correct_probability:follows_from_distance_between_a_is}
		\\
		&
		=
		{1 \over \turnix + {\turnix \cdot 2^{-k + 2} \cdot P_1^{-1}}}
		\nonumber
	\end{flalign}
	Where \Cref{prf_amplified_random_subset_correct_probability:k_wise_ind_function} is true since cards are assigned to buckets using $k$-wise independent function, thus for every $x_i$ the probability that $\kwisefunc(x_i) \in \jimage$ is the same, and
	\Cref{prf_amplified_random_subset_correct_probability:follows_from_distance_between_a_is} follows from \Cref{prf_amplified_random_subset_correct_probability:distance_between_a_i_and_a_i_tag}.
	
	For $k \ge 2 \log n$, and for $\turnix < n/2$, the term ${\turnix \cdot 2^{-k + 2} \cdot P_1^{-1}} \approx \turnix \cdot e^{1/2} / 2^{k-2} = o(1)$.
	We conclude
	\[
	\Pr[\jthsubseth \text{ yields a correct guess} | \jthsubseth \text{ yields a reasonable guess}]
	\ge
	{1 \over \turnix + o(1)}.
	\]
\end{proof}

The functions $\{\kwisefunc_1, \dots, \kwisefunc_\hnum\}$ are chosen independently of each other, and hence so are the $j$th subsets.
As a result, since the \namecref{claim_amplified_random_subset_correct_probability} holds for an arbitrary $\jthsubseth$, it implies that for the first subset that yields a reasonable guess, the probability that its guess is correct is at least ${1 \over \turnix + o(1)}$.

Recall the indicator random variable $\reasonableRV_\turnix$ that corresponds to the event that some subset yields a reasonable guess when $\turnix$ cards are left.
Let the indicator random variable $\correctsRV_\turnix$ be the event of a correct guess when $\turnix$ cards are left.
We conclude that the expected number of correct guesses since turn $n/2$ is at least
\begin{flalign}
	\sum_{\turnix = {n/2 -1}}^1
	\expected[\correctsRV_{\turnix}]
	&
	=
	\sum_{\turnix = {n/2 -1}}^1
	\Pr[\reasonableRV_{\turnix}= 1] \cdot \expected[\correctsRV_{\turnix} | \reasonableRV_{\turnix}= 1]
	+
	\Pr[\reasonableRV_{\turnix}= 0] \cdot \expected[\correctsRV_{\turnix} | \reasonableRV_{\turnix}= 0]
	\label[equa]{amplified_random_subsets:total_expectation}
	\\
	&
	=
	\sum_{\turnix = {n/2 -1}}^1
	\Pr[\reasonableRV_{\turnix}= 1] \cdot \expected[\correctsRV_{\turnix} | \reasonableRV_{\turnix}= 1]
	\label[equa]{amplified_random_subsets:non_reasonable_guess_cannot_be_correct}
	\\
	&
	>
	\sum_{\turnix = {n/2 -1}}^1
	(1-\ampliparam) \cdot {1 \over \turnix + o(1)}
	\label[ineq]{amplified_random_subsets:probability_or_reasonable_guess_to_be_correct}
	\\
	&
	\approx
	(1-\ampliparam) \ln (n/2)
	\nonumber
	.
\end{flalign}
Where
\Cref{amplified_random_subsets:total_expectation} is from law of total expectation,
\Cref{amplified_random_subsets:non_reasonable_guess_cannot_be_correct} is true since a correct guess is necessarily reasonable,
\Cref{amplified_random_subsets:probability_or_reasonable_guess_to_be_correct} follows from \Cref{claim_amplified_random_subset_correct_probability} and from \Cref{claim_amplified_random_subset_reasonable_probability}.

In term of memory, we allocate $2 \log n$ bits per subset, and a total $2 \log^2 n$ per collection of subsets.
To track $\hnum$ collections a Guesser requires $O(\log(1/\ampliparam) \log^2 n)$ bits of memory.

Sampling a function from a $(2\log n)$-wise independent family of functions requires $2\log^2 n$ bits of long lasting randomness, e.g., consider the randomness as representing coefficients of a polynomial over a finite field.
To sample $\hnum$ such functions, the Guesser requires $O(\log(1/\ampliparam) \log^2 n)$ bits of randomness.
Furthermore, the Guesser is computationally efficient:
In each turn $O(1/\delta)$ functions are computed and this many subsets are updated.

\begin{theorem}
	There exists a Guesser with $O(\log(1/\ampliparam) \log^2 n)$ bits of memory and with $O(\log(1/\ampliparam) \log^2 n)$ bits of randomness that scores $(1-\ampliparam) \ln n$ correct guesses on expectation against any static Dealer.
\end{theorem}

\subsubsection{Low-memory Case}

The guessing techniques seen so far assumed that the Guesser has about $\log^2 n$ bits of memory.
But what can be done if $m$ is small, say $m << log^2 n$?
It is possible to fall back to the subset guessing technique and get $\ln m$ correct guesses in expectation.
That would work for both the random shuffle and the static cases (also for the adaptive).
But we can do better.

Our Guessers can pretend as if the domain is of size $2^{\sqrt{m}}$ and ignore all other cards!
In that case, the Following-Subsets guessing technique is expected to yield about $1/2 \log 2^{\sqrt{m}} = 1/2 \sqrt{m}$ correct guesses against the random-shuffle Dealer.
The Random-Subsets guessing technique is expected to yield about ${1 \over 4} \ln 2^{\sqrt{m}}$ correct guesses against a static Dealer, i.e.\ also $O(\sqrt{m})$.

So we get that for any $m$, a Guesser can score at least $O(\min\{\ln n, \sqrt{m}\})$ when playing against any static Dealer.

\subsection{Bounds on best possible Guesser against Random Dealer}
\label{sec_static_dealer:upper_bound}
We show that \emph{the guessers of the previous section are the best possible low memory guessers}, up to constants.

\begin{theorem}
	Any Guesser using $m$ bits of memory can get \emph{at most} $O(\min\{\ln n, \sqrt{m}\})$ correct guesses in expectation when playing against the random-shuffle Dealer.
\end{theorem}

Our proof will use compression argument.
We will present an encoding scheme that utilizes correct guesses to achieve shorter descriptions.
As the expected length of the description is bounded by the entropy of a random input, we get an upper bound on the expected number of correct guesses for every memory bounded Guesser.
Our proof for an adaptive Dealer (\Cref{section_upper_bound}) will follow a similar structure.

Let $\guesserrandomness$ be the Guesser's randomness and  $\permutation$ be the shuffle sampled by the Dealer's randomness.
Denote by \guesserfixed a Guesser with fixed randomness $\guesserrandomness$, and \dealerstaticvanilla a static Dealer with a deck arranged according to $\permutation$.
Let $\gameval(\text{\dealerstaticvanilla, \guesserfixed})$ be the number of correct guesses during the last~$\orderedsetsize = n^{1-\orderedsetsizeparam}$ turns, for some $\orderedsetsizeparam > 0$.
Let a random variable $\correctsRV = \gameval(\dealersymbol, \text{\guesser})$ and let $\correctsexpected$ be the expected number of correct guesses during that last $\orderedsetsize$ turns where the expectation is taken over Guesser's and Dealer's randomness, i.e.
\[
\correctsexpected
=
\expected_{\guesserrandomness, \permutation}\left[\correctsRV \right]
=
\expected_{\guesserrandomness, \permutation}\left[\gameval(\text{\dealerstaticvanilla, \guesserfixed})\right]
.
\]

\newcommand{\shuffleset}{\Pi}

Denote by $\shuffleset_\orderedset$ the set of all deck arrangements such that the last $\orderedsetsize$ cards in the deck are the ordered set~$\orderedset$.
So we can consider the expectation over the choice of the last $\orderedsetsize$ cards.

\[
\correctsexpected
=
\expected_{\guesserrandomness, \permutation}\left[\correctsRV \right]
=
\expected_{\guesserrandomness, \orderedset} \expected_{\permutation \in \Pi_\orderedset}\left[\correctsRV \right]
=
\sum_\guesserrandomness \expected_\orderedset \expected_{\permutation \in \Pi_\orderedset} \left[ \correctsRV | \guesserrandomness \right] \cdot \Pr \left[\guesserrandomness\right]
.
\]
In particular, we focus on bounding the term
\[
\expected_\orderedset \expected_{\permutation \in \Pi_\orderedset} \left[ \correctsRV | \guesserrandomness \right]
=
\expected_\orderedset \expected_{\permutation \in \Pi_\orderedset} \left[\gameval(\text{\dealerstaticvanilla, \guesserfixed}) | \guesserrandomness \right]
.
\]

We claim that no Guesser can expect to guess correctly too many times at the last $\orderedsetsize$ turns.
We prove this by presenting an encoding scheme for ordered sets $\orderedset$ (the last $\orderedsetsize$ cards played by the Dealer) that utilizes correct guesses to achieve shorter descriptions.
The encode function works by simulating the Guesser on a deck of card, where the first $n-\orderedsetsize$ cards are from $[n] \backslash \orderedset$ and the $\orderedsetsize$ cards are ordered according to $\orderedset$.
Record the Guesser's memory ($\guessermemorysymbol$ bits) after the first $n-\orderedsetsize$ turns and from that point on see when the Guesser gives correct guesses.
These can be used to help describe $\orderedset$.
Let the number of correct guesses be $\correctsRV$.
If $\correctsRV \ge \correctsnum$ for some $\correctsnum > 0$, then to record $\orderedset$, we note the location of some $\correctsnum$ places with a correct guess and provide the remaining $\orderedsetsize-\correctsnum$ missing values.
So how many possibilities do we have?
For the memory $2^m$, for the correct guesses locations $\orderedsetsize \choose \correctsnum$ and for the other values an ordered set of size $\orderedsetsize-\correctsnum$ out of~$n$.

\newcommand{\bestpermutationforguesser}{{\permutation_{\orderedset, \guesserrandomness}}}
Recall that $\shuffleset_\orderedset$ is the set of all deck arrangements for which the last~$\orderedsetsize$ cards are the ordered set~$\orderedset$.
The order of the first $n-\orderedsetsize$ cards may lead the Guesser to different memory states; in terms of correct guesses, some of which may be more beneficial then others, especially for a Guesser with fixed randomness.
Given an ordered set $\orderedset$ and Guesser's randomness $\guesserrandomness$, let $\bestpermutationforguesser \in \shuffleset_\orderedset$ be the deck arrangement for which the Guesser \guesserfixed makes the most correct guesses in the last $\orderedsetsize$ turns.
That is
$$
\forall \pi \in \Pi_B\colon
\gameval(\text{\dealerstaticvanilla, \guesserfixed})
\le
\gameval(\text{\dealerstatic, \guesserfixed})
.
$$

The encoding function will simulate a game against a static Dealer with fixed deck order~$\bestpermutationforguesser$ to encode $\orderedset$.
Fix some prefix free code (\Cref{def_prefix_free_code}) for ordered subsets.
The scheme will use this code for the cases where there are not enough correct guesses to utilize.

\begin{definition}[$\encodefuncupper$]
	To encode $\orderedset$, an ordered subset of $[n]$ of size $\orderedsetsize$, the function $\encodefuncupper$ records and simulates a game between the Guesser~\guesserfixed and the static Dealer \dealerstatic.

Let $\turnssetsymbol'$ be the set of locations during the last $\orderedsetsize$ turns at which \guesserfixed makes a correct guess, i.e.\
	$$
	\turnssetsymbol' = \{n-\currentepochbegin \le \turnix \le n -\currentepochbegin+\currentlength-1 | g_\turnix \text{ is a correct guess}\}.
	$$
	
	\begin{itemize}
		\item
			If $|\turnssetsymbol'| < \correctsnum$ then the code is made of an indicator bit $0$ and an explicit prefix-free description of $\orderedset$.
		\item
			If $|\turnssetsymbol'| \ge \correctsnum$ then let $\turnssetsymbol$ be the first $\correctsnum$ turns at which the Guesser guessed correctly.
			
			The code is made of:
			\begin{enumerate}
				\item
					An indicator bit $1$.
				\item
					Guesser's memory state $\memorystatesymbol$ at turn $n-\orderedsetsize$ ($m$ bits).
				\item
					Description of $\turnssetsymbol$, the locations of the first $\correctsnum$ correct guesses made by \guesserfixed during the last~$\orderedsetsize$ turns ($\log {\orderedsetsize \choose \correctsnum}$ bits).
				\item
					Description of $\orderedset \setminus \{g_\turnix | \turnix \in \turnssetsymbol\}$  ($\log\left( n(n-1)\dots(n-\orderedsetsize+\correctsnum+1)\right)$ bits).
			\end{enumerate}
	\end{itemize}
\end{definition}
A visualization of the information stored is provided in \Cref{figure:low_memory_upper_bound_encoding_full}.

\begin{figure}[h]
	\begin{mdframed}	
		\resizebox{\textwidth}{!}{%
			\input{parts/figures/correct_guess_encoding_full.tex}
		}
		\caption{$\encodefuncupper$. Colored - information stored by the encoding scheme.}
		\label{figure:low_memory_upper_bound_encoding_full}
	\end{mdframed}
\end{figure}

Similarly we define the decode function.
\begin{definition} [$\decodefuncupper$]
	If the indicator bit is $0$, then decode the set in the natural way.
	If the indicator bit is $1$, then parse the other bits as a $3$-tuple $(\memorystatesymbol,  \turnssetsymbol, \orderedsetexplicit)$ as encoded by $\encodefuncupper$.
	The function $\decodefuncupper$ works by simulating and recording a partial game between Dealer \dealerstaticdec and Guesser \guesserfixed:
	\begin{itemize}
		\item
			Initialize the Guesser \guesserfixed with memory state $\memorystatesymbol$ at turn $n-\orderedsetsize$ and simulate $\orderedsetsize$ turns against the Dealer~\dealerstaticdec.
		\item
			If in turn $n-\orderedsetsize \le \turnix \le n$ the guess $g_\turnix$ is tagged as correct (by $\turnssetsymbol$), then \dealerstaticdec draws the card~$g_\turnix$, otherwise \dealerstaticdec draws the next card from $\orderedsetexplicit$.
		\item
			Output the set of cards drawn by Dealer \dealerstaticdec in the order they were drawn.
	\end{itemize}
\end{definition}
A procedural description of $\decodefuncupper$ is specified in \Cref{alg_decode_static}.

\begin{algorithm}
	\caption{$\decodefuncupper$}
	\label{alg_decode_static}
	\begin{algorithmic}
		\Parameter
			$\orderedsetsize \in [n]$,
			$\guesserrandomness$
		\Input
			$
			\memorystatesymbol \in \zo^\guessermemorysymbol,
			\turnssetsymbol \in {[\orderedsetsize] \choose \correctsnum},
			\orderedsetexplicit \text{ an ordered subset of $[n]$ of size $\orderedsetsize - \correctsnum$}
			$
		\State
			Initialize Guesser \guesserfixed at turn $n-\orderedsetsize$ with memory state $\memorystatesymbol$.
		\State
			$\orderedsetdec \gets \emptyset$
		\For{$\turnix \in \{n-\orderedsetsize, \dots, n\}$}
			\State
				$g_\turnix \gets$ guess made by Guesser \guesserfixed
			\If {$g_\turnix$ is tagged as correct (according to $\turnssetsymbol$)}
				\State
					$d_\turnix \gets g_\turnix$
			\Else
				\State
					$d_\turnix \gets$ next card from $\orderedsetexplicit$
			\EndIf
			\State
				Append $d_\turnix$ to $\orderedsetdec$
			\State
				Update \guesserfixed memory state according to $d_\turnix$
		\EndFor
		\State
			\Return $\orderedsetdec$
	\end{algorithmic}
\end{algorithm}
We assume that $\guesserrandomness$ is given to us ``for free'' and is known during encoding and decoding of the ordered set~$\orderedset$.
We justify this assumption in two different ways:
\begin{itemize}
	\item
		Fixing $\guesserrandomness$ we can consider a specific encoding scheme for ordered sets $\encodefuncupper$.
	\item
		We can assume that we encode a pair $(\guesserrandomness, \orderedset)$ where $\guesserrandomness$ is written explicitly in some natural way right next to $\encodefuncupper(\orderedset)$.
\end{itemize}

\begin{claim}
	\label{lemma_is_encoding_scheme_static}
	For every ordered set $\orderedset$: $$\decodefuncupper(\encodefuncupper(\orderedset)) = \orderedset.$$
\end{claim}
\begin{proof}
	If $\encodefuncupper(\orderedset)$ encodes $\orderedset$ explicitly, then this is trivial.
	Otherwise, both $\encodefuncupper$ and $\decodefuncupper$ works by simulating a game between a Guesser and a Dealer.
	Since the Guesser's randomness is fixed, the same game transcript is simulated in both simulations.
	Therefore, the decoder simulated Dealer \dealerstaticdec draws the same cards in the same order and thus recovers the same ordered set.
\end{proof}

\begin{claim}
	The code produced by $\encodefuncupper$ is prefix-free.
\end{claim}
\begin{proof}
	Let two ordered sets $\orderedset, \orderedset'$ and denote by $I$ the first bit of $\encodefuncupper(\orderedset)$ and $I'$ the first bit of $\encodefuncupper(\orderedset')$.
	Observe that if $I \neq I'$ then the two descriptions cannot be a prefix of one another, and if $I = I'$ then the two descriptions are of the same length.
	To prove that the encoding scheme produce a prefix-free code, it suffice to show that for every $\orderedset$ it holds that $\decodefuncupper(\encodefuncupper(\orderedset)) = \orderedset$.
	By \Cref{lemma_is_encoding_scheme_static}, this is indeed the case so the \namecref{lemma_is_encoding_scheme_static} follows.
\end{proof}

\begin{corollary}
	If \guesserfixed makes $\correctsRV \ge \correctsnum$ correct guesses in the last $\orderedsetsize$ turns when playing against the static Dealer \dealerstatic then $\encodefuncupper(\orderedset)$ is of length
	$$
		\log \left( 2 \cdot 2^\guessermemorysymbol \cdot {\orderedsetsize \choose \correctsnum} \cdot n(n-1) \cdots  (n-\orderedsetsize+\correctsnum+1)  \right)
	$$
\end{corollary}

So we get that the encoding scheme saves bits for every correct guess while ``paying'' only~$m$ bits of memory.
The contradiction comes from counting the number of the ordered sets $\orderedset$ in two different ways:
\begin{itemize}
	\item
		$n(n-1) \cdots (n-\orderedsetsize+1)$ are all the possible options for ordered set $\orderedset$,
	\item
		and ${\orderedsetsize \choose \correctsnum} \cdot n(n-1) \cdots  (n-\orderedsetsize+\correctsnum+1) 2^{\guessermemorysymbol + 1}$ - upper bound on the possible options for ordered set $\orderedset$ according to the encoding.
\end{itemize}
So we have
\begin{flalign*}
	n(n-1) \cdots (n-\orderedsetsize+1)
	&
	\le
	{\orderedsetsize \choose \correctsnum} \cdot n(n-1) \cdots  (n-\orderedsetsize+\correctsnum+1) 2^{m + 1}
	\\
	\therefore
	(n-\orderedsetsize+\correctsnum)(n-\orderedsetsize+\correctsnum-1) \cdots (n-\orderedsetsize+1)
	&
	\le
	{\orderedsetsize \choose \correctsnum} \cdot 2^{m+1}
	\\
	\therefore
	(n-\orderedsetsize)^\correctsnum
	&
	\le
	\orderedsetsize^\correctsnum \cdot 2^{m+1}
\end{flalign*}
Taking logs we get
$$
\correctsnum
\leq
\frac{m+1}{\ln (n-\orderedsetsize) - \ln \orderedsetsize}
\approx
\frac{1}{\orderedsetsizeparam} \cdot \frac{m+1} {\ln n}. $$

As the code is prefix free, and from \Cref{lemma_compresion_by_d_bits}, we get that the probability over the choice of $\orderedset$ for any correct guess beyond $\frac{m+1} {\orderedsetsizeparam \ln n}$ drops exponentially, so the expected number of correct guesses cannot be larger than that.
By the above, we get that
$$
\expected_\orderedset [\gameval(\text{\dealerstatic, \guesserfixed}) | \guesserrandomness]
\le
{\guessermemorysymbol +1 \over \orderedsetsizeparam \ln n} + 2
.
$$
Recall that $\bestpermutationforguesser$ is the deck arrangement that ends with $\orderedset$ for which the guesser \guesserfixed makes the most correct guesses in the last $\orderedsetsize$ turns.
It follows that
$$
\expected_\orderedset \expected_{\pi \in \Pi_\orderedset} [\gameval(\text{\dealerstaticvanilla, \guesserfixed}) | \guesserrandomness]
\le
\expected_\orderedset [\gameval(\text{\dealerstatic, \guesserfixed}) | \guesserrandomness]
$$

Therefore, the above term upper bounds the expected number of correct guesses over the choice of~$\orderedset$, i.e.\
\begin{flalign*}
	\expected_{\permutation} [C | \guesserrandomness]
	=
	\expected_{\permutation} [\gameval(\dealersymbol_{\permutation}, \text{\guesserfixed}) | \guesserrandomness]
	=
	\expected_\orderedset \expected_{\permutation \in \Pi_\orderedset} [\gameval(\dealersymbol_{\permutation}, \text{\guesserfixed}) | \guesserrandomness]
	\le	
	\expected_\orderedset [\gameval(\text{\dealerstatic, \guesserfixed}) | \guesserrandomness]
	\le
	{\guessermemorysymbol +1 \over \orderedsetsizeparam \ln n} + 2
	.
\end{flalign*}
Since the expected number of correct guesses over the randomness of both the Guesser and the Dealer, is a convex combination of the above, we conclude that the expected number of correct guesses in the last $\orderedsetsize$ turns is at most
$$
	\correctsexpected = \expected_{\guesserrandomness, \permutation} [\correctsRV] \le {\guessermemorysymbol +1 \over \orderedsetsizeparam \ln n} + 2.
$$

Now, consider the expected number of correct guesses throughout the game, where the expectation is over the deck shuffle and the Guesser's randomness.
Suppose that the Guesser is perfect in the first $n-\orderedsetsize$ steps, in the sense that all the guesses are reasonable.
Then the expected number of correct guesses in the first turns is $H_n-H_\orderedsetsize = \orderedsetsizeparam \ln n$.
So we get that the total number of correct guesses is not expected to be better than $$\orderedsetsizeparam \ln n + \frac{m+1} {\orderedsetsizeparam \ln n} + 2.$$
Taking the best $\orderedsetsizeparam$ to be $\sqrt{m+1}/\ln n$, we get that this is not better than $2\sqrt{m+1} + 2$.

Note that this bound still holds even if the Guesser has at it disposal a large amount of randomness that it can repeatedly access (i.e.\ storing the randomness is not charged to the memory).
So we conclude with tight bounds up to constants:
\begin{theorem}
	There is a Guesser using $m$ bits of memory that obtains $1/2 \min\{\log n, \sqrt{m}\}$ correct guesses in expectation against the random-shuffle Dealer and any Guesser using $m$ bits of memory can get \emph{at most} $O(\min\{\ln n, \sqrt{m}\})$ correct guesses in expectation.
\end{theorem}

The same impossibility result also stands against the static Dealer.

\section{Adaptive Dealer}
\label{section_upper_bound}
We show that for every $m$ there exists an adaptive Dealer \dealerm such that every Guesser with~$\guessermemorysymbol$ memory bits is expected to make at most $\result$ correct guesses when playing against Dealer \dealerm.

Our proof is similar in structure to that in \Cref{sec_static_dealer:upper_bound} in showing that a too successful Guesser can be used to compress a random set.
We present our ``\mttbdealer'' in \Cref{section_upper_bound:subsec_mttb_strategy}.
We describe an encoding scheme for unordered sets (\Cref{section_upper_bound:subsec_encoding_scheme}) that utilizes reasonable guesses made against our Dealer in order to achieve a shorter description of unordered sets.
In \Cref{section_upper_bound:subsec_reasonable_guesses_upper_bound} we show that having too many reasonable guesses implies compression, i.e.\ descriptions that are too short, and we get that the expected number of reasonable guesses is bounded.
By bounding the expected number of reasonable guesses, we bound the expected number of correct guesses (\Cref{section_upper_bound:subsec_correct_guesses_upper_bound}).
We analyze the performance of the entire Dealer, as a whole, in \Cref{section_upper_bound:subsec_mttb_dealer}.

In \Cref{section_upper_bound:subsec_universal_dealer} we show a \emph{universal} adaptive Dealer that doesn't know how much memory the Guesser has, against which any Guesser with $\guessermemorysymbol$ bits of memory can score at most $\resultUniversal$.

\subsection{\mttbdealer}
\label{section_upper_bound:subsec_mttb_strategy}
Consider a game between some memory bounded Guesser and a Dealer who selects cards adaptively in an adversarial manner.
Assume that at some turn the Guesser makes an incorrect guess.
This guess may be incorrect because the Guesser had no luck, but it may also be incorrect because that card was played earlier and the Guesser did not recall that.

The idea is to use the Guesser's past guesses against her, and by doing so, forcing the Guesser to keep track of both past guesses and cards drawn.
We achieve this by making incorrect available card guesses undrawable for some turns, i.e.\ ``moving cards to the back of the deck''.
The Dealer we present begins the game with a properly shuffled deck, similarly to the random-shuffle Dealer.
At a certain turn the Dealer begins to ``move cards to the back'' and every once in a while the Dealer reshuffles the deck, making undrawable cards available again.
Towards the end of the game our Dealer makes one last reshuffle and draws cards one by one.

\newcommand{\epoch}[2]{$({#1},{#2})$-epoch\xspace}
\newcommand{\epochkl}{\epoch{k}{\ell}}
\newcommand{\epochkli}{\epoch{k_i}{\ell}}
\newcommand{\currentepoch}{\epochkli}
\newcommand{\currentuniepoch}{\epoch{k_\epochix}{\ell_\epochix}}

\newcommand{\epochstraterm}{MtBE-strategy\xspace}
\newcommand{\epochstraterms}{MtBE-strategies\xspace}
\newcommand{\epochstra}[3]{$({#1},{#2},{#3})$-MtBE-strategy\xspace}
\newcommand{\epochstraklu}{\epochstra{k}{\ell}{\limit}}
\newcommand{\epochstraklui}{\epochstra{k_i}{\ell_i}{\limit_i}}
\newcommand{\currentepochstra}{\epochstra{k_i}{\ell}{\limit}}
\newcommand{\currentuniepochstra}{\epochstra{k_\epochix}{\ell_\epochix}{\limit_\epochix}}

\newcommand{\epochs}[4]{$({#1},{#2},{#3},{#4})$-sequence-of-epochs-strategy\xspace}

\paragraph{Epochs and the \epochstraterm:}
The span of turns between reshuffles is called an \textbf{epoch} In particular, for $k \in [n]$, $\ell \in [k]$, the span of~$\ell$ turns that begins when~$k$ cards are left, and ends when~$k-\ell +1$ cards are left, is called a \epochkl.
We refer to applying the strategy of ``moving cards to the back'' during a span of turns (epoch) by \textbf{MtBE-strategy} (which stands for Move-to-the-Back Epoch strategy).

\begin{figure}[h]
	\begin{mdframed}
		\begin{flushleft}
			The axis specifies the number of cards left in the deck.
		\end{flushleft}

		\resizebox{\textwidth}{!}{%
			\begin{tikzpicture}[x=3cm]
				\draw
				(0,0) node[circle,fill,inner sep=0.5mm,label=above:{$n$},alias=n] (n) {}
				--
				(5,0) node[circle,fill,inner sep=0.5mm,label=below:{$k$},alias=k] (k) {};
				
				\draw
				(5.5,0) node[circle,fill,inner sep=0.5mm,label=below:{$k-\ell$},alias=kell] (kell) {}
				--
				(6,0) node[circle,fill,inner sep=0.5mm,label=above:{$0$},alias=0] (0) {};
				
				\draw [color=blue, thick] (k) -- (kell) {};
				
				\begin{scope}[thick,decoration={calligraphic brace, amplitude=6pt}]
					\draw[decorate]   (current bounding box.south) coordinate (aux) (kell|-aux) -- node[below=1ex]{\epochkl} (k|-aux);
				\end{scope}
			\end{tikzpicture}
		}
		\caption{\epochkl}
		\label{figure:epoch}
	\end{mdframed}
\end{figure}

\begin{definition}[\epochstraklu]
	\label{def_epochstraklu}
	Given $k \in [n], \ell \in [k]$ and $\limit \le \limitbound$,
	when~$\turnix$ cards are left s.t.\ $k \le \turnix \le k-\ell+1$:
	Let $A_\turnix$ be the set of $\turnix$ available cards,
	let $B'_\turnix$ be the set of reasonable guesses made by the Guesser since when there were $k$ cards in the back and until there are $\turnix$,
	let $\limit_\turnix = \min\{\limit, |B'_\turnix|\}$ and
	let~$B_\turnix$ be the set of the first $\limit_\turnix$ guesses from $B'_\turnix$.
	
	A Dealer that follows \epochstraklu, draws a card uniformly at random from the set $A_\turnix \setminus B_\turnix$ when $\turnix$ cards are left in the deck for $k \le \turnix \le k-\ell+1$.
\end{definition}
The upper bound $\limit$ is necessary to make sure that during any point in \epochkl the Dealer has cards to draw and that these cards are not too predictable.
Though implicit, this definition describes a reshuffle, as when $\turnix = k$ the set $B_\turnix$ is empty.
A procedural description of this strategy is specified in \Cref{alg_mttb_epoch}.

\begin{algorithm}[h]
	\caption{\epochstraklu}
	\label{alg_mttb_epoch}
	
	\begin{algorithmic}
		\Parameter
			$k \in [n]$,
			$\ell \in [k]$,
			$\limit \le \limitbound$,
			$A \subseteq {[n] \choose k}$
		
		\State
			$B \gets \emptyset$
		\For{$\turnix \in \{k, \dots, k-\ell+1\}$}
			\State
				Draw a card $c \in_\cR A \setminus B$ and discard $c$ from $A$
			\State
				$g \gets $ guess made by Guesser
			\If {$g \in A$ \textbf{and}  $|B| < \limit$}
				\Comment
				Move to the back
				\State
				$B \gets B  \cup \{g\}$
			\EndIf
		\EndFor
	\end{algorithmic}
\end{algorithm}

We notice that, when the Dealer follows a \epochstraklu, a guess is reasonable if it is available and being guessed for the first time in the current epoch, assuming no more than~$\limit$ cards were moved to the back.
We get that moving cards to the back works well against guessing techniques that repeat the same guess over and over.
Recall that the guessing techniques that were successful against a static Dealer (namely the Following-Subsets technique from \Cref{sec_static_dealer:sec_following_subset_guesser} and the Random-subsets technique from \Cref{sec_static_dealer:sec_random_subsets_guesser}) did exactly that.

Observe that the Dealer cannot move cards to the back for too many rounds, as cards will become too predictable as $A_\turnix \setminus B_\turnix$ shrinks.
Therefore, we apply the \epochstraterm in a sequence and reshuffle the deck at the beginning/end of each epoch.
Reshuffling the deck sets $B_\turnix$ to be the empty set again.
This is visualized in~\Cref{figure:mttb_with_reshuflle}.

As the Dealer refrains from drawing reasonably guessed cards during an epoch, a significant portion (if not all) of these cards would reside in the deck at the beginning of the following epoch.
Therefore, a Guesser can repeat her reasonable guesses from the previous epoch to get another chance, and most of these guesses will be reasonable.
Repeating reasonable guesses can be done either by generating a pseudorandom sequence of guesses from which some portion would be reasonable, or by tracking cards using memory.
We discuss this in detail after~\Cref{lemma_informal_expected_reasonable_guesses_upperbound}.

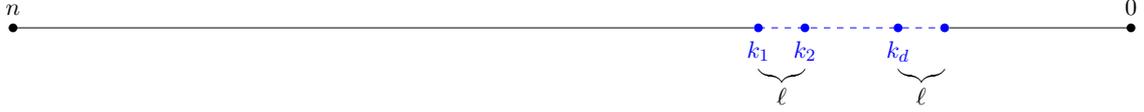
\begin{figure}[H]
\begin{mdframed}
	\begin{flushleft}
	In \textcolor{blue}{blue}: a sequence of $d$ epochs of length $\ell$ where the $i$th epoch begins when $k_i$ cards are left in the deck. The \textcolor{blue}{blue} dots indicates a reshuffle.
	\end{flushleft}
	\resizebox{\textwidth}{!}{%
		\begin{tikzpicture}[x=3cm]
			\draw [dashed,color=blue]
			(4,0) node[circle,fill,inner sep=0.5mm,label=below:{${k_1}$},alias=k_1] (k1) {}
			--
			(4.25,0) node[circle,fill,inner sep=0.5mm,label=below:{${k_2}$},alias=k_2] (k2) {}
			--
			(4.75,0) node[circle,fill,inner sep=0.5mm,label=below:{${k_d}$},alias=k_d] (kd) {}
			--
			(5,0) node[circle,fill,inner sep=0.5mm,alias=k_d1] (kd1) {};
			
			\draw
			(0,0) node[circle,fill,inner sep=0.5mm,label=above:{$n$},alias=n] (n) {}
			--
			(k1) {};
			
			\draw
			(kd1) {}
			--
			(6,0) node[circle,fill,inner sep=0.5mm,label=above:{$0$},alias=0] (0) {};
			
			\begin{scope}[thick,decoration={calligraphic brace, amplitude=6pt,mirror}]

				\draw[decorate]   (current bounding box.south) coordinate (aux) (k1|-aux) -- node[below=1ex]{$\ell$} (k2|-aux);
				\draw[decorate]   (current bounding box.south) coordinate (aux-1) (kd|-aux) -- node[below=1ex]{$\ell$} (kd1|-aux);

			\end{scope}
		\end{tikzpicture}
	}
	\caption{A sequence of \epochstraterms in epochs of equal lengths}
	\label{figure:mttb_with_reshuflle}
\end{mdframed}
\end{figure}

\newcommand{\koneepochs}{{n \over 8e \log n}}
\newcommand{\koneepochsend}{2 \guessermemorysymbol \log n}
\newcommand{\kfirstepochs}{{n \over 10e^2}}
\newcommand{\klastepochs}{{n \over 10e^2}}
\newcommand{\doneepochs}{({\firstepochbeginterm - 2m \cdot \log n) / \currentlength}}
\newcommand{\dfirstepochs}{{\sqrt{n} \over 10e^2 \sqrt{m}} - 2}
\newcommand{\dlastepochs}{\lceil 20e^2 \rceil}
\newcommand{\elloneepochs}{\guessermemorysymbol \cdot \log n}
\newcommand{\ellfirstepochs}{{\sqrt{nm}}}
\newcommand{\elllastepochs}{{2\sqrt{nm} - 4m \over \dlastepochs}}
\newcommand{\limitfirstepochs}{\sqrt{nm}}
\newcommand{\limitlastepochs}{\max\{2m, 2\log^2 n\}}

\newcommand{\firstuniepochbegin}{{n \over 8e \log^2 n}}
\newcommand{\duniepochs}{{\log_{\log n}\left(n \over 8e \log^6 n\right)}}

Finally, we present the Dealer, termed \mttbdealer, in  all her glory.
\begin{definition}[\mttbdealer]
	\label{def_mttb_dealer}
	Given $m \le \memorybound$, a \mttbdealer~\dealerm plays according to the strategy:
	\begin{enumerate}
		\item
			Shuffle the deck uniformly at random and draw cards one by one until $\koneepochs$ cards are left.
		\item
			Play the \epochstraterm $d$ times in a sequence, where $d = \doneepochs$, each epoch for $\ell = \elloneepochs$ turns, and move at most $\limit=\ell$ cards to the back during each epoch.
			Begin when $k_1= \koneepochs$ cards are left in the deck.
		\item
			When $2m \log n$ cards left, shuffle the deck and draw cards one by one for the rest of the game.
	\end{enumerate}
\end{definition}
A visual description of \dealerm is provided in \Cref{figure:dealer}.
As per scale, consider~\Cref{figure:mttb_in_scale}.
Note that \dealerm is computationally efficient.

\newcommand{\figkpoint}{5.0}
\begin{figure}[h]
\begin{mdframed}	
	\resizebox{\textwidth}{!}{%
		\begin{tikzpicture}[x=3cm]
			\draw [dashed,color=blue]
			(\figkpoint,0) node[circle,fill,inner sep=0.5mm,label=above:{$\koneepochsend$},alias=k_d1] (kd1) {}
			--
			(\figkpoint-1,0) node[circle,fill,inner sep=0.5mm,label=below:{${k_{d}}$},label=above:{$3 m \log n$},alias=k_d] (kd) {};
			
			\draw [dashed,color=blue]
			(\figkpoint-2,0) node[circle,fill,inner sep=0.5mm,label=below:{${k_2}$},alias=k_2] (k2) {}
			--
			(\figkpoint-3,0) node[circle,fill,inner sep=0.5mm,label=below:{${k_1}$},label=above:{$\koneepochs$},alias=k_1] (k1) {};
			
			\draw [dotted,color=blue] (k2) -- (kd) {};
			
			\draw [dotted]
			(0,0) node[circle,fill,inner sep=0.5mm,label=above:{$n$},alias=n] (n) {}
			--
			(k1) {};

			\draw [dashed]
			(5.75,0) node[circle,fill,inner sep=0.5mm,label=above:{$m$},alias=m] (m) {}
			--
			(6,0) node[circle,fill,inner sep=0.5mm,label=above:{$0$},alias=0] (0) {};
			
			\draw [dotted]
			(kd1) {}
			--
			(m) {};
			
			\begin{scope}[thick,decoration={calligraphic brace, amplitude=5pt}]
				\draw[decorate,color=blue]  (current bounding box.north) coordinate (aux) (k1|-aux) -- node[above=1ex,align=center]{Sequence of \epochstraterms} (kd1|-aux);

				\draw[decorate]  (current bounding box.north) coordinate (aux) (n|-aux) -- node[above=1ex]{Draw at random} (k1|-aux);
				\draw[decorate]  (current bounding box.north) coordinate (aux-1) (kd1|-aux) -- node[above=1ex]{Draw at random} (0|-aux);
				
			\end{scope}
			
			\begin{scope}[thick,decoration={calligraphic brace, amplitude=5pt,mirror}]
				
				\draw[decorate,color=blue]   (current bounding box.south) coordinate (aux) (k1|-aux) -- node[below=1ex]{$\elloneepochs$} (k2|-aux);

			\end{scope}
			\node[draw] at (5.4,-0.5) {Not drawn to scale};
			
		\end{tikzpicture}
	}
	\caption{Our Dealer \dealerm.}
	\label{figure:dealer}
	In \textcolor{blue}{blue}: Play the \epochstraterm in a sequence for $d$ times, where $d=\doneepochs$, each epoch for $\ell = \elloneepochs$ turns, begin when $\firstepochbegin = \koneepochs$ cards are left in the deck.
	
\end{mdframed}
\end{figure}

\newcommand{\tikzaxis}{6}
\begin{figure}
	\begin{mdframed}
		
		\begin{flushleft}
			In \textcolor{red}{red}, the span of turns during which the Dealer follows the \epochstraterm.
		\end{flushleft}
		\resizebox{\textwidth}{!}{%
			\begin{tikzpicture}[x=3cm]
				\draw
				(0,0) node[circle,fill,inner sep=0.3mm,label=below:{first turn},alias=n] (n) {}
				--
				(\tikzaxis*0.9822,0) node[circle,fill,inner sep=0.3mm,label=above:{$\koneepochs$},alias=k1] (k1) {}
				--
				(\tikzaxis*0.996,0) node[circle,fill,inner sep=0.3mm,label=below:{},alias=k2] (k2) {}
				--
				(\tikzaxis,0) node[circle,fill,inner sep=0.3mm,label=below:{last turn},alias=0] (0) {};
				\draw [color=red, thick] (k1) -- (k2) {};

			\end{tikzpicture}
		}
		\caption{\mttbdealer roughly to scale.}
		\label{figure:mttb_in_scale}
	\end{mdframed}
\end{figure}

We refer to the turn at which the first epoch begins as $n-\firstepochbegin$.
Observe that for every epoch in the sequence played by our Dealer we get that $\currentlength \le \currentepochbegin - \currentlength$ so we can set the maximal number of cards moved to the back $\limit$ to $\currentlength$.

The main theorem of this section states that \mttbdealer works well against any memory bounded guesser.
\begin{theorem}	
	For any $\guessermemorysymbol$, every Guesser with~$m$ bits of memory is expected to make at most $$\result$$ correct guesses when playing against the \mttbdealer~\dealerm~(\Cref{def_mttb_dealer}).
\end{theorem}
Thinking about this theorem, it is clear that a Guesser with $m$ memory bits can easily achieve $\ln m$ correct guesses in expectation by using the simple Subset Guessing strategy (from \Cref{sec_basic_guessing_strategies}).
So essentially, this theorem states that moving cards to the back and reshuffling every once in a while, is a \emph{very} effective strategy against a memory bounded Guesser.

\subsection{Towards a proof}
\label{section_upper_bound:subsec_towards_a_proof}
Consider $\guessermemorysymbol$ and \mttbdealer \dealerm.
Let~$\guesserrandomness$ be the Guesser's randomness and let~$\dealerrandomness$ be the Dealer's randomness.
Let~$\currentreasonableRV$ denote the number of reasonable guesses made during the $\epochix$th epoch.
Let~$\currentreasonableexpected$ be the expectation of $\currentreasonableRV$ taken over the Guesser's and the Dealer's randomness, i.e.
$$
\currentreasonableexpected = \expected_{\guesserrandomness, \dealerrandomness}\left[\currentreasonableRV \right].
$$

To prove that our Dealer works well against any memory bounded Guesser we analyze the reasonable guesses during a single epoch.
We claim that no Guesser can expect to make too many reasonable guesses during \emph{any} of the epochs while our Dealer follows the \epochstraterm.
Consider the \currentepoch where the Dealer follows the \epochstraterm.

\begin{lemma}[Informal]
	\label{lemma_informal_expected_reasonable_guesses_upperbound}
	A Guesser with $\guessermemorysymbol$ bits of memory that plays against a \mttbdealer~\dealerm is expected to make at most
	$$
	\currentreasonableexpected \le \reasonablesboundi + 2.
	$$
	reasonable guesses during any \currentepoch played by the \dealerm.
\end{lemma}

Observe that when $\firstepochbegin$ cards are left, the probability that a random guess is a card that is still in the deck is $\firstepochbegin \over n$.
By linearity of expectation we get that guessing randomly for $\currentlength$ turns would yield at most ${\firstepochbegin \currentlength \over n}$ reasonable guesses in expectation.
These cards are a reasonable guess exactly once during each epoch, as the Dealer avoids drawing them, but for the same reason it follows that a significant portion of them would still be available (and reasonable) in the next epoch.
So we get that by using the same set of random guesses in each epoch the Guesser can get near the claimed upper bound of reasonable guesses.
On the other hand, with carefully managed $\guessermemorysymbol$ bits, it may be possible in some cases to keep track of $\guessermemorysymbol$ cards that have not appeared (as we did in the Subset guessing technique in \Cref{sec_basic_guessing_strategies}).
So essentially, this lemma states that any memory bounded Guesser that plays against our Dealer cannot do much better then guessing cards at random or tracking $\guessermemorysymbol$ cards.

At any turn, the Dealer's strategy determines a distribution to sample a card from.
In our case, this distribution is uniform on the available cards that were not moved to the back.
We think of the Dealer as using precedence represented by a permutation $\pi$ in order to make this choice:

\begin{definition}[\minorderphrase]
Given a permutation $\permutation \in S_n$ and a set $C \subseteq [n]$ we say that a card $x \in C$ is the \minorderpi card from $C$  if $x$  is the element of $C$ with the smallest $\pi$ value.
\end{definition}

We describe the randomness $\dealerrandomness$ of the Dealer in an indirect way:
the Dealer has two independent parts for its randomness,
(i)  a sequence of permutations $\permutationsfull$ and (ii) a set $\encsetdesc$.
The way they are used is:
\begin{itemize}
	\item
	For turn $1 \le \turnix \le n-\firstepochbegin$ the Dealer draws the \minorderpit card from $A_\turnix \setminus \encset$.
	\item
	For turn $n-\firstepochbegin+1 \le \turnix \le n$ the Dealer draws \minorderpit card from $A_\turnix \setminus B_\turnix$.
\end{itemize}
As the \mttbdealer draws cards uniformly at random at the first~$n-\firstepochbegin$ turns, it follows that every set $\encset$ can be kept for the last $\firstepochbegin$ turns, so this is well defined.

\newtheorem{observation}[theorem]{Observation}
\begin{observation}
If the sequence of permutations $\permutationsfull$  and a set $\encsetdesc$ are chosen uniformly at random, then this implementation is equivalent to Definition~\ref{def_mttb_dealer}.
\end{observation}

Note that it was important to choose a permutation $\pi_\turnix$ independently for each turn $\turnix$, since a common permutation $\pi$ for all turns might leak information regarding the relative ranking of cards that were moved to the back at different times during an epoch.
In particular, for any two reasonable guesses during some epoch, the earlier one has a higher probability of preceding the latter.
As a result, the earliest reasonably guessed card has a higher probability of being drawn at the first turn in the following epoch, i.e., cards are drawn in a non-uniform manner.

We conclude that:
$$
	\currentreasonableexpected
	=
	\expected_{\guesserrandomness, \dealerrandomness}\left[\currentreasonableRV \right]
	=
	\expected_{\guesserrandomness, \permutations, \encset}\left[\currentreasonableRV \right]
	=
	\sum_{\guesserrandomness, \permutations} \expected_{\encset}\left[\currentreasonableRV | \guesserrandomness, \permutations \right] \cdot \Pr\left[\guesserrandomness, \permutations\right]
	.
$$
The next two sections are dedicated to bounding the term
$$
	\expected_{\encset}\left[\currentreasonableRV | \guesserrandomness, \permutations \right]
$$
for any permutations sequence $\permutations$ and any Guesser's randomness $\guesserrandomness$.

\subsection{Encoding scheme}
\label{section_upper_bound:subsec_encoding_scheme}

In order to bound the expected number of reasonable guesses in a single epoch we present an encoding scheme for subsets of $[n]$ of size $\firstepochbegin$.
The encoding scheme utilizes reasonable guesses against our Dealer during any of the epochs to achieve a shorter description.

Consider some Guesser \guesser that plays against the \mttbdealer \dealerm (\Cref{def_mttb_dealer}) and fix one of the epochs \currentepoch played by \dealerm.
For every triplet $\guesserrandomness, \permutationsfull, \encset$ we associate the Guesser \guesserfixed with fixed randomness~$\guesserrandomness$ and the Dealer \dealerdet with fixed randomness that corresponds to~$\permutationsfull$ and~$\encset$ as the last cards to be played.

Fix some prefix-free code (\Cref{def_prefix_free_code}) for sets of size $\firstepochbegin$.
The encoding scheme will use this code for the cases when there are not enough reasonable guesses to utilize.

\begin{definition}[$\encodefuncadaptive$]
	To encode a set $\encsetdesc$, the function $\encodefuncadaptive$ simulates a game between Guesser \guesserfixed and Dealer \dealerdet.
	
	Denote by $\turnssetsymbol'$ the set of turns during the \currentepoch at which \guesserfixed made a reasonable guess, i.e.\ $$\turnssetsymbol' = \{n-\currentepochbegin \le \turnix \le n -\currentepochbegin+\currentlength-1 | g_\turnix \text{ is a reasonable guess}\}.$$
	
	\begin{itemize}
		\item
			If $|\turnssetsymbol'| < \reasonablesnum$, then the code is made of an indicator bit set to $0$ and an explicit prefix free representation of $\encset$.
		\item
			If $|\turnssetsymbol'| \ge \reasonablesnum$, then let $\turnssetsymbol$ be the first $\reasonablesnum$ turns from $\turnssetsymbol'$ during which the Guesser guessed reasonably during the \currentepoch.
			The code is made of:
			\begin{enumerate}
				\item
					Indicator bit set to $1$ ($1$ bit).
				\item
					Guesser's memory state $\memorystatesymbol$ at turn $n-\firstepochbegin$ ($m$ bits).
				\item
					Description of $\turnssetsymbol$, the first $\reasonablesnum$ turns at which the Guesser guessed reasonably during the \currentepoch ($\log {\currentlength \choose \reasonablesnum}$ bits).
				\item
					Binary vector $\correctvector$ of length~$\reasonablesnum$ that tags which of the reasonable guesses described in~$\turnssetsymbol$ were also correct ($\reasonablesnum$ bits\footnote{Though a shorter representation is possible, it suffices for our purpose.}).
				\item
					Description of $\encset \setminus \{g_\turnix | \turnix \in \turnssetsymbol\}$  ($\log {n \choose \firstepochbegin - \reasonablesnum}$ bits).
			\end{enumerate}
	\end{itemize}
	
\end{definition}
A visual representation of the stored information is provided in \Cref{figure:adaptive_dealer_upper_bound_encoding_full}.
\newcommand{\nottoscalewarn}{Not to scale (for scale consider~\Cref{figure:mttb_in_scale}).}
\begin{figure}[h]
	\begin{mdframed}	
		\resizebox{\textwidth}{!}{%
			\input{parts/figures/reasonable_guess_encoding_full.tex}
		}
		\caption{$\encodefuncadaptive$. Colored - information stored by the encoding scheme.} 
		\label{figure:adaptive_dealer_upper_bound_encoding_full}
	\end{mdframed}
\end{figure}

We first define the Dealer we will use during the decoder simulation.
Since we simulate this Dealer, we can break the usual course of the game.
In particular, we will assume that the Dealer begins playing the game at the middle and with a partial deck, that the Dealer is aware of the Guesser's guess \emph{before} placing a card and that the Dealer can place cards that do not reside in the deck.

\begin{definition}
	\label{def_dealer_det_decode}
	The Dealer \dealerdetdec plays according to the \epochstraterm, as configured for \dealerm (\Cref{def_mttb_dealer}) with few modifications:
	\begin{itemize}
		\item
			When $\turnix$ cards are left for $\turnix \in \{\firstepochbegin, \dots, \currentepochbegin +1 \}$: 
			play according to the \epochstraterm with deck~$\encodedexplicit$ moving cards to the back when they are guessed and still in the deck.
		\item
			When $\turnix$ card are left, for $\turnix \in \{\currentepochbegin, \dots, \currentepochbegin - \currentlength + 1\}$, and until the $\reasonablesnum$th reasonable guess: if turn~$\turnix$ is tagged both as reasonable (by~$\turnssetsymbol$) and correct (by~$\correctvector$), then draw the card $g_\turnix$ guessed by the Guesser. Otherwise draw the \minorderpit card from the deck that was not moved to the back.
		\item
			Once $\reasonablesnum$ reasonable guesses occurred during the $\epochix$th epoch,  stop playing.
	\end{itemize}
\end{definition}
A procedural description of the Dealer \dealerdetdec is specified in~\Cref{alg_dealer_det_decode}.
\begin{algorithm}[h]
	\caption{Behavior of \dealerdetdec while simulated by $\decodefuncadaptive$}
	\label{alg_dealer_det_decode}
	\begin{algorithmic}

		\Parameter
			$\encodedexplicit \subseteq {[n] \choose k-\reasonablesnum}$,
			$\permutationsfull$,
			$\turnssetsymbol$,
			$\correctvector$,
			$\reasonablesnum$,
			$\epochix$
		
		\For {$j \in [\epochix]$}
			\Comment
				For epochs prior to $\epochix$
			\State
				$B \gets \emptyset$
			\For{$\turnix \in \{k_j, \dots, k_j-\ell_j+1\}$}
				\State
					Draw the \minorderpit card from $\encodedexplicit \setminus B$ and discard from $\encodedexplicit$
				\State
					$g_\turnix \gets $ guess made by Guesser
				\If {$g_\turnix \in \encodedexplicit$ \textbf{and}  $|B| < \limit$}
					\Comment
						Move to the back
					\State
						$B \gets B  \cup \{g_\turnix\}$
				\EndIf
			\EndFor
		\EndFor
		\State
			$B \gets \emptyset$
			\Comment
				$\epochix$th epoch
		\For{$\turnix \in \{\currentepochbegin, \dots, \currentepochbegin-\currentlength+1\}$}
			\State
				$g_\turnix \gets $ guess made by Guesser
			\If {$g_\turnix$ is tagged as both reasonable and correct (according to $\turnssetsymbol$ and $\correctvector$)}
				\State
					$d_\turnix \gets g_\turnix$
			\Else
				\State
					$d_\turnix \gets$ \minorderpit card from $\encodedexplicit \setminus B$ and discard from $\encodedexplicit$
			\EndIf
			\State
				Draw $d_\turnix$
			\If {$g_\turnix$ is tagged as reasonable (according to $\turnssetsymbol$)}
				\Comment
					Move to the back
				\State
					$B \gets B \cup \{g_\turnix\}$
			\EndIf
			\If {$|B| = \reasonablesnum$}
				\State
					Stop playing
			\EndIf
		\EndFor
	\end{algorithmic}
\end{algorithm}

Note that after turn $n-\firstepochbegin$, the Dealer simulated by the \emph{encoder} recognizes a guess as reasonable if it is from $\encset$ and wasn't drawn yet.
But when the Dealer simulated by the \emph{decoder} observes a guess of a card not in $\encodedexplicit$, then the Dealer cannot tell whether this is a card that will turn to be reasonable at the \currentepoch or a card that has been played before turn $n-\firstepochbegin$.
Therefore, the Dealer simulated by the decoder recognizes a guess as reasonable (and moves it to the back) if it is from $\encodedexplicit$.
We will soon see that this behavior allows the decoder to reproduce the same game transcript, and by doing so, decode the original set.

We now define the decode function.
\begin{definition} [$\decodefuncadaptive$]
	If the indicator bit is $0$ then decode a set from the remaining bits in the natural way.
	Otherwise, the function parses the remaining bits as a $4$-tuple $(\memorystatesymbol,  \correctvector, \turnssetsymbol, \encodedexplicit)$ as encoded by $\encodefuncadaptive$, and then simulates and records a partial game between the Dealer~\dealerdetdec (see \Cref{def_dealer_det_decode}) and the Guesser \guesserfixed:
	\begin{itemize}
		\item
			Initialize Guesser \guesserfixed with memory state $\memorystatesymbol$ at turn $n-\firstepochbegin$ and simulate a game against \dealerdetdec until the $\reasonablesnum$th reasonable guess in the \currentepoch.
		\item
			Let $\encodedrecovered$ be the set of card guesses that were tagged as reasonable (by $\turnssetsymbol$).
		\item
			Output $\decset = \encodedexplicit \cup \encodedrecovered$.
	\end{itemize}
\end{definition}

A procedural description of $\decodefuncadaptive$ is specified in \Cref{alg_decode}.
\begin{algorithm}[H]
	\caption{$\decodefuncadaptive$}
	\label{alg_decode}
	\begin{algorithmic}
		\Parameter
			$\currentepochbegin \in [n]$,
			$\currentlength \in [\currentepochbegin]$,
		 	$\permutationsfull$,
		 	$\guesserrandomness$
		\Input
			$x \in \zo^\ast$
		\If {$x_1 = 0$}
			\State
				Parse $\decset$ from $x$
			\State
				\Return $\decset$
		\EndIf
		\State
			Parse $
			\memorystatesymbol \in \zo^\guessermemorysymbol,
			\correctvector \in \zo^\reasonablesnum,
			\turnssetsymbol \in {[\currentlength] \choose \reasonablesnum},
			\encodedexplicit \in {[n] \choose \firstepochbegin - \reasonablesnum}
			$
			from $x$
		\State
			Initialize Guesser \guesserfixed at turn $n-\firstepochbegin$ with memory state $\memorystatesymbol$.
		\State
			$\encodedrecovered \gets \emptyset$
		\For{$\turnix \in \{\firstepochbegin, \dots, \currentepochbegin+1\}$}
			\State
				Simulate a turn between \guesserfixed and \dealerdetdec and update Guesser's memory accordingly
		\EndFor
		\For{$\turnix \in \{\currentepochbegin, \dots, \currentepochbegin-\currentlength+1\}$}
			\State
				Simulate a turn between \guesserfixed and \dealerdetdec and update Guesser's memory accordingly
			\State
				$g_\turnix \gets$ guess made by Guesser \guesserfixed
			\If {$g_\turnix$ is tagged as reasonable (according to $\turnssetsymbol$)}
				\State
					$\encodedrecovered \gets \encodedrecovered \cup \{g_\turnix\}$
			\EndIf
		\EndFor
		\State
			\Return $\encodedexplicit \cup \encodedrecovered$
	\end{algorithmic}
\end{algorithm}
We assume that both the Dealer's precedence $\permutationsfull$ and the Guesser's randomness $\guesserrandomness$ are given to us ``for free'' and are known during encoding and decoding of the set $\encset$.
We justify this assumption in two different ways:
\begin{itemize}
	\item
		Fixing $\guesserrandomness$ and $\permutationsfull$ (i.e.\ fix the Dealer's random source for the order) we can consider an encoding scheme for sets $\encsetdesc$.
	\item
		We can assume that we encode a triplet $(\guesserrandomness, \permutationsfull, \encset)$ where $\guesserrandomness$ and $\permutation$ are written explicitly in some natural way right next to $\encodefuncadaptive(\encset)$.
\end{itemize}

\begin{claim}
	\label{lemma_is_encoding_scheme}
	For all $\encsetdesc$: $$\decodefuncadaptive(\encodefuncadaptive(\encset)) = \encset.$$
\end{claim}
\begin{proof}
	For the case that the set is encoded explicitly, this is trivial.
	
    Recall that both $\encodefuncadaptive$ and $\decodefuncadaptive$ work by simulating a partial game between~\guesserfixed and \dealerdet, \dealerdetdec (resp.).
	Denote by $g_\turnix, d_\turnix$ the card guessed by the Guesser and the card drawn by Dealer (resp.) at turn $\turnix$ during the game simulated by $\encodefuncadaptive$, and equivalently $g'_\turnix, d'_\turnix$ those simulated by $\decodefuncadaptive$.
	
	We first prove by induction that the same game transcript is produced by both simulations, i.e.\ that for every turn $n-\firstepochbegin \le \turnix < n-\currentepochbegin + \currentlength$ it holds that $g_\turnix = g'_\turnix$ and $d_\turnix = d'_\turnix$.
	
	We highlight two facts:
	\begin{enumerate}
		\item
			If both simulated Guessers have the same memory state at the beginning of some turn, they guess the same card at that turn.
			This is true because both simulated Guessers also have the same fixed randomness $\guesserrandomness$.
		\item
			If both simulated Guessers have the same memory state at the beginning of some turn, and both simulated Dealers draw the same card at that turn, then both Guessers begin the next turn with the same memory state.
	\end{enumerate}
	At the beginning of turn $n-\firstepochbegin$, both simulated Guessers have the same memory state.
	Therefore, to prove that the same transcript is simulated, it suffice to prove that both Dealers draw the same card in every turn.

	Without loss of generality, assume that the \currentepoch is not the first epoch.
	
	\paragraph{Base:}
	The encoder simulated Dealer \dealerdet draws the \minorderphrase card $d_{n-\firstepochbegin}$ from~$\encset$.
	Since $d_{n - \firstepochbegin}$ was drawn before the \currentepoch it is also in $\encodedexplicit$, thus it is the \minorderphrase card from $\encodedexplicit$ as well, so it is drawn by \dealerdetdec and we get that $d_{n-\firstepochbegin} = d'_{n-\firstepochbegin}$.

	\paragraph{Step for turn $n-\firstepochbegin < \turnix < n-\currentepochbegin$:}
	Since turn $n-\firstepochbegin$ the Dealer simulated by the encoder \dealerdet recognizes an available guess as reasonable if it is from $\encset$ while the Dealer simulated by the decoder \dealerdetdec recognizes it as reasonable if it is from $\encodedexplicit$.
	By the induction hypothesis, the same transcript was simulated so far by the encoder and decoder, and since $\encodedexplicit \subseteq \encset$ it follows that the back of \dealerdetdec is always a subset of the back of \dealerdet.
	In particular, the {\emph only} cards that are missing from the back of \dealerdetdec are those who will turn to be  reasonable in the \currentepoch.
	
	Since $d_\turnix$ was drawn before $n-\currentepochbegin$ then $d_\turnix \in \encodedexplicit \subset \encset$ (it cannot be in $\encset  \setminus \encodedexplicit$, since this would mean that it could not be a candidate for reasonable guess in  \currentepoch).
	Since $d_\turnix$ is the \minorderpit card from $\encset$ then $d_\turnix$ is also the \minorderpit card from $\encodedexplicit$.
	As a result, $d'_\turnix = d_\turnix$.
 	
 	\paragraph{Step for turn $n-\currentepochbegin \le \turnix < n-\currentepochbegin+\currentlength$:}
	The same argument holds as before, but we need to clarify the case where a correct guess was made.
	If the encoder simulated Guesser guessed correctly then $g_\turnix = d_\turnix$ and $d_\turnix \notin \encodedexplicit$.
	In that case, the turn $\turnix$ is tagged as both reasonable (by $\turnssetsymbol$) and correct (by $\correctvector$).
	Therefore, $d'_\turnix = g'_\turnix$ because that's how \dealerdetdec works.
	Since $g'_\turnix = g_\turnix$ we get that $d'_\turnix = d_\turnix$.
	\newline

	Finally, since the same partial game is simulated, the same guesses are tagged as reasonable during the $\epochix$th epoch (by $\turnssetsymbol$) and $\encodedrecovered = \{g'_\turnix | \turnix \in \turnssetsymbol\} = \{g_\turnix | \turnix \in \turnssetsymbol\}$.
	We finish the proof by observing that
	$$
	\encset
	=
	(\encset \setminus \{g_\turnix | \turnix \in \turnssetsymbol\}) \cup \{g_\turnix | \turnix \in \turnssetsymbol\}
	=
	\encodedexplicit \cup \encodedrecovered
	=
	\decodefuncadaptive(\encodefuncadaptive(\encset))
	.
	$$
\end{proof}

\begin{claim}
	\label{claim_code_is_prefix_free_adaptive}
	The code produced by $\encodefuncadaptive$ is prefix-free.
\end{claim}
\begin{proof}
	Let two sets $\encset, \encset' \in\encsetdomain$ and denote by $I, I'$ the first bit in the encoding of $\encset, \encset'$ (resp.).
	Observe that if $I \neq I'$ then the two descriptions cannot be a prefix of one another, if $I=I'=0$ then this is trivial and if $I=I'=1$ then the two descriptions are of the same length.
	Therefore, to prove that $\encodefuncadaptive$ produce a prefix-free code, it suffice to show that for every $\encset$ it holds that $\decodefuncadaptive(\encodefuncadaptive(\encset)) = \encset$.
	By \Cref{lemma_is_encoding_scheme}, this is indeed the case so the \namecref{claim_code_is_prefix_free_adaptive} follows.
\end{proof}

\newcommand{\wordlength}{w}
\newcommand{\wordlengthterm}{\wordlength(m, \firstepochbegin, \currentlength, \reasonablesnum)}
\newcommand{\wordlengthtermdiff}{\wordlength(m, \firstepochbegin, \currentlength, \reasonablesnum+\reasonablesdiff)}
\newcommand{\wordlengthtermdiffone}{\wordlength(m, \firstepochbegin, \currentlength, \reasonablesnum+1)}

Denote by $\encdomainsymbol_\reasonablesnum \subseteq \encsetdomain$ the collection of all sets $\encsetdesc$ such that \guesserfixed makes at least $\reasonablesnum$ reasonable guesses against \dealerdet during the \currentepoch.
Observe that all sets in $\encdomainsymbol_\reasonablesnum$ are encoded by $\encodefuncadaptive$ to bit strings of the same length.
Denote by $\wordlength(m, \firstepochbegin, \currentlength, \alpha)$ the length in bits of $\encodefuncadaptive(\encset)$ for~$\encset \in \encdomainsymbol_\alpha$.

\newcommand{\codewordlengthterm}{2 \cdot 2^\guessermemorysymbol \cdot 2^\reasonablesnum \cdot {n \choose \firstepochbegin-\reasonablesnum} \cdot {\currentlength \choose \reasonablesnum}}
\newcommand{\codewordlength}{\log\left(  \codewordlengthterm \right)}
\begin{corollary}
	\label{cor_encoding_length}
	If \guesserfixed makes at least $\reasonablesnum$ reasonable guesses against the Dealer \dealerdet during the \currentepoch, then the length of $\encodefuncadaptive(\encset)$ is
	\begin{flalign*}
		\wordlengthterm
		&
		\triangleq
		\codewordlength
		.
	\end{flalign*}
		
\end{corollary}
Looking at the term from this corollary, it can be seen that for every increase in $\reasonablesnum$ we save roughly $\log n$ bits and pay roughly $1+ \log \currentlength$ bits.
In our Dealer (\Cref{def_mttb_dealer}) each epoch consists of $\currentlength = m \log n$ turns, so we get that we expect to save order of $\log n - \log m$ bits for every reasonable guess.
We analyze this in detail in \Cref{lemma_length_decrease_by_one_bit}.

\subsection{Upper Bound on the Number of Reasonable Guesses}
\label{section_upper_bound:subsec_reasonable_guesses_upper_bound}
The function $\encodefuncadaptive$ yields descriptions of lengths that varies with $\reasonablesnum$.
For example, for $\reasonablesnum = 0$, we get that in every simulation there are at least $0$ reasonable guesses, so all descriptions are of length $\log\left(2 \cdot 2^\guessermemorysymbol {n \choose \firstepochbegin} \right)$, which is much larger than storing the set explicitly.
The following two claims deals with the encoding length.
We first claim that making more reasonable guesses implies shorter descriptions, and more specifically, that each additional guess saves at least a bit.

\begin{claim}[Length decreases with reasonable guesses]
	\label{lemma_length_decrease_by_one_bit}
	For every~$m,\firstepochbegin \in [n/4]$, $\currentlength \in [\firstepochbegin]$, and every ${4 \firstepochbegin \currentlength \over n} \le \reasonablesnum$ it holds that
	$$\wordlengthtermdiffone \le \wordlengthterm - 1.$$
\end{claim}
\begin{proof}
	\label{prf_length_decrease_by_one_bit}
	We will calculate $\wordlengthterm - \wordlengthtermdiffone$ and show that it is at least~$1$.

	\begin{flalign*}
		\log\left(2 \cdot 2^m \cdot 2^\reasonablesnum \cdot {n \choose \firstepochbegin-\reasonablesnum} \cdot {\currentlength \choose \reasonablesnum}\right)
		&
		-
		\log\left(2 \cdot 2^m \cdot 2^{\reasonablesnum + 1}\cdot {n \choose \firstepochbegin-\reasonablesnum-1} \cdot {\currentlength \choose \reasonablesnum+1}\right)
		\\
		&
		=
		\log\left(
		{1 \over 2} \cdot
		{{n \choose \firstepochbegin-\reasonablesnum}  \over {n \choose \firstepochbegin-\reasonablesnum-1} } \cdot {{\currentlength \choose \reasonablesnum} \over {\currentlength \choose \reasonablesnum+1}}
		\right)
		\\
		&
		=
		\log\left(
		{1 \over 2} \cdot
		{n -\firstepochbegin+\reasonablesnum+1  \over \firstepochbegin-\reasonablesnum} \cdot {\reasonablesnum+1 \over \currentlength -\reasonablesnum}
		\right)
		\\
		&
		>
		\log\left(
		{n \reasonablesnum \over 2 \firstepochbegin \currentlength}
		\right)
		.
	\end{flalign*}
	By assumption $\reasonablesnum \ge {4 \firstepochbegin \currentlength \over n}$ so the term inside the $\log$ is larger than $2$.
\end{proof}
\begin{corollary}
	\label{lemma_length_decrease_by_diff_bits}
	For every~$m, \firstepochbegin \in [n/4]$, $\currentlength \in [\firstepochbegin]$,
	and every ${4 \firstepochbegin \currentlength \over n} \le \reasonablesnum < \currentlength$ and for every~$\reasonablesdiff \in [\currentlength - \reasonablesnum]$ it holds that
	$$\wordlengthtermdiff \le \wordlengthterm - \reasonablesdiff.$$
\end{corollary}

Consider a random set $\encset$ chosen uniformly at random from $\encsetdomain$.
The entropy of $\encset$ is
$$H(\encset) = \log \left|{[n] \choose \firstepochbegin}\right| = \log {n \choose \firstepochbegin}.$$
The next \namecref{lemma_encoding_achieves_compression} describes the amount of reasonable guesses required to achieve compression, i.e.\ descriptions of length shorter than the entropy of a random input.
\begin{claim}[Encoding achieves compression]
	\label{lemma_encoding_achieves_compression}
	For every~$m < n, \firstepochbegin \in [n/8e]$, $\currentlength \in [\firstepochbegin]$,
	if
	$$\reasonablesnum \ge \reasonablesboundi$$
	then
	$$
	\wordlengthterm < \log {n \choose \currentepochbegin}.
	$$
\end{claim}
\begin{proof}
	\label{proof_encoding_achieves_compression}
	We will show that $\log {n \choose \firstepochbegin} - \wordlengthterm \ge 0$.
	To do so, we will use an upper bound on binomial coefficients derived from Stirling's approximation (${b \choose a} < \left(be \over a\right)^a$) and the fact that
	${n \choose k} /{n \choose k -a} = {(n-k+a)\dots(n-k+1) \over (k) \dots (k-a+1)} > \left(n -k+ a \over k\right)^a > \left(n  \over k\right)^a$.
	
	\begin{flalign*}
		\log {n \choose \firstepochbegin} &- \codewordlength
		\\
		&
		=
		\log\left(
		{
			{n \choose \firstepochbegin}
			\over
			{n \choose \firstepochbegin-\reasonablesnum}
		}
		\cdot
		{1 \over 2^{\guessermemorysymbol + 1}}
		\cdot
		{1 \over 2^\reasonablesnum}
		\cdot
		{1 \over {\currentlength \choose \reasonablesnum} }
		\right)
		\\
		&
		>
		\log\left(
		\left(
			n \over \firstepochbegin
		\right)^\reasonablesnum
		\cdot
		{1 \over 2^{\guessermemorysymbol + 1}}
		\cdot
		{1 \over 2^\reasonablesnum}
		\cdot
		\left(\reasonablesnum \over \currentlength e\right)^\reasonablesnum
		\right)
		\\
		&
		=
		\log\left(
		\left(
		n \reasonablesnum \over 2 e \firstepochbegin \currentlength
		\right)^\reasonablesnum
		\cdot
		{1 \over 2^{\guessermemorysymbol + 1}}
		\right)
		\\
		&
		>
		\log\left(
		\left(
		n \reasonablesnum \over 4 e \firstepochbegin \currentlength
		\right)^\reasonablesnum
		\cdot
		{1 \over 2^{\guessermemorysymbol}}
		\right)
		.
	\end{flalign*}
	By assumption, $\reasonablesnum \ge \reasonablesboundtermi$ and $\reasonablesnum \ge \guessermemorysymbol$, so we get that the term inside the $\log$ is greater than~$1$, so the \namecref{lemma_encoding_achieves_compression} follows.

\end{proof}

Combining the above claims we get that no Guesser can make too many reasonable guesses against our Dealer.
Recall the random variable $\currentreasonableRV$ that denotes the number of reasonable guesses that a Guesser makes during the \currentepoch.
Consider the sequence of epochs $\{\text{\currentepoch}\}_{i=1}^\epochsnum$ played by the \mttbdealer (from \Cref{def_mttb_dealer})~\dealerm.
\begin{claim}
	\label{claim_probability_for_reasonables_decays}
	For every Guesser \guesser with $\guessermemorysymbol$ bits of memory, and for every epoch $\epochix \in [d]$, for $\reasonablesdiff < \currentlength - \reasonablesboundi$,
	the probability that \guesser makes more than $\reasonablesdiff + \reasonablesboundi$ reasonable guesses during the \currentepoch, is at most
	$$
		\Pr_{\encsetdesc}\left[\currentreasonableRV \ge \reasonablesboundi + \reasonablesdiff \right] \le 2^{-\reasonablesdiff}
		.
	$$
\end{claim}
\begin{proof}
	Fix some $\guesserrandomness$ and $\permutationsfull$ and consider the probability conditioned on these choices of randomness.
	Let $\reasonablesnum = \reasonablesdiff+ \reasonablesboundi$.

	If a Guesser with fixed randomness \guesserfixed makes at least $\reasonablesnum$ reasonable guesses against our \mttbdealer with fixed randomness \dealerdet during the \currentepoch, then the set $\encset$ is encoded by $\encodefuncadaptive$ into a string of $\wordlengthterm$ bits.

	Since $\reasonablesnum > \reasonablesboundi$, from \Cref{lemma_encoding_achieves_compression}, the length of $\encodefuncadaptive(\encset)$ is shorter than the entropy of a random $\encset$ sampled uniformly from~$\encsetdomain$.
	Since $\reasonablesnum > {4 \firstepochbegin \currentlength \over n}$, from \Cref{lemma_length_decrease_by_diff_bits} we get that the length of $\encodefuncadaptive(\encset)$ is at least $\reasonablesdiff$ bits shorter than the entropy of a random $\encset$.
	\Cref{lemma_compresion_by_d_bits} tells us that the probability to encode a random element by $\beta$ bits below the entropy decays exponentially, i.e.
	\begin{flalign*}
		\Pr_{\encset \in\encsetdomain}\left[ \currentreasonableRV \ge \reasonablesboundi + \beta  \right]
		&
		\le
		\Pr_{\encset \in\encsetdomain}\left[ \encodefuncadaptive( \encset) = H(\encset) - \beta \right]
		\\
		&
		\le
		2^{-\beta}
		.
	\end{flalign*}

	As this is true to any fixed $\guesserrandomness$ and $\permutationsfull$, we get that this is true unconditionally.
\end{proof}

What about the upper bound $\limit$?
Recall that while the Dealer follows the \epochstraklu (\Cref{def_epochstraklu}), only the first $\limit$ reasonably guessed cards are moved to the back; therefore only the first $\limit$ reasonable guesses are guaranteed to be distinct.
Any reasonable guess beyond $\limit$ may be useless for set encoding.
Further, if the Guesser is somehow able to reach~$\limit$ reasonable guesses, then moving cards to the back works in the Guesser's favor (as cards become predictable), and the probability to guess reasonably grows with every guess.
We address this concern by recalling that our \mttbdealer (\Cref{def_mttb_dealer}) plays the \epochstraterm in a sequence of epochs for which $\limit = \currentlength$, therefore, it is impossible to make more than $\limit$ reasonable guesses in an epoch.
In \Cref{section_upper_bound:subsec_universal_dealer} we present a Dealer for which $\limit < \currentlength$, and we restate the \namecref{claim_probability_for_reasonables_decays} to consider the upper bound $\limit$ (see \Cref{claim_universal_probability_for_reasonables_decays}).

As a corollary we get a bound on the expected number of reasonable guesses.

\begin{corollary}[Formal]
	\label{lemma_reasonable_expected_upper_bound}	
	For every Guesser with $\guessermemorysymbol$ bits of memory, every epoch~$\epochix \in [\epochsnum]$, the expected number of reasonable guesses during \currentepoch is at most
	$$
		\currentreasonableexpected \le \reasonablesboundi + 2.
	$$
\end{corollary}

For this \namecref{lemma_reasonable_expected_upper_bound} to be meaningful we require that $\firstepochbegin \le {n \over 8e}$, as otherwise it implies that $\currentreasonableexpected < \currentlength$, i.e.\ that the expected number of reasonable guesses in the epoch is less than the number of turns in the epoch, which is always true.
Observe this \namecref{lemma_reasonable_expected_upper_bound} is meaningful for every epoch played by our Dealer as the first epoch begins at turn $n-\firstepochbeginterm$.

\subsection{Upper Bound on the Number of Correct Guesses}
\label{section_upper_bound:subsec_correct_guesses_upper_bound}
Determining an upper bound on the number of \emph{reasonable} guesses, we can establish an upper bound on the number of \emph{correct guesses} per epoch that a memory bounded guesser can make against our Dealer.

Recall that during the $\epochix$th epoch our \mttbdealer (\Cref{def_mttb_dealer}) plays according to the \currentepochstra.
Recall that~$\currentreasonableexpected$ denotes the expected number of reasonable guesses during the $\epochix$th epoch.
Denote by $\correctsexpected_\epochix$ the expected number of correct guesses during the $\epochix$th epoch, where the expectation is taken over the Guesser's and Dealer's randomness $\guesserrandomness, \dealerrandomness$.
\begin{lemma}
	\label{thm_expected_number_of_correct_guesses_bound}
	For every Guesser \guesser with $\guessermemorysymbol$ bits of memory and every epoch $\epochix \in [\epochsnum]$,
	the expected number of correct guesses during the \currentepoch is at most
	$$
	\correctsexpected_\epochix \le {\currentreasonableexpected \over \currentepochbegin - \currentlength - \currentlimit}.
	$$
\end{lemma}
\begin{proof}
	\label{prf_expected_number_of_correct_guesses_bound}
	\newcommand{\movedtothebacknum}{\alpha}
	Let~$R_{i,j}$ be an indicator random variable for the event that the $j$th guess during the \currentepoch is reasonable.
	Let~$C_{i,j}$ be an indicator random variable for the event that the $j$th guess during the \currentepoch is correct.
	Let a random variable~$C_i$ denote the number of correct guesses that the Guesser \guesser made during the \currentepoch.
	
	We first bound the probability for a reasonable guess to be correct.
	Consider the~$j$th turn during a \currentepoch and assume that~$\movedtothebacknum$ cards were moved to the back until the~$j$th turn.
	At the beginning of the \currentepoch there are $\currentepochbegin$ cards left to play.
	In each turn one card is discarded, thus $j-1$ cards have been discarded since the beginning of the epoch.
	Each reasonable guess can be correct if it is one of the $\currentepochbegin - (j-1) - \movedtothebacknum$ remaining cards.
	Therefore
	\begin{flalign}
		\Pr\left[C_{i,j} = 1 | R_{i,j} = 1\right]
		&
		=
		\nonumber
		{1 \over \currentepochbegin - (j-1) - \movedtothebacknum}
		\\
		&
		\le
		{1 \over \currentepochbegin - (\currentlength-1-1) - \movedtothebacknum}
		\label[ineq]{prf_expected_number_of_correct_guesses_bound:turns_in_epoch}
		\\
		&
		\le
		{1 \over \currentepochbegin - \currentlength - \currentlimit + 2}
		\label[ineq]{prf_expected_number_of_correct_guesses_bound:back_limit}
		\\
		&
		<
		{1 \over \currentepochbegin - \currentlength - \currentlimit}
		\label[ineq]{prf_expected_number_of_correct_guesses_bound:prob_result}
		.
	\end{flalign}
	Where
	\Cref{prf_expected_number_of_correct_guesses_bound:turns_in_epoch} is due to $j \le \currentlength$,
	\Cref{prf_expected_number_of_correct_guesses_bound:back_limit} is because at most $\currentlimit$ cards are moved to the back.
	
	We can now bound the expected number of correct guesses directly
	\begin{flalign}
		c_i = \expected[C_i]
		&
		=
		\sum_{j \in [\currentlength]} \expected[C_{i,j}]
		\label[equa]{prf_expected_number_of_correct_guesses_bound:linearity_of_expectation}
		\\
		&
		=
		\sum_{j \in [\currentlength]}
		\Pr[R_{i,j}= 1] \cdot \expected[C_{i,j} | R_{i,j}= 1]
		+
		\Pr[R_{i,j}= 0] \cdot \expected[C_{i,j} | R_{i,j}= 0]
		\label[equa]{prf_expected_number_of_correct_guesses_bound:total_expectation}
		\\
		&
		=
		\sum_{j \in [\currentlength]} \Pr[R_{i,j}= 1] \cdot \expected[C_{i,j} | R_{i,j}= 1]
		\label[equa]{prf_expected_number_of_correct_guesses_bound:non_reasonable_guess_cannot_be_correct}
		\\
		&
		<
		\sum_{j \in [\currentlength]}  \Pr[R_{i,j}= 1] \cdot {1 \over \currentepochbegin - \currentlength- \currentlimit}
		\label[ineq]{prf_expected_number_of_correct_guesses_bound:probability_or_reasonable_guess_to_be_correct}
		\\
		&
		=
		{\currentreasonableexpected \over \currentepochbegin - \currentlength- \currentlimit}
		\label[equa]{prf_expected_number_of_correct_guesses_bound:number_of_reasonable_guesses}
		.
	\end{flalign}
	Where
	\Cref{prf_expected_number_of_correct_guesses_bound:linearity_of_expectation} is true by definition and due to linearity of expectation,
	\Cref{prf_expected_number_of_correct_guesses_bound:total_expectation} is from law of total expectation,
	\Cref{prf_expected_number_of_correct_guesses_bound:non_reasonable_guess_cannot_be_correct} is true since a correct guess is necessarily reasonable,
	\Cref{prf_expected_number_of_correct_guesses_bound:probability_or_reasonable_guess_to_be_correct} is due to \Cref{prf_expected_number_of_correct_guesses_bound:prob_result},
	and \Cref{prf_expected_number_of_correct_guesses_bound:number_of_reasonable_guesses} is true by definition since~$\sum_{j \in [\currentlength]}  \Pr[R_{i,j}= 1] = \currentreasonableexpected$.

\end{proof}

\subsection{Analysis of the \mttbdealer}
\label{section_upper_bound:subsec_mttb_dealer}
Having established a bound for a single epoch, we are ready to conclude the analysis of our Dealer and show its overall performance.

We recall that our \mttbdealer (\Cref{def_mttb_dealer} and \Cref{figure:dealer}) starts the game with a properly shuffled deck from which the Dealer draws until $\firstepochbegin$ cards are left, the Dealer then follows the MtBE-strategy over and over again and reshuffles every $\currentlength$ turns, and when $2m \cdot \log n$ cards are left, our Dealer shuffles the deck one last time and draws cards randomly until the end of the game.
In particular, the Dealer plays the \epochstraterm in a sequence of~$d$ epochs, where $d =\doneepochs$, each epoch consists of $\currentlength= \lengthterm$ turns, and as for every epoch it holds that $\currentlength \le \currentepochbegin - \currentlength$ it follows that we can set $\currentlimit=\currentlength$.

In the upcoming lemma, we will analyze and bound the cumulative number of correct guesses that any memory bounded Guesser can expect to make throughout the sequence of epochs.

\begin{lemma}
	\label{lemma_epochs_correct_guesses_bound}
	For $\guessermemorysymbol \le \memorybound$, every Guesser with~$\guessermemorysymbol$ memory bits that play against the \mttbdealer \dealerm is expected to guess correctly at most $1$ time in total throughout the sequence of epochs that follows the MtBE-strategy.
\end{lemma}
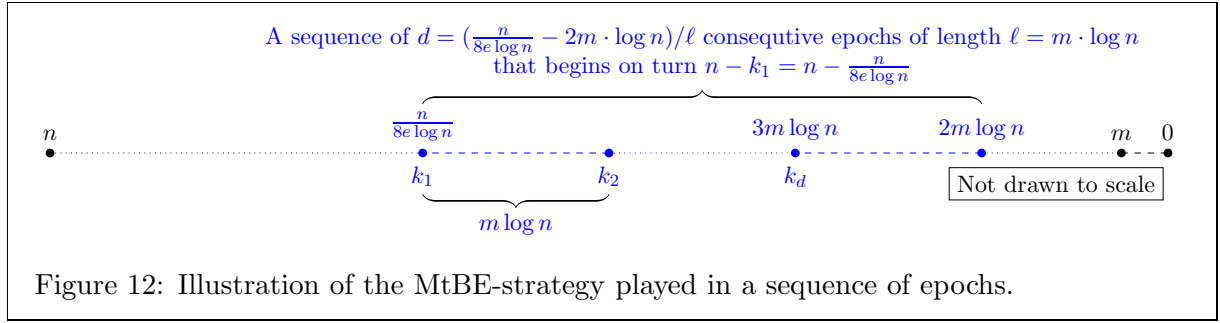
\begin{figure}[h]
	\begin{mdframed}
		\resizebox{\textwidth}{!}{%
			\begin{tikzpicture}[x=3cm]
				\draw [dashed,color=blue]
				(\figkpoint,0) node[circle,fill,inner sep=0.5mm,label=above:{$\koneepochsend$},alias=k_d1] (kd1) {}
				--
				(\figkpoint-1,0) node[circle,fill,inner sep=0.5mm,label=below:{${k_{d}}$},label=above:{$3 m \log n$},alias=k_d] (kd) {};
				
				\draw [dashed,color=blue]
				(\figkpoint-2,0) node[circle,fill,inner sep=0.5mm,label=below:{${k_2}$},alias=k_2] (k2) {}
				--
				(\figkpoint-3,0) node[circle,fill,inner sep=0.5mm,label=below:{${k_1}$},label=above:{$\koneepochs$},alias=k_1] (k1) {};
				
				\draw [dotted,color=blue] (k2) -- (kd) {};
				
				\draw [dotted]
				(0,0) node[circle,fill,inner sep=0.5mm,label=above:{$n$},alias=n] (n) {}
				--
				(k1) {};

				\draw [dashed]
				(5.75,0) node[circle,fill,inner sep=0.5mm,label=above:{$m$},alias=m] (m) {}
				--
				(6,0) node[circle,fill,inner sep=0.5mm,label=above:{$0$},alias=0] (0) {};
				
				\draw [dotted]
				(kd1) {}
				--
				(m) {};
				
				\begin{scope}[thick,decoration={calligraphic brace, amplitude=6pt,mirror}]
					\draw[decorate,color=blue]  (current bounding box.north) coordinate (aux) (kd1|-aux) -- node[above=1ex,align=center]{A sequence of $d = \doneepochs$ consequtive epochs of length $\ell = \elloneepochs$  \\ that begins on turn $n-\firstepochbegin = n-\koneepochs$} (k1|-aux);
					
					\draw[decorate,color=blue]   (current bounding box.south) coordinate (aux) (k1|-aux) -- node[below=1ex]{$m \log n$} (k2|-aux);
				\end{scope}
				\node[draw] at (5.4,-0.5) {Not drawn to scale};
			\end{tikzpicture}
		}
		\caption{Illustration of the \epochstraterm played in a sequence of epochs.}
		\label{figure:epochs}
	\end{mdframed}
\end{figure}
\begin{proof}
	\label{prf_epochs_correct_guesses_bound}
	
	The Dealer plays the \epochstraterm in a sequence for $d= \doneepochs$ epochs.
	All epochs are of the same length $\currentlength = \elloneepochs$, and during each epoch, the Dealer follows the \currentepochstra where $\currentlimit = \currentlength = \elloneepochs$.
	
	Since the first epoch begins when there are $\koneepochs$ cards left, and as a corollary from \Cref{lemma_reasonable_expected_upper_bound}, we get that the expected number of reasonable guesses during any epoch in the sequence is upper bounded by ${\currentlength \over \log n } + 2$.
	\begin{flalign}
		\currentreasonableexpected \le \reasonablesboundi + 2 \le {\currentlength \over \log n} + 2
		\label[ineq]{prf_epochs_correct_guesses_bound:current_reasonable_bound}
		.
	\end{flalign}
	We now bound the expected number of correct guesses during a single epoch
	\begin{flalign}
		c_i
		&
		<
		{\currentreasonableexpected \over \currentepochbegin - \currentlength - \currentlimit}
		\label[ineq]{prf_epochs_correct_guesses_bound:by_thm_expected_number_of_correct_guesses_bound}
		\\
		&
		\le
		{{\currentlength / \log n} + 2 \over \currentepochbegin - \currentlength - \currentlimit}
		\label[ineq]{prf_epochs_correct_guesses_bound:by_max_reasonable}
		\\
		&
		=
		{1 \over \log n} \cdot {\currentlength \over \currentepochbegin - 2\currentlength} + { 2 \over \currentepochbegin - 2\currentlength}
		\label[equa]{prf_epochs_correct_guesses_bound:since_limit_equals_length}
	\end{flalign}
	Where
	\Cref{prf_epochs_correct_guesses_bound:by_thm_expected_number_of_correct_guesses_bound} is by \Cref{thm_expected_number_of_correct_guesses_bound},
	\Cref{prf_epochs_correct_guesses_bound:by_max_reasonable} is true because of \Cref{prf_epochs_correct_guesses_bound:current_reasonable_bound},
	\Cref{prf_epochs_correct_guesses_bound:since_limit_equals_length} is true since $\currentlength = \currentlimit$.

	Observe that $$k_i = 3 m \log n + (d-i) \cdot m \log n = (d-i + 3) \currentlength.$$
	Summing over the epochs we get that
	\begin{flalign*}
		\sum_{i=1}^{d} c_i
		&
		<
		{1 \over \log n} \sum_{i=1}^d {\currentlength \over \currentepochbegin - 2 \currentlength}
		+
		2 \cdot \sum_{i=1}^d {1 \over \currentepochbegin - 2 \currentlength}
		\\
		&
		=
		{1 \over \log n} \sum_{i=1}^d {\currentlength \over (d-i + 3) \currentlength - 2 \currentlength}
		+
		2 \cdot \sum_{i=1}^d {1 \over (d-i + 3) \currentlength - 2 \currentlength}
		\\
		&
		=
		{1 \over \log n} \sum_{i=1}^d {1 \over d-i+1}
		+
		{2 \over \currentlength} \cdot \sum_{i=1}^d {1 \over d-i+1}
		\\
		&
		\approx
		{1 \over \log n} \ln d + {2 \over \currentlength} \ln d
		.
	\end{flalign*}
	Since $d = \doneepochs$ we get that
	$$
	\ln d  = \ln\left(\doneepochs\right) < \ln\left(n \over \currentlength \right) < \ln n.
	$$
	Combining the two above, we get that
	\begin{flalign*}
		\sum_{i=1}^{d} c_i
		<
		{\ln n \over \log n} + {2 \ln n \over \elloneepochs} < 1
		.
	\end{flalign*}
\end{proof}
With this, we can now analyze the performance of our Dealer.

\begin{theorem}	
	For any $m \le n$ there exists a Dealer \dealerm such that every Guesser \guesser with~$m$ memory bits is expected to make at most $\result$ correct guesses throughout the game.
\end{theorem}
\begin{proof}
	If $\guessermemorysymbol > \memorybound$ then $\result \ge \ln (n)$ so the argument is correct against any random-shuffle Dealer.
	If $\guessermemorysymbol \le \memorybound$ then consider a game played between any Guesser~\guesser with $\guessermemorysymbol \le \memorybound$ memory bits and our \mttbdealer \dealerm (\Cref{def_mttb_dealer}).
	
	\begin{itemize}
		\item
			Assume that all first $n-\koneepochs$ guesses are reasonable, i.e.\ while \dealerm draws at random, before the Dealer begins moving cards to the back.
			By linearity of expectation, the expected number of correct guesses in these rounds is:
			$$
			{1 \over n} + {1 \over n-1} + .... + {1 \over \koneepochs} \approx \ln n - \ln \koneepochs = \ln \log n + \ln 8e
			.
			$$
		\item
			By \Cref{lemma_epochs_correct_guesses_bound}, every Guesser \guesser with $\guessermemorysymbol$ memory bits is expected to make at most $1$ correct guess while the Dealer moves cards to the back, i.e., from turn~$n-\koneepochs$ until turn~$n-\koneepochsend$.
		\item
			Assume that the Guesser guesses reasonably in all remaining $\koneepochsend$ turns.
			These guesses yield~$\ln \koneepochsend = \ln \guessermemorysymbol + \ln \log n + \ln 2$ correct guesses in expectation.
	\end{itemize}
	Overall, the number of correct guesses that any Guesser with $\guessermemorysymbol \le \memorybound$ bits of memory is expected to make when playing against \mttbdealer \dealerm is at most $\ln m + 2 \ln \log n + \ln 16e$.
\end{proof}

\subsection{Universal \mttbdealer}
\label{section_upper_bound:subsec_universal_dealer}
So far, we have configured our Dealer differently according to the amount of memory bits that the Guesser had.
Using the building blocks and ideas seen so far in the section, we present a \emph{universal} adaptive Dealer that works well against any Guesser with any amount of memory, without knowing how much memory the Guesser has.
Albeit, with a drawback that a Guesser with $\guessermemorysymbol$ bits of memory is expected to make slightly more than $\ln m$ correct guesses.
I.e.\ a Guesser with perfect memory is expected to achieve more than $$\resultUniversaln$$ correct guesses in expectation.

The starting point for the universal Dealer is the same as that of the \mttbdealer (\Cref{def_mttb_dealer}).
Similarly, our universal Dealer separates the turns to epochs during which the Dealer follows the MtBE-strategy.
However, the epochs will shrink and become shorter as more cards are drawn, and the analysis will be different.

\begin{definition}
	\label{def_universal_dealer}
	The \universaldealer~\dealeruni plays according to the strategy:
	\begin{enumerate}
		\item
			Shuffle the deck uniformly at random, and draw cards one by one until $\firstuniepochbegin$ cards are left in the deck.
		\item
			Play the \epochstraterm in a sequence of $\epochsnum$ epochs, where $ \epochsnum = \duniepochs$, such that the $\epochix$th epoch begins when $\currentepochbegin = {n \over 8e \log^{1+\epochix} n}$ cards are left, i.e.\ the length of the $\epochix$th epoch is $\currentulength = \currentepochbegin (1 - {1 \over \log n})$, and during each epoch at most $\currentulimit = {2 \currentulength \over \log^2 n}$ cards are moved to the back.
		\item
			When $\log^4 n$ cards left, shuffle the deck one last time and draw cards at random.
	\end{enumerate}
\end{definition}

We begin our analysis in the same way as we did for the \mttbdealer.
We consider the same implementation of the Dealer (\Cref{section_upper_bound:subsec_towards_a_proof}), and the same encoding scheme for sets (\Cref{section_upper_bound:subsec_encoding_scheme}).

Recall the discussion about the maximal number of cards moved to the back during an epoch, right before \Cref{lemma_reasonable_expected_upper_bound}.
In that discussion we argued that we may ignore the role of the bound~$\limit$ since it is impossible to make more than $\limit = \ell$ reasonable guesses.
This is not the case for the \universaldealer, as during the $\epochix$th epoch at most $\currentulimit = \reasonablesbounduni$ cards are moved to the back.
We restate, without proof, \Cref{claim_probability_for_reasonables_decays}.

\begin{claim}[Restate \Cref{claim_probability_for_reasonables_decays}]
	\label{claim_universal_probability_for_reasonables_decays}
	For every Guesser \guesser with $\guessermemorysymbol$ bits of memory, and every epoch $\epochix \in [\epochsnum]$, for $\reasonablesdiff < \currentulimit - \reasonablesbounduniterm$,
	the probability that \guesser makes more than $\reasonablesdiff + \reasonablesbounduniterm$ reasonable guesses during the $\epochix$th epoch, is at most
	$$
	\Pr_{\encsetdesc}\left[\currentreasonableRV \ge \reasonablesbounduniterm + \reasonablesdiff \right] \le 2^{-\reasonablesdiff}
	.
	$$
\end{claim}
As the \universaldealer begins following the \epochstraterm when  $\firstuniepochbegin$ cards are left, we get that the probability for more than $\max\{{\currentulength \over \log^2 n}, \guessermemorysymbol\}$ reasonable guesses decays exponentially.
It follows that the expected number of reasonable guesses per epoch depends on the amount of memory the Guesser has.
In particular, this \namecref{claim_universal_probability_for_reasonables_decays} clarifies that we must analyze differently the epochs for which~$\guessermemorysymbol \ge \reasonablesbounduni$ than the other epochs.

Therefor, for every $\guessermemorysymbol$, we separate the epochs into two eras.
During the \emph{low-memory era}, moving cards to the back works in the Dealer's favor as the MtBE-strategy guarantees that no Guesser can guess well.
During the \emph{high-memory era}, moving cards to the back works in the Guesser's favor, and we assume that the Guesser gains the maximal advantage from it.
\begin{corollary}
	\label{claim_universal_reasonable_expected_upper_bound}
	During the low-memory era, that is for $\epochix \in [\epochsnum]$ and $\guessermemorysymbol < \reasonablesbounduni$, any Guesser with~$\guessermemorysymbol$ bits of memory,
	makes in expectation at most
	$$
	\currentreasonableexpected \le \reasonablesbounduni + 3
	$$
	reasonable guesses during the \currentuniepoch.
\end{corollary}
\begin{proof}
	Observe that since $\firstepochbegin = \firstuniepochbegin$ then ${8 \cdot e \cdot k_1 \currentulength \over n} =\reasonablesbounduni$ and by assumption $\guessermemorysymbol < {\currentulength \over \log^2 n}$ so it follows that $\reasonablesbounduniterm =  \reasonablesbounduni$.
	To ease the analysis, we assume that the first $\reasonablesbounduni$ guesses in the epoch are reasonable.
	It follows that
	\begin{flalign*}
		\currentreasonableexpected
		&
		=
		\sum_{\reasonablesnum=0}^{\currentulength} \reasonablesnum \cdot \Pr\left[ \currentreasonableRV = \reasonablesnum\right]
		\\
		&
		=
		\sum_{\reasonablesdiff=0}^{\currentulength-\reasonablesbounduni} \left(\reasonablesbounduni+\reasonablesdiff\right) \cdot \Pr\left[ \currentreasonableRV = \reasonablesbounduni+\reasonablesdiff\right]
		\\
		&
		=
		\reasonablesbounduni
		+
		\underbrace{
		\sum_{\reasonablesdiff=0}^{\currentulimit -\reasonablesbounduni} \reasonablesdiff \cdot \Pr\left[ \currentreasonableRV = \reasonablesbounduni+\reasonablesdiff\right]
		}_{(\ast)}
		+
		\underbrace{
		\sum_{\reasonablesdiff=\currentulimit -\reasonablesbounduni+1}^{\currentulength-\reasonablesbounduni} \reasonablesdiff \cdot \Pr\left[ \currentreasonableRV = \reasonablesbounduni+\reasonablesdiff\right]
		}_{(\ast\ast)}
	\end{flalign*}
	From \Cref{claim_universal_probability_for_reasonables_decays}, we know that the probability for reasonable guesses decays exponentially until $\reasonablesdiff \le \currentulimit - \reasonablesbounduni$, so the first sum $(\ast)$ is upper bounded by
	
	$$(\ast) < \sum_{\reasonablesdiff=0}^\infty \reasonablesdiff \cdot 2^{-\reasonablesdiff} = 2.$$
	The shortest epoch consists of more than $\log^4 n$ turns, and since $\currentulimit = 2 \cdot \reasonablesbounduni$, it follows that
	$$\currentulimit - \reasonablesbounduni = \reasonablesbounduni \ge \log^2 n.$$
	From \Cref{claim_universal_probability_for_reasonables_decays} we get that the probabilities in the second sum $(\ast\ast)$ are upper bounded by $2^{-\log^2 n} = {1 \over n^{\log n}}$.
	Which means that the second sum $(\ast\ast)$ is upper bounded
	
	$$
	(\ast\ast)
	<
	{\currentulength \cdot \currentulength \over n^{\log n}} < 1
	$$.
\end{proof}
From \Cref{lemma_epochs_correct_guesses_bound} we conclude an upper bound on the expected number of correct guesses during an epoch in the low-memory era.
\begin{corollary}
	\label{claim_universal_correct_expected_upper_bound}
	During the low-memory era, that is for $\epochix \in [\epochsnum]$ and $\guessermemorysymbol < \reasonablesbounduni$, any Guesser with~$\guessermemorysymbol$ bits of memory,
	makes on expectation at most
	$$
	\currentcorrectexpected \le {1 \over \log (n) - 2}
	$$
	correct guesses during the \currentuniepoch.
\end{corollary}
\begin{proof}
	\label{prf_universal_correct_expected_upper_bound}
	Observe that for any \currentuniepoch it holds that $\currentepochbegin - \currentulength = {\currentepochbegin \over \log n}$ and since $\currentulength = \currentepochbegin (1 - {1 \over \log n})$ then $\currentepochbegin - \currentulength = {\currentulength \over \log (n) - 1}$.
	We get that
	\begin{flalign}
		\currentcorrectexpected
		&
		\le
		{\currentreasonableexpected \over \currentepochbegin - \currentulength - \currentulimit}
		\label[ineq]{prf_universal_correct_expected_upper_bound:by_thm_expected_number_of_correct_guesses_bound}
		\\
		&
		=
		{1 \over \log^2 n} \cdot {\currentulength \over \currentepochbegin - \currentulength - \currentulimit}
		\label[ineq]{prf_universal_correct_expected_upper_bound:by_claim_universal_reasonable_expected_upper_bound}
		\\
		&
		=
		{1 \over \log^2 n} \cdot {\currentulength \over {\currentulength \over \log (n) - 1} - {2 \currentulength \over \log^2 n}}
		\label[ineq]{prf_universal_correct_expected_upper_bound:assing_variables}
		\\
		&
		=
		{1 \over \log^2 n} \cdot {1 \over {1 \over \log (n) - 1} - {2 \over \log^2 n}}
		\nonumber
		\\
		&
		<
		{1 \over \log (n) - 2}
		\nonumber
	\end{flalign}
	Where
	\Cref{prf_universal_correct_expected_upper_bound:by_thm_expected_number_of_correct_guesses_bound} follows from \Cref{thm_expected_number_of_correct_guesses_bound},
	\Cref{prf_universal_correct_expected_upper_bound:by_claim_universal_reasonable_expected_upper_bound} is true since for $m < \reasonablesbounduni$ we get from \Cref{claim_universal_reasonable_expected_upper_bound} that $\currentreasonableexpected \le \reasonablesbounduni$,
	and \Cref{prf_universal_correct_expected_upper_bound:assing_variables} is true since $\currentepochbegin - \currentulength = {\currentulength \over \log (n) - 1}$ and $\currentulimit = {2 \currentulength \over \log^2 n}$.
\end{proof}

The above \namecref{claim_universal_correct_expected_upper_bound} states that the MtBE-strategy works well while the Guesser has insufficient memory.
However, as mentioned already, once the Guesser reaches the maximal number of cards moved to the back, the MtBE-strategy works in the Guesser's favor.
We want to bound the Guesser's benefit during such an epoch.

\begin{claim}
	\label{claim_universal_correct_expected_upper_bound_general}
	For every epoch, the expected number of correct guesses that any Guesser makes, is at most
	$$
	\ln\left(\log(n) + 3\right).
	$$
\end{claim}
\begin{proof}
	\label{prf_univeral_correct_expected_upper_bound_general}
	Let $n-\turnix$ be some turn during the $\epochix$th epoch \currentuniepoch.
	As we did in the proof of \Cref{thm_expected_number_of_correct_guesses_bound}, since at most $\currentulimit$ cards can be moved to the back, it follows that the probability for a reasonable guess to be correct at turn $n-\turnix$ is at most $1 \over \turnix - \currentulimit$.
	It follows that if all guesses in the epoch are reasonable then the expected number of correct guesses is $$\sum_{\turnix=\currentepochbegin}^{\currentepochbegin - \currentulength} {1 \over \turnix - \currentulimit} \approx \ln\left(\currentepochbegin - \currentulimit\right) -\ln\left(\currentepochbegin - \currentulength- \currentulimit\right) = \ln \left(\currentepochbegin - \currentulimit \over \currentepochbegin - \currentulength- \currentulimit\right).$$
	We bound the term inside the $\ln$.
	\begin{flalign}
		{\currentepochbegin - \currentulimit \over \currentepochbegin - \currentulength- \currentulimit}
		&
		=
		{\currentepochbegin - {2 \currentulength \over \log^2n} \over {\currentepochbegin \over \log n}- {2 \currentulength \over \log^2n}}
		\label[equa]{prf_univeral_correct_expected_upper_bound_general:remaining_cards}
		\\
		&
		<
		{\currentepochbegin - {2 \currentepochbegin \over \log^2n} \over {\currentepochbegin \over \log n}- {2 \currentepochbegin \over \log^2n}}
		\label[ineq]{prf_univeral_correct_expected_upper_bound_general:ineq_a_over_b}
		\\
		&
		=
		\log (n) + 2 + {2 \over \log (n) + 2}
		\label[equa]{prf_univeral_correct_expected_upper_bound_general:division}
		\\
		&
		<
		\log (n) + 3
		\nonumber
		.
	\end{flalign}
	Where
	\Cref{prf_univeral_correct_expected_upper_bound_general:remaining_cards} is true since $\currentepochbegin - \currentulength = {\currentepochbegin \over \log n}$ and since $\currentulimit= {2 \currentulength \over \log^2n}$,
	\Cref{prf_univeral_correct_expected_upper_bound_general:ineq_a_over_b} is true since if $a > b$ and $c < d < b$ then ${a -c \over  b -c} < {a -d \over b - d}$,
	and
	\Cref{prf_univeral_correct_expected_upper_bound_general:division} is the result of division.
	
	It follows that the expected number of correct guesses during any epoch is upper bounded by $\ln\left(\log(n) + 3\right)$.
\end{proof}

We therefore have two upper bounds on the number of correct guesses, one for the low-memory era (\Cref{claim_universal_correct_expected_upper_bound}) and a general one  (\Cref{claim_universal_correct_expected_upper_bound_general}) that we will use for epochs during the high-memory era.

\begin{theorem}
	There exists an adaptive \universaldealer against which any Guesser with~$\guessermemorysymbol$ bits of memory can score at most $$\resultUniversal$$ correct guesses in expectation.
\end{theorem}
\begin{proof}
	Consider our \universaldealer from \Cref{def_universal_dealer}.
	Assume that all the guesses before turn $n-\firstepochbegin = n- \firstuniepochbegin$ were reasonable, this results in~$2\ln \log n + \ln 8e$ correct guesses.

	\paragraph{Low-memory era:}
	Recall that the \universaldealer plays the \epochstraterm for $\epochsnum$ epochs where $\epochsnum=\duniepochs$.
	Observe that $\epochsnum \le \log_{\log n} n = {\log n \over \log \log n}$.
	Let $\epochsnum_1$ be the number of epochs for which the Guesser has insufficient memory, i.e., $\epochsnum_1 = |\{\epochix \in [\epochsnum] \colon \guessermemorysymbol \le \reasonablesbounduni\}|$.
	\Cref{claim_universal_correct_expected_upper_bound} states that the expected number of correct guesses during an epoch in the low-memory era is at most~$1 \over \log(n) - 2$.
	It follows that the cumulative number of correct guesses during these epochs is less than
	\begin{flalign*}
	\epochsnum_1 \cdot {1 \over \log(n) - 2}
	&
	\le
	\epochsnum \cdot {1 \over \log(n) - 2}
	\\
	&
	<
	{\log n \over \log \log n}
	\cdot
	{1 \over \log(n) - 2}
	\\
	&
	<
	1
	.
	\end{flalign*}

	\paragraph{High-memory era:}
	Let $\epochsnum_2 = \epochsnum - \epochsnum_1$ be the number of epochs in the high-memory era.
	Observe that $\guessermemorysymbol \ge {\currentulength \over \log^2 n}$ for every epoch for which $\currentepochbegin \le \guessermemorysymbol \cdot \log^2 n$.
	Therefore,
	$$
	\epochsnum_2
	\le
	\log_{\log n} \left(\guessermemorysymbol \cdot \log^2 n\right)
	=
    {\ln \guessermemorysymbol \over \ln \log n} + 2
	.
	$$
	\Cref{claim_universal_correct_expected_upper_bound_general} states that the expected number of correct guesses during any epoch is upper bounded by $\ln\left(\log(n) + 3\right)$.
	It follows that total number of correct guesses during the high-memory era is
	\begin{flalign*}
		\epochsnum_2 \cdot \ln\left(\log(n) + 3\right)
		&
		\le
        \left({\ln \guessermemorysymbol \over \ln \log n} + 2\right) \cdot \ln\left(\log(n) + 3\right)
		\\
		&
		=
		\ln \guessermemorysymbol \cdot {\ln\left(\log(n) + 3\right) \over \ln \log n}+ 2 \ln (\log(n) + 3)
		\\
		&
		=
		\ln \guessermemorysymbol \cdot (1 + o(1))+ 2 \ln (\log(n) + 3)
		.
	\end{flalign*}

	Assume the Guesser guesses reasonably the last $\log^4 n$ turns, the expected number of correct guesses is at most $4 \ln \log n$.

	Summing it all together, we get that the expected number of correct guesses that any Guesser with $\guessermemorysymbol$ bits of memory can score against our \universaldealer is at most
	$$
	\resultUniversal.
	$$
\end{proof}

\section{Discussion and Open Problems}

\subsection{Relation to Mirror Game}
\label{sec_discussion:subs_discussion_card_guessing_mirror_game}
\input{parts/discussion/relation_to_mirror_games_full.tex}

\subsection{Card Guessing variants}

\input{parts/discussion/card_guessing_variants.tex}

\subsection{Low Memory Dealer: a Conjecture}
What happens when the Dealer has limited memory, say $m$ bits, and wants to pick a permutation that is unpredictable by any Guesser, that has no limitation on the number of bits it can store or its computational power.
The dealer also has a limited number of long lived random bits.
It \emph{seems} that the best such a Dealer can do is pick the next card from a set of $m$ cards at random, and the question is how to assure that there is such a set available for as many rounds as possible.
We have found such a method that makes any Guesser pick correctly only  $O(n/m+\log m)$ cards in expectation.
The method does not require any secrecy from the dealer.
We conjecture that this bound is the best possible, at least for dealers without any secret memory.

\subsection{Prediction as Approximation and Data Structures}
\label{sec_discussion:sub_streaming_algorithms_and_data_structures}

In the streaming model of computation an algorithm observes a stream of elements and computes a function on the stream seen so far.
For a memory bounded algorithm, it is a typical relaxation that the algorithm outputs an approximate value of the function.
In card guessing, we ask the algorithm to predict the next card, but this prediction can be though of as a way to measure distance, and thus, as an approximation.

\newcommand{\vecg}{\vec{g}}
\newcommand{\vecd}{\vec{d}}

Consider a partial game played for some turns between some Guesser and the random-shuffle Dealer.
Let $\vecg$ be an $n$-vector associated with the probability for each card to be guessed by the Guesser at that turn.
Let $D$ denote the set of cards that are still in the deck, and let $\vecd = \fOne_D \cdot {1 \over |D|}$ be the vector associated with the probability to draw each one of them at random.
In general, the inner product between two normalized vectors indicates how close they are.
Intuitively, if $\vecg$ and $\vecd$ are close, it means that the Guesser captured more accurately the set of cards that are still in the deck.
So if the inner product of $\vecg$ and $\vecd$ is bounded, we can say that $\vecg$ approximate $\vecd$.

Now consider the chance of a correct guess.
The probability to guess correctly is the probability that both the Guesser and the Dealer sampled the same card independently.
As the vectors $\vecg$ and $\vecd$ represents the probability for each card to be drawn, we get that the probability for a correct guess is the inner product between $\vecg$ and $\vecd$.
\begin{flalign*}
	\Pr[\text{correct guess}]
	=
	\sum_{x \in [n]} \vecg_x \cdot \vecd_x
	=
	\innprod{\vecg}{\vecd}
\end{flalign*}
The idea is visualized in~\Cref{tbl_inner_product}.

\begin{table}[h]
	\caption{Each card belong to one of the four groups. The probability to guess correctly is the inner product between the vertices in vertical circle and the vertices in the horizontal one.}
	\label{tbl_inner_product}
	\begin{tabular}{|l|c|c|}
		\hline
		\diagbox[height=3\line,width=8cm]{Guesser thinks the card drawn}{Card drawn} & Yes & No
		\\
		\hline
		Yes & & \marktopleft{c1} Will not guess
		\\
		\hline
		No & \marktopleft{c2} Futile guess &  Reasonable guess\markbottomright{c1} \markbottomright{c2}
		\\
		\hline
	\end{tabular}
\end{table}

Similarly, we can think about the Guesser as holding a set-membership data structure.
The Guesser guesses a card from the set of elements for which the data structure claims that are absent from the set.
The prediction remains a way to measure performance.
On that aspect, our work joins that of Naor and Yogev~\cite{NY19} who studied the ability of an adversary to find a False Positive in a Bloom Filter, i.e., to find an element that does not reside in the set though the Bloom Filter thinks it does.
In particular, they considered the advantage of the adversary to find a False-Positive, and in our terms, to upper bound the probability to guess reasonably.

\subsection{Security}
Unpredictability is a goal of many security mechanisms to ensure secure and reliable operation of services and authentication.
Security measures are taken to increase the unpredictability of critical services and protocols.

As an example, in TCP, the source port, the sequence number and the ack number are initially randomized in order to decrease the probability of hijacking the session.
Being able to predict these numbers increases the probability of a successful  attack on the server in the form of connection hijack or denial of service.
As the amount of source ports is finite, by opening many connections to a server it may be possible to get a better prediction the next source port to be used without actually observing all cards.
In this aspect, it is interesting whether a similar approach to \mttbphrase may be beneficial for protecting services against such attacks.

\section*{Acknowledgments}
\input{parts/acknowledgments.tex}



\addcontentsline{toc}{section}{References}
\printbibliography



\end{document}

%% file: parts/figures/following_subsets.tex
\tikzmath{
	\ampparam = 1.4;
	\w1 = 1.8; 
	\w2 = \w1 * \ampparam;
	\w3 =  \w2 * \ampparam;
	\w4 =  \w3 * \ampparam;
	\y1 = 0.26;
} 

\begin{tikzpicture}[x=3cm]
	\draw[thick,->] (0,0) -- (4.2,0) node[anchor=north west] {n};
	
	\draw[color=red,-] (0,0) -- (0,1) node[anchor=south] {$1$};
	\draw[color=red,-] (\w1,0) -- (\w1,1) node[anchor=south] {$w$};
	\draw[color=red,-] (\w2,0) -- (\w2,1) node[anchor=south] {$(1+\ampliparam)w$};
	\draw[color=red,-] (\w3,0) -- (\w3,1) node[anchor=south] {$(1+\ampliparam)^2w$};
	
	\draw[thick,-,color=blue] (0, \y1 * 3) -- (\w1, \y1 * 3);
	\draw[thick,-,color=blue] (0, \y1 * 2) -- (\w2, \y1 * 2);
	\draw[thick,-,color=blue] (0, \y1) -- (\w3, \y1);
\end{tikzpicture}

%% file: parts/figures/correct_guess_encoding.tex
\tikzmath{
	\ampparam = 1.4;
	\p1 = 2.; 
	\p2 = 2.9; 
	\p3 = 4.2;
	\w1 = 0.1; 
	\correctsx = 3.5;
	\correctsy = 1;
	\memory = 1.5;
} 

\begin{tikzpicture}[
		x=3cm,
		circ/.style={shape=circle, inner sep=3pt, draw, fill=white,node contents=}
		]
	\draw[thick,->] (0,0) node[anchor=north] {first turn} -- (5,0) node[anchor=north] {last turn};
	
	\draw[thick, color=BlueViolet,->] (\memory,1) node[anchor=south] {Guesser's memory state} -- (\memory,0);
	
	\draw[line width=0.5mm, color=purple,->] (\memory,0) -- (5,0) node[midway,anchor=north] {Cards drawn during incorrect guesses};
	
	\draw node (p1) at (\p1, 0) [circ];
	\draw node (p2) at (\p2, 0) [circ];
	\draw node (p3) at (\p3, 0) [circ];
	
	\draw[thick, color=OliveGreen,->] (\correctsx,\correctsy) node[anchor=south] {Correct guesses} -- (p1);
	\draw[thick, color=OliveGreen,->] (\correctsx,\correctsy) -- (p2);
	\draw[thick, color=OliveGreen,->] (\correctsx,\correctsy) -- (p3);

	\node[draw] at (4.4,-1.) {Not drawn to scale};
	
\end{tikzpicture}

%% file: parts/figures/adaptive_dealer_general_scheme.tex
\tikzmath{
	\figkpoint = 5.;
} 

\begin{tikzpicture}[x=3cm]
	\draw [dashed,color=blue]
	(\figkpoint,0) node[circle,fill,inner sep=0.5mm,label=above:{reshuffle},alias=k_d1] (kd1) {}
	--
	(\figkpoint-1,0) node[circle,fill,inner sep=0.5mm,label=above:{reshuffle},alias=k_d] (kd) {};
	
	\draw [dashed,color=blue]
	(\figkpoint-2,0) node[circle,fill,inner sep=0.5mm,label=above:{reshuffle},alias=k_2] (k2) {}
	--
	(\figkpoint-3,0) node[circle,fill,inner sep=0.5mm,label=above:{reshuffle},alias=k_1] (k1) {};
	
	\draw [dotted,color=blue] (k2) -- (kd) {};
	
	\draw [dotted]
	(0,0) node[circle,fill,inner sep=0.5mm,label=below:{first turn},alias=n] (n) {}
	--
	(k1) {};

	\draw [dotted]
	(kd1) {}
	--
	(6,0) node[circle,fill,inner sep=0.5mm,label=below:{last turn},alias=0] (0) {};
	
	\begin{scope}[thick,decoration={calligraphic brace, amplitude=5pt}]
		
		\draw[decorate]  (current bounding box.north) coordinate (aux) (n|-aux) -- node[above=1ex]{Draw at random} (k1|-aux);
		\draw[decorate]  (current bounding box.north) coordinate (aux-1) (kd1|-aux) -- node[above=1ex]{Draw at random} (0|-aux);
		
	\end{scope}
	
	\begin{scope}[thick,decoration={calligraphic brace, amplitude=5pt,mirror}]
		
		\draw[decorate,color=blue]   (current bounding box.south) coordinate (aux) (k1|-aux) -- node[below=1ex]{Move to the Back} (k2|-aux);
		\draw[decorate,color=blue]   (current bounding box.south) coordinate (aux-1) (kd|-aux) -- node[below=1ex]{Move to the Back} (kd1|-aux);
	\end{scope} 
	
\end{tikzpicture}

%% file: parts/figures/reasonable_guess_encoding.tex
\tikzmath{
	\ampparam = 1.4;
	\shift = -0.6;
	\epochbeginx = 2.15+\shift;
	\epochendx = 4.2+\shift;
	\epochy = 0.2;
	\p1 = 2.5+\shift; 
	\p2 = 2.7+\shift; 
	\p3 = 2.9+\shift;
	\p4 = 3.3+\shift;
	\p5 = 3.45+\shift;
	\w1 = 0.1+\shift; 
	\reasonablesx = 2.55+\shift;
	\correctsx = 3.85+\shift;
	\titles = 1;
	\memory = 1.;
} 

\begin{tikzpicture}[x=3cm,
	circ/.style={shape=circle, inner sep=1pt, draw, node contents=},
	circorrect/.style={shape=circle, draw, fill=white,inner sep=3pt,minimum size=3pt,node contents=},
	circb/.style={shape=circle, inner sep=3pt, draw, fill=white,node contents=},
	reasonabline/.style={thick, color=OliveGreen,->},
	correctline/.style={thick, color=MidnightBlue,->}
	]
	\draw[thick,->] (0,0) node[anchor=south] {first turn} -- (5,0) node[anchor=south] {last turn};
	
	\draw[thick, color=BlueViolet,->] (\memory,-\titles) node[anchor=north] {Guesser's memory state} -- (\memory,0);
	
	\draw (\epochbeginx, \epochy) node[anchor=south] {Epoch begin} -- (\epochbeginx, 0);
	\draw (\epochendx, \epochy) node[anchor=south] {Epoch end} -- (\epochendx, 0);
	
	\draw[line width=0.5mm, color=purple,->] (\memory,0) -- (5,0) node[midway,anchor=north] {Cards drawn minus reasonable guess during epoch};
	
	\draw node (p1) at (\p1, 0) [circ];
	\draw node (p2) at (\p2, 0) [circ];
	\draw node (p3) at (\p3, 0) [circorrect];
	\draw node (p4) at (\p4, 0) [circ];
	\draw node (p5) at (\p5, 0) [circorrect];
	
	\draw[inner sep=0pt] node (reasonables) at (\reasonablesx, \titles) {};
	
	\draw[reasonabline] (reasonables) node[anchor=south] {Reasonable guesses} -- (p1);
	\draw[reasonabline] (reasonables) -- (p2);
	\draw[reasonabline] (reasonables) -- (p3);
	\draw[reasonabline] (reasonables) -- (p4);
	\draw[reasonabline] (reasonables) -- (p5);
	
	\draw[inner sep=0pt] node (corrects) at (\correctsx, \titles) {};
	\draw[correctline] (corrects) node[anchor=south] {Correct guesses} -- (p3);
	\draw[correctline] (corrects) node[anchor=south] {Correct guesses} -- (p5);
	
	\node[draw] at (4.4,-1) {Not drawn to scale};

	
)\end{tikzpicture}

%% file: parts/figures/correct_guess_encoding_full.tex
\tikzmath{
	\ampparam = 1.4;
	\p1 = 2.; 
	\p2 = 3.; 
	\p3 = 4.2;
	\w1 = 0.1; 
	\correctsx = 3.5;
	\correctsy = 1;
	\memory = 1.0;
} 

\begin{tikzpicture}[
	x=3cm,
	circ/.style={shape=circle, inner sep=3pt, draw, fill=white,node contents=}
	]
	\draw[thick,->] (0,0) node[anchor=north] {first turn} -- (5,0) node[anchor=north] {last turn};
	
	\draw[thick, color=BlueViolet,->] (\memory,1) node[anchor=south] {Guesser's memory state $\memorystatesymbol$ at turn $n-\orderedsetsize$} -- (\memory,0);
	
	\draw[line width=0.5mm, color=purple,->] (\memory,0) -- (5,0) node[midway,anchor=north] {Cards drawn during incorrect guesses $\orderedset \setminus \{g_\turnix | \turnix \in \turnssetsymbol\}$};
	
	\draw node (p1) at (\p1, 0) [circ];
	\draw node (p2) at (\p2, 0) [circ];
	\draw node (p3) at (\p3, 0) [circ];
	
	\draw node[anchor=south,thick, inner sep=0pt, color=OliveGreen,"\textcolor{OliveGreen}{Correct guesses locations $\turnssetsymbol$}"] (corrects) at (\correctsx,\correctsy) {};
	
	\draw[thick, color=OliveGreen,->] (corrects) -- (p1);
	\draw[thick, color=OliveGreen,->] (corrects) -- (p2);
	\draw[thick, color=OliveGreen,->] (corrects) -- (p3);
	
	\pic [draw, dashed, color=OliveGreen, "$\correctsnum$", angle eccentricity=1.3] {angle = p1--corrects--p3};
	
	\node[draw] at (4.4,-1.) {Not drawn to scale};
	
\end{tikzpicture}

%% file: parts/figures/reasonable_guess_encoding_full.tex
\tikzmath{
	\ampparam = 1.4;
	\shift = -0.6;
	\epochbeginx = 2.15+\shift;
	\epochendx = 4.2+\shift;
	\epochy = 0.2;
	\p1 = 2.5+\shift; 
	\p2 = 2.6+\shift; 
	\p3 = 2.9+\shift;
	\p4 = 3.1+\shift;
	\p5 = 3.45+\shift;
	\w1 = 0.1+\shift; 
	\reasonablesx = 2.55+\shift;
	\correctsx = 3.85+\shift;
	\titles = 1;
	\memory = 1.;
} 

\begin{tikzpicture}[x=3cm,
	circ/.style={shape=circle, inner sep=1pt, draw, node contents=},
	circorrect/.style={shape=circle, draw, fill=white,inner sep=3pt,minimum size=3pt,node contents=},
	circb/.style={shape=circle, inner sep=3pt, draw, fill=white,node contents=},
	reasonabline/.style={thick, color=ForestGreen,->},
	correctline/.style={thick, color=blue,->}
	]
	\draw[thick,->] (0,0) node[anchor=south] {first turn} -- (5,0) node[anchor=south] {last turn};
	
	\draw[thick, color=Fuchsia,->] (\memory,-\titles) node[anchor=north] {Guesser's memory state $\memorystatesymbol$ at turn $n-\firstepochbegin$} -- (\memory,0);
	
	\draw (\epochbeginx, \epochy) node[anchor=south] {Epoch begin} -- (\epochbeginx, 0);
	\draw (\epochendx, \epochy) node[anchor=south] {Epoch end} -- (\epochendx, 0);
	
	\draw[line width=0.5mm, color=red,->] (\memory,0) -- (5,0) node[midway,anchor=north] {Cards drawn minus epoch's reasonable guesses $\encset \setminus \{g_\turnix | \turnix \in \turnssetsymbol\}$};
	
	\draw node (p1) at (\p1, 0) [circ];
	\draw node (p2) at (\p2, 0) [circ];
	\draw node (p3) at (\p3, 0) [circorrect];
	\draw node (p4) at (\p4, 0) [circ];
	\draw node (p5) at (\p5, 0) [circorrect];
	
	\draw[inner sep=0pt] node (reasonables) at (\reasonablesx, \titles) {};
	
	\draw[reasonabline] (reasonables) node[anchor=south] {Reasonable guesses $\turnssetsymbol$} -- (p1);
	\draw[reasonabline] (reasonables) -- (p2);
	\draw[reasonabline] (reasonables) -- (p3);
	\draw[reasonabline] (reasonables) -- (p4);
	\draw[reasonabline] (reasonables) -- (p5);
	\pic [draw, dashed, color=OliveGreen, "$\reasonablesnum$", angle eccentricity=1.3] {angle = p1--reasonables--p5};

	\draw[inner sep=0pt] node (corrects) at (\correctsx, \titles) {};
	\draw[correctline] (corrects) node[anchor=south] {Correct guesses $\correctvector$} -- (p3);
	\draw[correctline] (corrects) -- (p5);

	\node[draw] at (4.4,-1.) {Not drawn to scale};
	
	)\end{tikzpicture}

%% file: parts/discussion/relation_to_mirror_games_full.tex
In this section we discuss the relation between Card Guessing and Mirror Game. Recall the Mirror Game presented by Garg and Schneider~\cite{GargS18}, where Alice (the first player) and Bob take turns saying a name of a card (i.e.\ a number) from a deck of size $2n$, and a player loses if this card was mentioned already by either one of the players.
When no more cards are left to say, then the result of the game is a draw.
Bob, who plays second, has a {\em low memory, simple, and efficient strategy called mirroring}:
Bob fixes any matching on the cards, and for every card said by Alice, Bob responds with the matched card.
This allows Bob to say in every turn a card that has not appeared yet, and in our terms, to yield a reasonable guess, and {\em is assured not to lose}.

The question at hand is how much memory Alice needs in order not to lose (i.e.\ assure a draw).
Garg and Schneider~\cite{GargS18} showed that every deterministic ``winning" (drawing) strategy for Alice requires space that is linear in $n$.
They also showed a randomized strategy that draws with high probability ($1-{1 \over n}$) and requires $O(\sqrt n)$ bits of memory while relying on access to a secret random matching oracle.
Using a similar setting (with respect to the secret matching), Feige~\cite{Feige19} showed a randomized strategy for Alice that requires only $O(\log^3 n)$ bits of memory.
In fact, as we will show soon, Alice can supply her own matching while being {\em computationally efficient} and using $O(n \log n)$ bits of long lived randomness to produce reasonable response (or alternatively using cryptography and a small amount of long lived randomness while assuming computational limitation (poly time) on Bob).

{\em This stands in contrast to our impossibility result on the adaptive Dealer:}
In \Cref{section_upper_bound}, we have bounded the number of reasonable guesses by the Guesser's memory, regardless of run time, how much randomness she holds, and what cryptography she uses.
We present here a simplified computationally efficient version of Feige's construction that requires no access to a secret random matching but requires long lived randomness with random access for efficiency.

Our first point is that it is possible to construct a secret matching from more standard assumptions (long lived random bits or cryptographic ones).
Ristenpart and Yilek~\cite{RistenpartYilek2013} and Morris and Rogaway~\cite{MR2014} showed a transformation from (pseudo)random functions to (pseudo)random permutations that is secure {\em even if all the permutation is given to the distinguisher}.
Applying their constructions in our setting means using $O(n \log n)$ random bits (with random access) to construct a random permutation $\pi$ so that given $x$ it is possible to evaluate $\pi(x)$ on the fly (same for $\pi^{-1}(x)$), simply by looking at the randomness in $O(\log n)$ places and using $O(\log n)$ memory bits for intermediate calculations.
Naor and Reingold~\cite{NaorReingold2002} showed a construction that takes a permutation and its inverse and produces a permutation with any desired cycle structure\footnote{Naor and Reingold took a permutation $\pi$, and a permutation with the structure of choice $\sigma$ and returned $\pi^{-1} \circ \sigma \circ \pi$.}.
In particular, we can turn $\pi$ to an involution, i.e., a matching\footnote{For example, by taking $\sigma$ to be the involution $\sigma(1) = 2, \sigma(3) =4, \dots$.
The construction also gives each pair an ``index" in $[n]$ that is retrievable from either member of the pair.}.
Therefore combining these two we get that Alice can have a secret matching provided she has either (i)  $O(n \log n)$ secret random bits with random access or (ii) A key to a pseudorandom function and Bob is computationally limited and cannot distinguish the results from random\footnote{Existentially, this is is equivalent to one-way functions.}.

Having a secret matching with the above properties we discuss how to use the machinery of \Cref{sec_static_dealer:sec_random_subsets_guesser,sec_static_dealer:sec_random_subsets_guesser:sec_amplification} in order to suggest a strategy for Alice.
Alice will use her secret matching to imitate the mirror strategy of Bob, with an arbitrary starting point.
However, from time to time she will fail in that Bob will select as a response the matched card of the starting point, leaving her with no obvious response.
What she should do at this point is select a card that has not appeared yet as a new starting point.
For this, she needs the moral equivalent of a reasonable guess, and in this setting, any reasonable guess is good.
How many times do we expect this to happen?
This is similar to card guessing with perfect memory, i.e.\ $\ln n$ times.
Which means that she needs at least that many reasonable answers.

So in more detail, Alice uses her memory to allocate the $2n$ cards to subsets, similarly to what we did in \Cref{sec_static_dealer:sec_random_subsets_guesser,sec_static_dealer:sec_random_subsets_guesser:sec_amplification}.
She uses her long lived randomness to sample~$O(\log^2 n)$ permutations from a family of pairwise independent permutations and splits the functions to $\log n$ collections of equal size.
From each function in the $j$th collection, Alice produces a set of size $2^{j-1}$.
Resulting in a collection of sets of equal size.
In detail, for every function $\kwisefunc$ in the $j$th collection, Alice assigns the card $x$ to the subset~$\jthsubseth$ if $\kwisefunc(x) \in \{2^{j-1}+1, \dots, 2^j\}$.
For each subset, Alice tracks the number of cards that appeared from that subset and their sum, exactly as we did for our Randomized-Subset.

As her first card, Alice chooses an arbitrary card $g_1$, announces $g_1$ as her choice and stores it in her memory.
When Bob says a card $x$, Alice responds with $M(x)$, where $M$ is her secret matching.
She does that until Bob says $M(g_1)$.
Alice cannot say $M(M(g_1)) = g_1$, as this would result in her losing.
So instead, she attempts to recover a {\em reasonable response} from her memory.
If the recovering attempt succeeds and she retrieves a card $g_2$, then she announces $g_2$ as her choice and stores it in her memory in place of $g_1$.
She continues until Bob says $M(g_2)$ and so forth.

So we get that Alice loses the game only in case she fails to recover a reasonable card when Bob makes a correct guess.
For the first half of the game, Alice stores $O(\log n)$ singletons, cards that have not appeared (i.e.\ subsets of size $1$), and takes one of them that hasn't been declared yet.
When $\turnix \le n/2$ turns are left for Alice, she makes a recovery attempt from the $j$th collection of subsets for $j = \log \lfloor n /2 \turnix \rfloor$.
She checks the subsets, in some fixed order, until she finds one that yields a reasonable response.
For the sake of analysis, assume that Alice accesses each subset exactly once.
Afterward, that subset is deleted and ignored for the rest of the game, regardless of whether a recovery occurred.

Alice attempts to recover a reasonable guess whenever Bob guesses correctly.
There is a limited number of attempts she can make, and each attempt succeeds only with some probability.
So we first bound the probability that Bob makes too many correct guesses, then we bound the probability that Alice succeeds in making that many recovery attempts.

\newcommand{\mgcorrectguesses}{{2 \log n}}
\newcommand{\mgProbTooManyGuesses}{{n^{-2}}}
Let $x_\turnix$ be the indicator random variable for the event that Bob made a correct guess when $\turnix$ turns are left.
As Alice's matching is random, and since all of Bob's guesses are reasonable, then when Bob has $\turnix$ turns left, the probability that Bob guesses correctly is $1 \over {2 \turnix - 1}$.
For every $j \le n$, consider the span of $j/2$ turns that begins when Bob has $j$ turns left.
Since the matching is random, then any pair is independent of the remaining pairs.
That is, figuring a single match tells nothing about the remaining pairs, so we get that every $x_\turnix,$ is independent of all previous ones.
Let $x$ be the number of correct guesses made by Bob in the corresponding period, i.e., $x = \sum_{\turnix = j}^{j/2} x_\turnix$.
By linearity of expectation, the expected number of times that Bob guesses correctly is at most $$\mu = \expected[x] \approx 1/2 \ln j - 1/2 \ln j/2 = 1/2 \ln 2.$$
What is the probability that $x \ge \mgcorrectguesses$?
By a Chernoff bound (Theorem 4.4~(3) in~\cite{MitzenmacherUpfal2017}), the probability to make more than~$\mgcorrectguesses$ correct guesses, i.e., $4\log(n)$ times more than the expectation, is less then
$$
\Pr \left[x \ge \mgcorrectguesses \right]
\le
2^{-\mgcorrectguesses}
=
\mgProbTooManyGuesses
.
$$
So we get that Bob makes more than $\mgcorrectguesses$ correct guesses in a relatively small probability.

\newcommand{\mgProbInsufficientReasonableAnswers}{n^{-2}}
\newcommand{\mgfuncsnum}{{32 \log n}}

Alice may still lose if she fails to supply sufficiently many reasonable answers to Bob's successful guesses.
We will show that this also happens with small probability.
Recall that Alice samples functions independently and produces a subset from each function.
Assume that for every $j$ (and also for first half of the game), Alice samples $\mgfuncsnum$ functions.
Let $z_i$ be an indicator random variable for the event that the $i$th subset yields a reasonable answer.
From \Cref{claim_amplified_random_subset_yields_a_reasonable_guess}, we know that each subset yields a card that hasn't been played with probability of at least $1/4$.
We compare the probability that Alice runs out of functions to the process where we have $\{0,1\}$ random variables $y_1, y_2, \ldots y_{\mgfuncsnum}$ where each $y_i$ is independently chosen with probability exactly~$1/4$, and the probability of interest is that there are less than $\mgcorrectguesses$ 1's.
Take any configuration of the $z_1, \dots, z_\mgfuncsnum$ and take any configuration of the $y_1, \dots, y_\mgfuncsnum$ that is covered (or dominated) by the $z_i$'s configuration in the sense that $z_i = 0$ implies $y_i = 0$.
Then the probability for the configuration of the $y_i$'s is larger than that of the $z_i$'s, as the $y_i$ are at least as likely to yield 0's.
As a result, upper bounding the probability of less than $\mgcorrectguesses$ 1's of $y_i$s (more 0's) upper bounds the probability of less than $\mgcorrectguesses$ 1's of $z_i$s, and so, for Alice running out of reasonable answers.
Let $y=\sum_{i=1}^{\mgfuncsnum} y_i$.
By linearity of expectation
$$
	\mu = \expected[y] = 0.25 \cdot \mgfuncsnum = 8 \cdot \log n
	.
$$
By a Chernoff bound (Theorem 4.5 (2) in~\cite{MitzenmacherUpfal2017}), for $\delta = 3/4$, the probability for this event is at most
$$
\Pr\left[y < (1- \delta)\mu \right]
\le
e^{-{\delta^2 \mu \over 2}}
=
e^{-{9 \cdot 8 \cdot \log n\over 32}}
<
\mgProbInsufficientReasonableAnswers
.
$$

Considering the span of turns at which Alice queries the $j$th collection of sets.
We get that with probability at most $\mgProbTooManyGuesses$ Bob makes too many correct guesses during that period, and with probability at most $(1 - \mgProbInsufficientReasonableAnswers)\mgProbTooManyGuesses$ Alice fails to produce sufficiently many reasonable responses at that period.
It follows that Alice loses while she considers a specific $j$ with probability at most $\mgProbTooManyGuesses + (1 - \mgProbTooManyGuesses)\mgProbInsufficientReasonableAnswers < 2n^{-2}$.
By the union bound no failure occurred during any of the $\log n$ periods with probability at most $2\log n / n^4$, i.e., Alice draws (or wins) with probability at least $$1 - {2 \log n \over n^2} > 1 - {1 \over n}.$$

We conclude that with $O(\log^3 n)$ bits of memory and $O(\log^3 n + n \log n)$ bits of long lasting randomness,  Alice has a  strategy against  Bob, that draws or wins with probability at least $1 - {1 \over n}$.
If we wish to go use computational assumptions, then Alice needs only to store a key to a pseudo-random function (and assume that Bob is  computationally bounded, i.e. cannot distinguish between the PRF and a truly random function) and we get the desired result.

\begin{question}
	Does there exist an algorithm for Alice that requires at most polylog long lived bits of randomness and no cryptographic assumptions and gives her a reasonable chance of not losing?
\end{question}
\begin{question}
	Is the $\Omega(\log^2 n)$ lower bound of \Cref{sec_static_dealer:upper_bound} relevant for this setting as well?
\end{question}

%% file: parts/discussion/card_guessing_variants.tex
Consider a the Card Guessing game with a deck that contains $c$ copies of each card. 
Diaconis and Graham~\cite{DG1981} showed that the optimal strategy against a Random Static Dealer scores $c + \theta(\sqrt{c \log n}) $ correct guesses in expectation.
This is achieved by tracking all cards that appeared so far and guess the one with the highest probability, which can be easily done with $n \log c$ memory bits.
By tracking the deviation from the expected number of cards seen so far, it may be possible, in some cases, to achieve a slightly better memory consumption for the static case.

Note that the low-memory Guessers from \Cref{sec_static_dealer} may still work if we assign all copies of a card to the same subset, however, the performance of these Guessers remains $O(\ln n)$.
For $c > \ln n$ these techniques are far from optimal.
In fact, with no memory at all, it is possible to get $c$ correct guesses simply by repeating the same guess over and over again.
So we get that for $c \ge \ln n$ our low memory Guessers perform worse than a Guesser with no memory at all.

\begin{question}
	How much memory and randomness does a Guesser requires to score ``near-optimal'' results against a Dealer with a deck containing multiple copies of each card?
	Is there a difference between the different kinds of Dealers?
\end{question}

As for other variants of Card Guessing, the literature considers Card Guessing with partial feedback (was the guess correct or not) and no feedback at all.
The optimal~\cite{DGHS2020} and near-optimal~\cite{DGS2020} guessing strategies for partial feedback requires little to no memory, and so goes for the optimal strategy for no feedback~\cite{DG1981}.
We suggest the study of a general theory of when we can convert a feedback type into a low memory guessing.

%% file: parts/acknowledgments.tex
We thank Eylon Yogev and Yotam Dikstein for many suggestions and advice. 
We thank Hila Dahari, Uri Feige, Tomer Grossman, and Adi Schindler for meaningful discussions and insights. 
We thank Samuel Spiro for his comments.
We also thank Gal Vinograd for reading a preliminary version of this document.
